\documentclass[a4paper]{article}

\usepackage[english]{babel}
\usepackage[utf8]{inputenc}
\usepackage{amsmath}
\usepackage{amsthm}
\usepackage{amsfonts}
\usepackage{subfigure}
\usepackage{xcolor}
\usepackage{mdframed}

\usepackage{graphicx,subfig}
\usepackage[colorinlistoftodos]{todonotes}
\usepackage{geometry}
\usepackage{multirow}
\usepackage{float}

\usepackage{apacite}
\usepackage{bbm}
\usepackage{natbib}
\usepackage{framed}
\usepackage{comment}

\usepackage{booktabs}

\usepackage{enumitem}
\usepackage{setspace}

\usepackage{algorithm}
\usepackage{algorithmic}
\usepackage{appendix}

\usepackage{xcolor}
\usepackage{listings}
\newtheorem{definition}{Definition}[section]
\newtheorem{proposition}{Proposition}[section]

\newtheorem{remark}{Remark}

\def\tb{{\boldsymbol{\theta}}}

\def\zerob{{\boldsymbol{0}}}

\def\mub{{\boldsymbol{\mu}}}

\def\oneb{{\boldsymbol{1}}}

\def\rmd{\,{\rm d}}

\def\xb{{\mathbf{x}}}

\def\Xb{{\mathbf{X}}}

\def\ub{{\mathbf{u}}}
\def\Ub{{\mathbf{U}}}

\def\mb{{\mathbf{m}}}

\def\rbb{{\mathbb R}}
\def\pbb{{\mathbb P}}

\def\onebb{{\mathbbm 1}}

\DeclareMathOperator{\argmax}{argmax}

\DeclareMathOperator{\CM}{CM}

\def\VC{{\cal V}}
\def\TC{{\cal T}}

\def\NC{{\cal N}}
\def\EC{{\cal E}}

\usepackage{authblk}

\title{Reverse stress testing via multivariate modeling with vine copulas}
\author{Menglin Zhou$\dag$ and Natalia Nolde$\dag$}
\affil{$\dag$Department of Statistics, University of British Columbia, Canada}
\date{}

\begin{document}

\maketitle
\setcounter{page}{1}

\begin{abstract}
As an important tool in financial risk management, stress testing aims to evaluate the stability of
financial portfolios under some potential large shocks from extreme yet plausible scenarios of risk factors. The effectiveness of a stress test crucially depends on the choice of stress scenarios. In this paper we consider a pragmatic approach to stress scenario estimation that aims to address several practical challenges in the context of real life financial portfolios of currencies from a bank. Our method utilizes a flexible multivariate modelling framework based on vine copulas.

\end{abstract}

\setcounter{equation}{0}

\section{Introduction}

Stress testing aims to evaluate the resilience of a financial system under some extreme yet plausible scenarios for its risk factors. Such scenarios are referred to as stress scenarios. It is argued that effectiveness of a stress test heavily relies on the choice of stress scenarios (see, e.g., \cite{Breuer_etal2009}, \cite{BreuerCsiszar2013}, \cite{Glasserman_etal2015}). However, selecting scenarios that are both severe and plausible is a significant challenge. Reserve stress testing involves estimation of stress scenarios given that the system experiences a certain adverse outcome. This is in contrast to traditional stress testing which considers impact of given stress scenarios on the financial system. The importance of reverse stress testing (RST) has been highlighted in the aftermath of the 2017-2019 financial crisis by supervisory authorities, such as \cite{Basel2009}, \cite{Committee2009} and \cite{FSA2009}.

Given a specified unfavorable portfolio outcome, the purpose of RST of a financial portfolio is to uncover stress scenarios for risk factors which result in that outcome. In the literature, different definitions of the adverse outcome are considered. \cite{Grundke2011} and \cite{GrundkePliszka2018} use the economic capital requirement to define the adverse effect. However, this definition is not easy to use in the context of statistical estimation of stress scenarios. Alternatively, the adverse effect can be defined as the loss of a portfolio exceeding (or equaling) some given threshold value (see, e.g., \cite{McneilSmith2012}, \cite{Glasserman_etal2015}, \cite{Kopeliovich_etal2015}). In this paper, we adopt the definition of stress scenarios as given in \cite{Glasserman_etal2015}. Let $\Xb$ be a $d$-dimensional random vector representing changes in risk factors and $L = g(\Xb)$ denote the corresponding portfolio loss for some real-valued function $g$ on $\rbb^d$. Then, for a large threshold $\ell>0$, a stress scenario is defined as:
\begin{equation}
    \mb(\ell) = \argmax_{\xb\in \rbb^d}f(\xb|L\geq \ell),
    \label{eq:RST_L}
\end{equation}
where $f(\cdot\mid L\geq \ell)$ denotes the conditional density of $\Xb$ given $L\geq \ell$, assumed to exist. That is, a stress scenario is the maximizer of the density of risk factors conditional on the portfolio loss to be above the given threshold. Statistically, it is just the conditional mode of risk factors, thus providing a most likely scenario for the risk factors while ensuring that the associated portfolio loss exceeds the threshold.

When it comes to modelling the distribution of risk factors and the associated portfolio loss, the assumption of either joint normality or ellipticality of risk factors along with linearity of the function $g(\cdot)$ are common; see, e.g., \cite{Grundke2011}, \cite{Kopeliovich_etal2015} and \cite{Glasserman_etal2015}. In high-dimensional settings, principle component analysis is often used for dimension reduction as in \cite{GrundkePliszka2018}. A framework which allows for non-elliptical distribution of risk factors is proposed by \cite{McneilSmith2012}, based on the concept of half-space depth. \cite{Pesenti2019} propose a flexible data-driven framework for stress scenario estimation based on a large Monte Carlo sample of risk factors and corresponding portfolio losses. 


In this paper we aim to fill some gaps in existing literature on stress scenario estimation by focusing specifically on practical challenges associated with the problem of reverse stress testing. Below is a list of various aspects of this problem, which we formulated based on the data from real life financial portfolios of currencies. While we list them as separate points, they are naturally interconnected. 

\begin{itemize}

\item Limited data: The data available for estimating stress scenarios includes historical time series of risk factors which may be of only moderate length and hence may not contain extreme values of portfolio losses that would be of interest for a financial institution. In this setting, the estimation procedure should allow for extrapolation outside the available data range to be of practical relevance. Hence, there is a need for either a fully parametric model or a semi-parametric approach that allows for extrapolation, for example, using techniques from extreme value theory.

\item Dimensionality: While sub-portfolios may only be determined by a few risk factors, their number can be considerably larger for middle level portfolios. Hence, estimation methods need to be computationally feasible for large portfolios.

\item Flexibility of modelling assumptions to capture stochastic properties of data: This requirement needs to be applied to both marginal behaviour and joint dependence structure of risk factors. The need for a flexible model is particularly important in the case of larger portfolios where risk factors may exhibit a variety of stochastic properties, including tail dependence but also tail independence among the risk factors as well as jointly with portfolio losses. Commonly used multivariate Gaussian or elliptical distributions are unlikely to provide an adequate fit to the data in high dimensions.

\item Multiple regimes in the data: Some datasets exhibit what can be described as the presence of multiple regimes each of which has different dependence characteristics among risk factors and portfolio losses. One example of such data is presented in Figure~\ref{fig:T_clust} showing scatter plots of risk factors against portfolio losses for a portfolio of five currencies. In these situations, the conditional density of risk factors given that the portfolio loss exceeds a high threshold may not be unimodal. This makes the definition of the stress scenario in~\eqref{eq:RST_L} not suitable as it would pick only the global maximum in the conditional density as the stress scenario but will ignore local maxima corresponding to stress scenarios in other regimes. One approach is to attempt to model these more complex data and then modify the definition of stress scenarios to capture potential multimodality. However, as a simpler option, we propose to split the data into different regimes using clustering and then apply a stress scenario estimation procedure within each cluster separately. This approach provides a solution that is a reasonable compromise in view of modelling and estimation complexities in a multi-regime setting. 

\item Lack of an explicit form for function $g$ linking risk factors to portfolio losses: It is common to assume that $g$ is a linear function or at least can be well approximated by a linear function. However, this is not always the case. The function $g$, while being deterministic, is not usually given explicitly. Portfolios are evaluated internally and only portfolio losses corresponding to historical values of risk factors are available for modelling and estimation purposes. This fact needs to be incorporated into the modelling framework.

\end{itemize}

Considering the practical challenges listed above, we propose a flexible modelling approach, which makes use of  vine copulas (\cite{Joe1996}, \cite{BedfordCooke2001}, \cite{BedfordCooke2002}, \cite{Czado2019}). Vine copulas offer a useful tool for constructing multivariate distributions in a high-dimensional setting that combines vine graphs and bivariate copulas. The need for only bivariate copula families is particularly advantageous as there exist a great variety of bivariate copulas, including elliptical copulas, Archimedean copulas (\cite{Joe1997}), extreme-value copulas (\cite{GudendorfSergers2010}) and nonparametric copulas (\cite{NaglerCzado2016}). In comparison,  the choice of multivariate copulas is much more limited and in general it would be hard to find a single multivariate copula that is able to capture the dependence structure for a high-dimensional data. Furthermore, modelling dependence with copulas provides flexibility in handling univariate marginal components separately from their dependence.

The rest of the paper is organized as follows. Section~\ref{sec::data} introduces three motivating datasets, which informed the list of challenges mentioned above and which will subsequently be used for numerical illustrations.
Section~\ref{sec::vine} reviews relevant background on vine copulas and their construction. In Section~\ref{sec::est}, we present three methods to estimate stress scenarios defined in~\eqref{eq:RST_L}. We begin with a naive approach, which relies on an estimate of the joint density of risk factors~$\Xb$ and the associated portfolio loss~$L$. This is followed by two modifications, one involving explicit modelling of function $g$ and the other exploring the connection between conditional densities of $\Xb$ given $L\ge\ell$ versus $L=\ell$. Simulation studies to assess performance of the different estimation methods are summarized in Section~\ref{sec::sim}. Section~\ref{sec::app} illustrates the methods when applied to real data. Conclusions are given in Section~\ref{sec::con}.

%

\section{Data description}
\label{sec::data}
In this section we describe three real life portfolios of currencies that will be used to set up a realistic data generating process in the simulation studies as well as to illustrate the proposed methodology in the application.

The data consists of daily time series of currency exchange rates (against the US dollar) for the period from August 1, 2008 to June 29, 2020, resulting in 3078 observations for each series. As risk factors, we consider daily relative changes in exchange rates. The series of the associated portfolio values is obtained using a bank’s internal valuation model. For the purpose of modelling, we consider day-to-day portfolio losses, $L_t = -(V_{t} - V_{t-1})$, where $V_t$ denotes the value of the portfolio at time $t$. These values have been scaled for data anonymization. 

The univariate marginal components of the risk factors and portfolio losses are modelled using either a skew-t distribution introduced by \cite{JonesFaddy2003}, or semi-parametrically with a kernel density estimator for the central part of the data and the generalized Pareto distribution for the upper and lower tails. The skew-t distribution provides a flexible model able to capture both heavy-tailed behaviour and asymmetry often observed in financial data. The use of the generalized Pareto distribution is justified by extreme value theory for threshold exceedances. Details on the estimation of univariate margins are provided in~\ref{appen::margin}. To assess dependence between portfolio losses and each of the portfolio risk factors we provide dependence summary statistics including the Kendall's tau coefficient and the upper tail dependence coefficient. The latter is  estimated using the method in~\cite{Lee_etal2018}.




\subsection{Portfolio A}\label{spA}
Portfolio A consists of four currencies: Canadian dollar (CAD), British pound (GBP), Australian dollar (AUD) and Hong Kong dollar (HKD). Table~\ref{tab:A_mar} gives estimates of the parameters of the marginal distributions. Based on the quantile-quantile (Q-Q) plots (not included in the paper), we found the skew-t distribution to provide a good fit for the portfolio losses and risk factors based on CAD, GBP and AUD currencies. Relative changes in exchange rate for HKD were better captured via the semi-parametric approach with the generalized Pareto distribution fitted to the tails. It should be noted that larger values of shape parameters $\alpha$ and $\beta$ of the skew-t distribution or the reciprocal of the shape parameter of the generalized Pareto distribution indicate lighter tails. The marginal estimates suggest that HKD data have rather different decay rates in their upper and lower tails in contrast to all the other variables.

\begin{table}[H]
\centering
\small
\caption{Marginal estimates for Portfolio A. The columns correspond to portfolio losses (L) and risk factors given by relative changes in currency exchange rates. Parameters $\alpha$ and $\beta$ are for the fitted skew-t distribution; ${\xi}_L$ and ${\xi}_U$ are shape parameters of the generalized Pareto distribution fitted to the lower and upper tails, respectively.}
\begin{tabular}{cccccc}
\hline
     & L & CAD    & GBP    & AUD    & HKD    \\ \hline
$\hat{\alpha}$  & 1.851 & 1.984 & 2.153 & 1.698 & - \\
$\hat{\beta}$  & 1.878 & 1.922 & 2.087 & 1.560 & - \\ 
$1/\hat{\xi}_L$ & -  & -  & -  & -  & 2.825\\
$1/\hat{\xi}_R$ & - & -  & -  & - & 4.274\\ \hline
\hline
\end{tabular}
\label{tab:A_mar}
\end{table}

Figure~\ref{fig:A_sc} displays scatter plots of normal scores for portfolio losses and each of the four risk factors. It can be observed that, for this portfolio, risk factors tend to be negatively associated with portfolio losses. The strength of dependence is largest for the Canadian dollar, nearly approaching counter-monotonicity. The British pound and Australian dollar exhibit intermediate levels of dependence, while the Hong Kong dollar displays only weak association with portfolio losses. This can be further quantified in term of the Kendall's tau correlation coefficient reported in Table~\ref{tab:A_dep}, as well as estimates of the upper tail dependence coefficient for negated values of the risk factors. The estimates of the tail dependence coefficient suggest presence of tail dependence for portfolio losses and risk factors based on CAD, GBP and AUD currencies.  

%
%

\begin{table}[H]
\centering
\small
\caption{Kendall's tau correlation coefficient ($\tau$) and estimated upper tail dependence coefficient ($\hat\lambda$) between risk factors and portfolio losses for Portfolio A. Risk factor values were negated in computation of tail dependence coefficient estimates.}
\begin{tabular}{ccccc}
\hline
     & CAD    & GBP    & AUD    & HKD    \\ \hline
$\tau$  & -0.963 & -0.302 & -0.475 & -0.095 \\
$\hat{\lambda}$ & 0.850  & 0.264  & 0.413  & 0.000  \\
\hline
\hline
\end{tabular}
\label{tab:A_dep}
\end{table}

\begin{figure}[ht]
\centering
\subfigure{\includegraphics[scale=0.22]{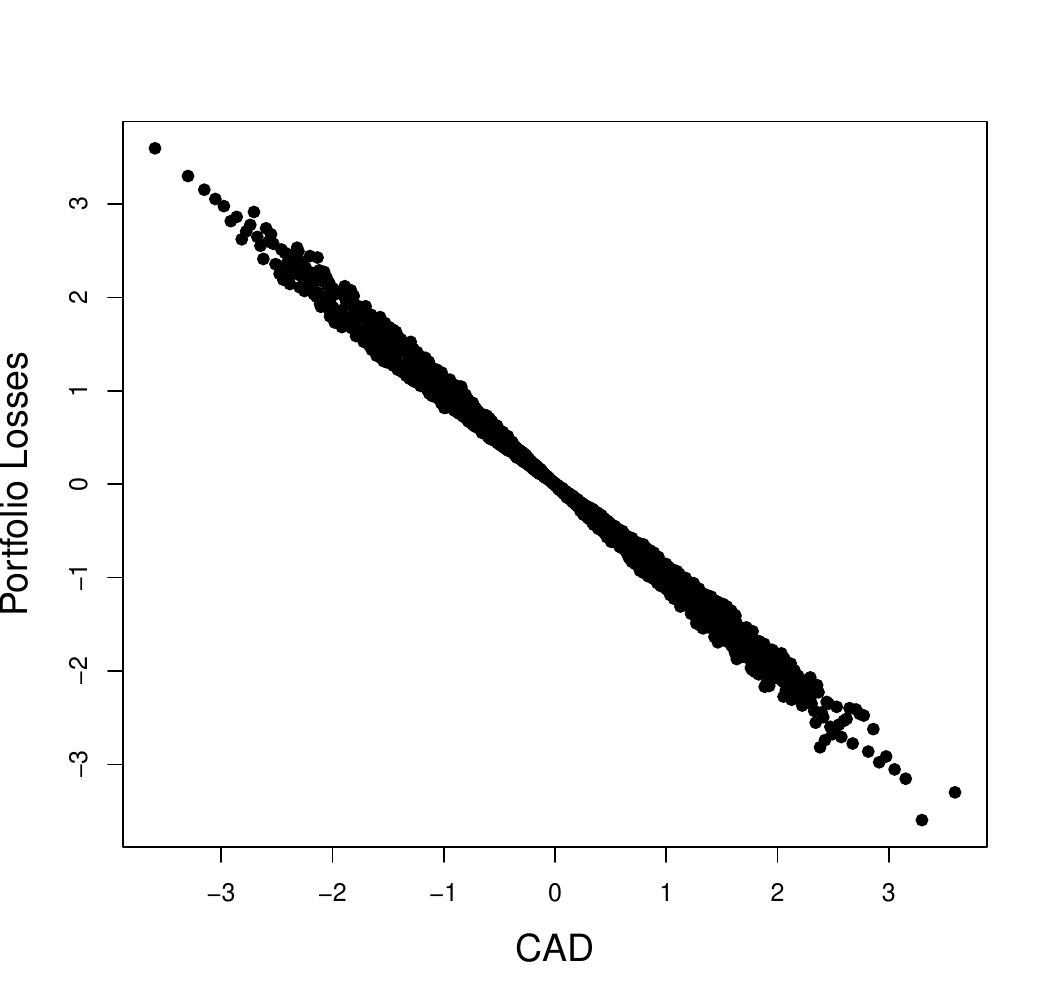}}
\subfigure{\includegraphics[scale=0.22]{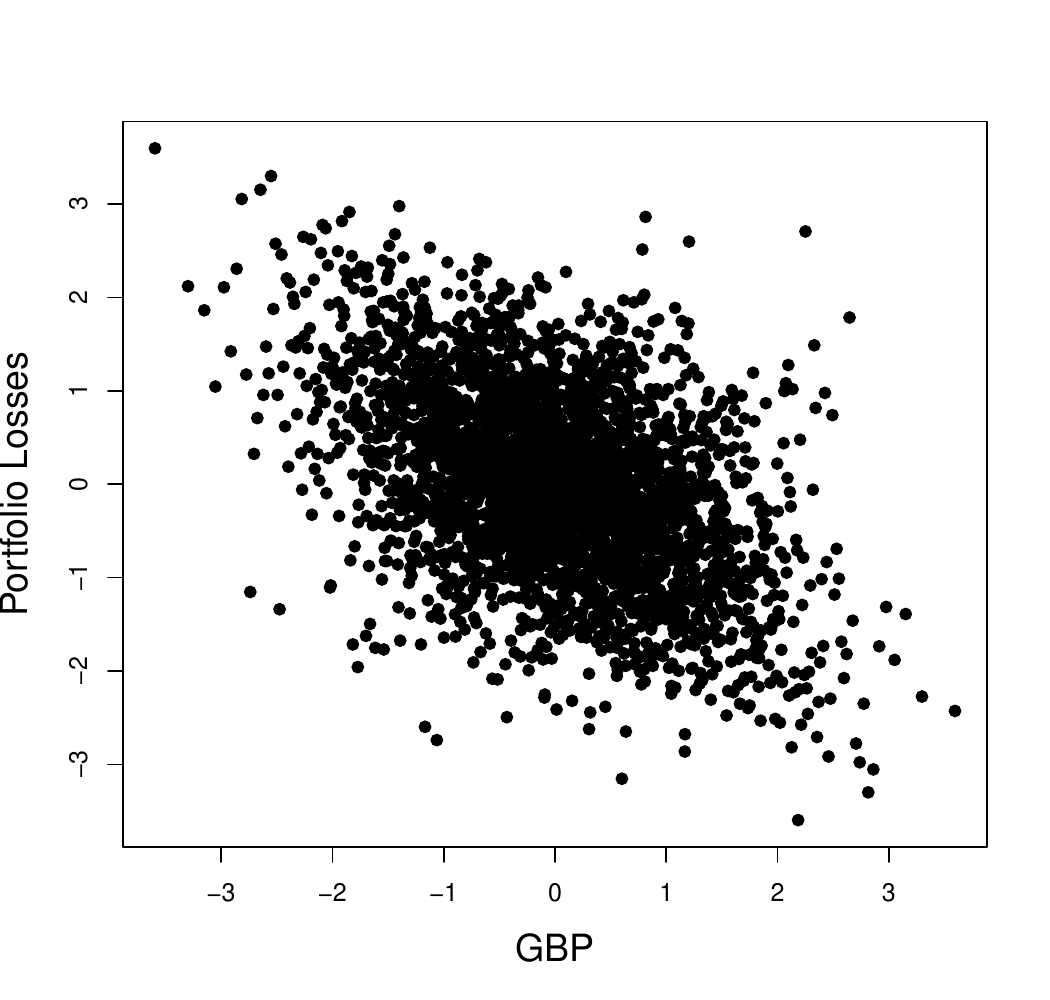}}
\subfigure{\includegraphics[scale=0.22]{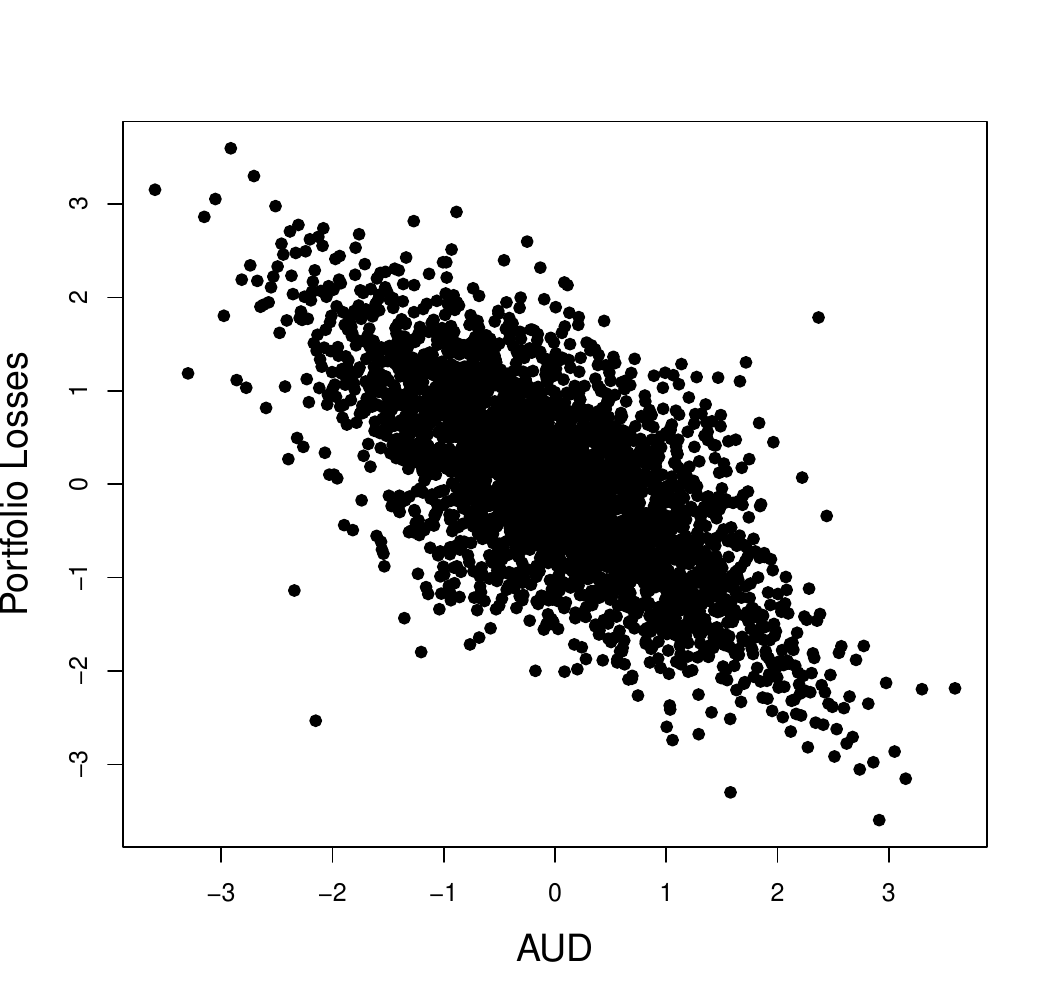}}
\subfigure{\includegraphics[scale=0.22]{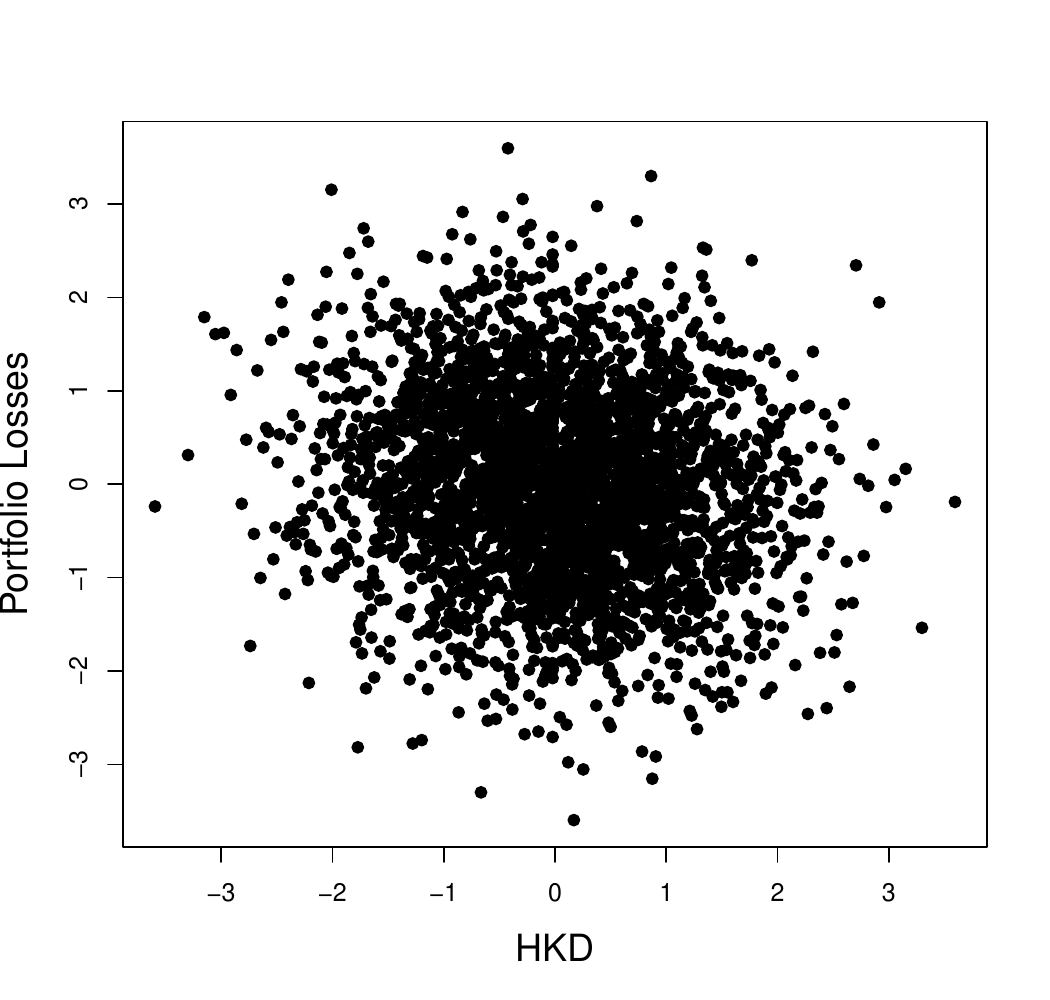}}
\caption{Scatter plots of normal scores between portfolio losses and relative changes in currency exchange rates for portfolio A.}
\label{fig:A_sc}
\end{figure}



\subsection{Portfolio B}\label{siB}

Portfolio B consists of five currencies: CAD, Japanese yen (JPY), AUD, Mexican peso (MXN) and EUR. The scatter plots of normal scores between daily relative changes in currency exchange rates and portfolio losses are shown in Figure~\ref{fig:T_sc}. It is apparent from these plots that dependence characteristics between risk factors and portfolio losses are more complex in portfolio~B than were previously seen in portfolio~A. As before, the strength of dependence ranges from weak in the case of JPY to very strong for CAD. However, a salient feature for these data is the lack of monotonicity in the dependence structure. In particular, here large portfolio losses are sometimes driven by large values in risk factors based on the CAD and AUD currencies, and other times by their small values. One possible explanation for this phenomenon is the presence of multiple regimes in the data. To explore this idea we consider clustering the data into two regimes,  where in one regime risk factors exhibit positive association with portfolio losses and in the other regime negative association.

\begin{figure}[ht]
\centering
\subfigure{\includegraphics[scale=0.17]{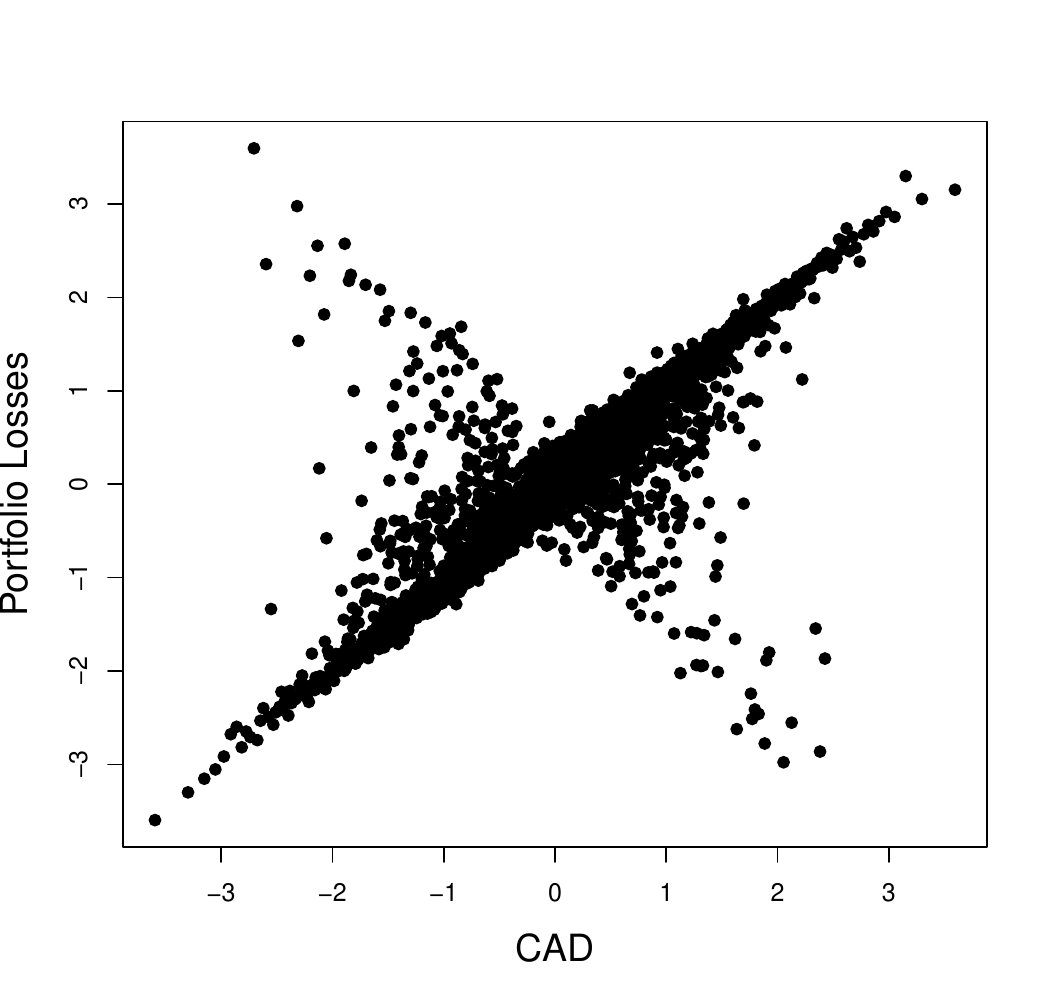}}
\subfigure{\includegraphics[scale=0.17]{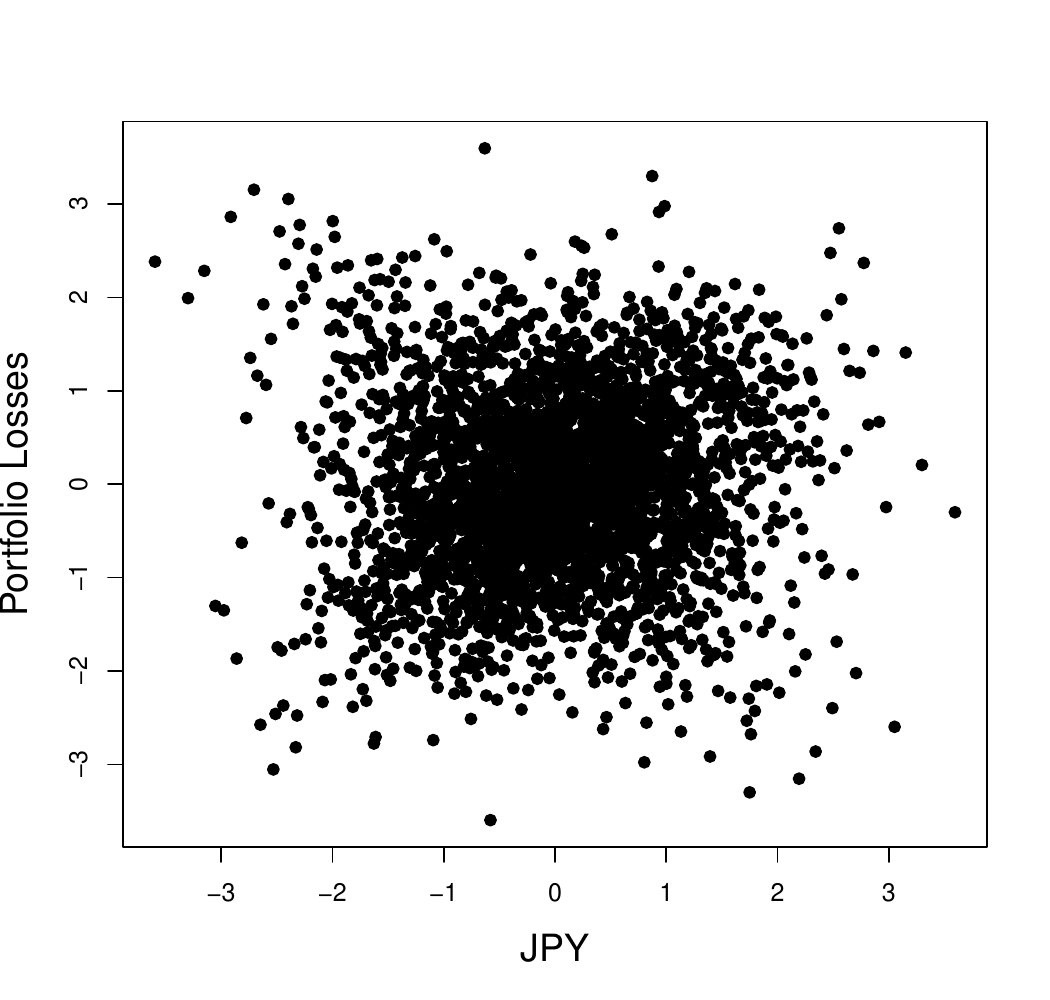}}
\subfigure{\includegraphics[scale=0.17]{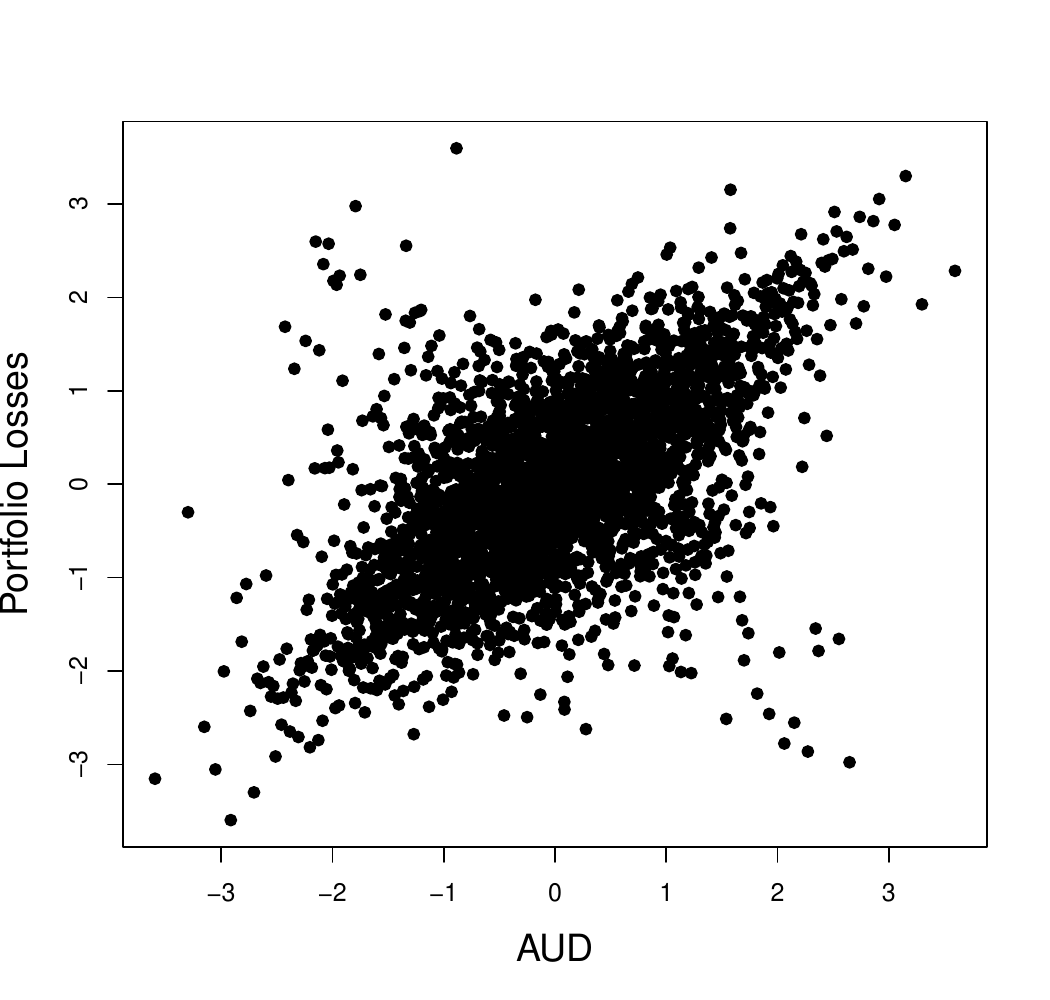}}
\subfigure{\includegraphics[scale=0.17]{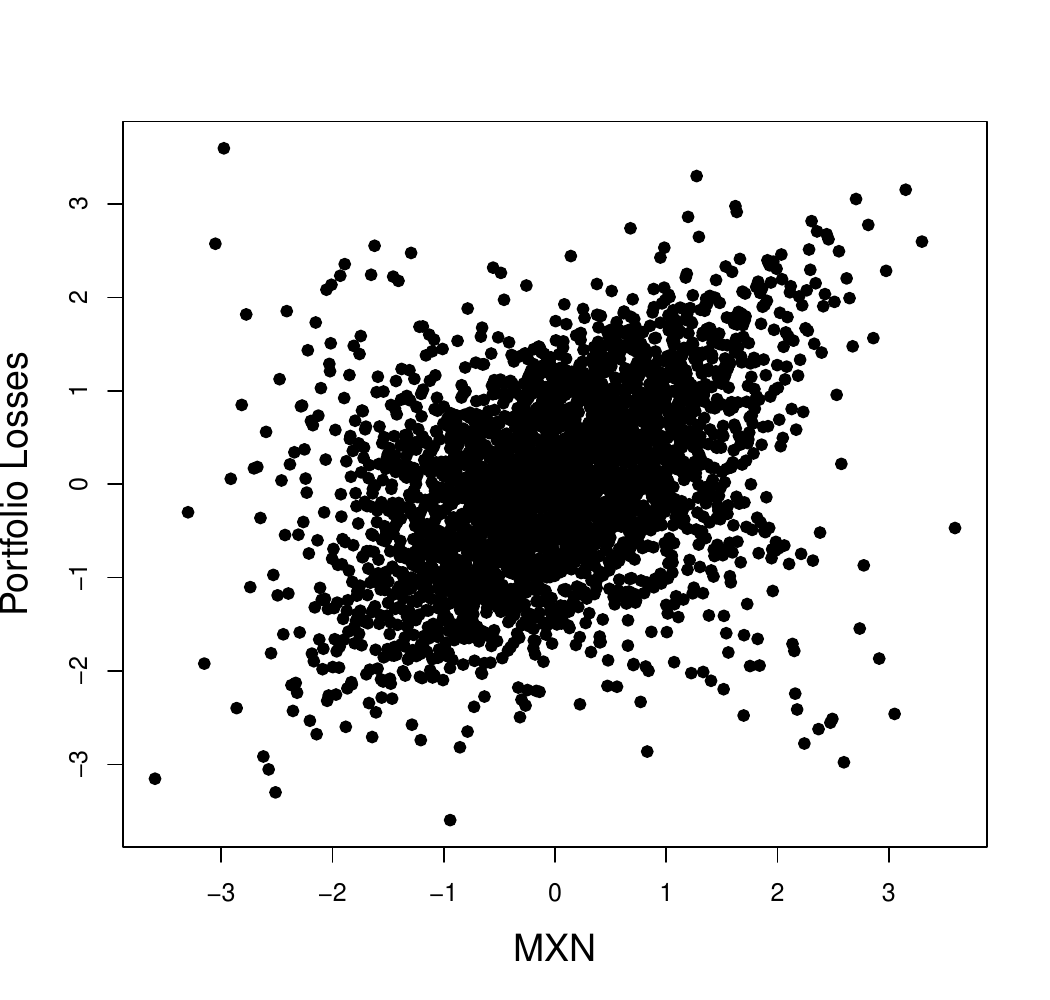}}
\subfigure{\includegraphics[scale=0.17]{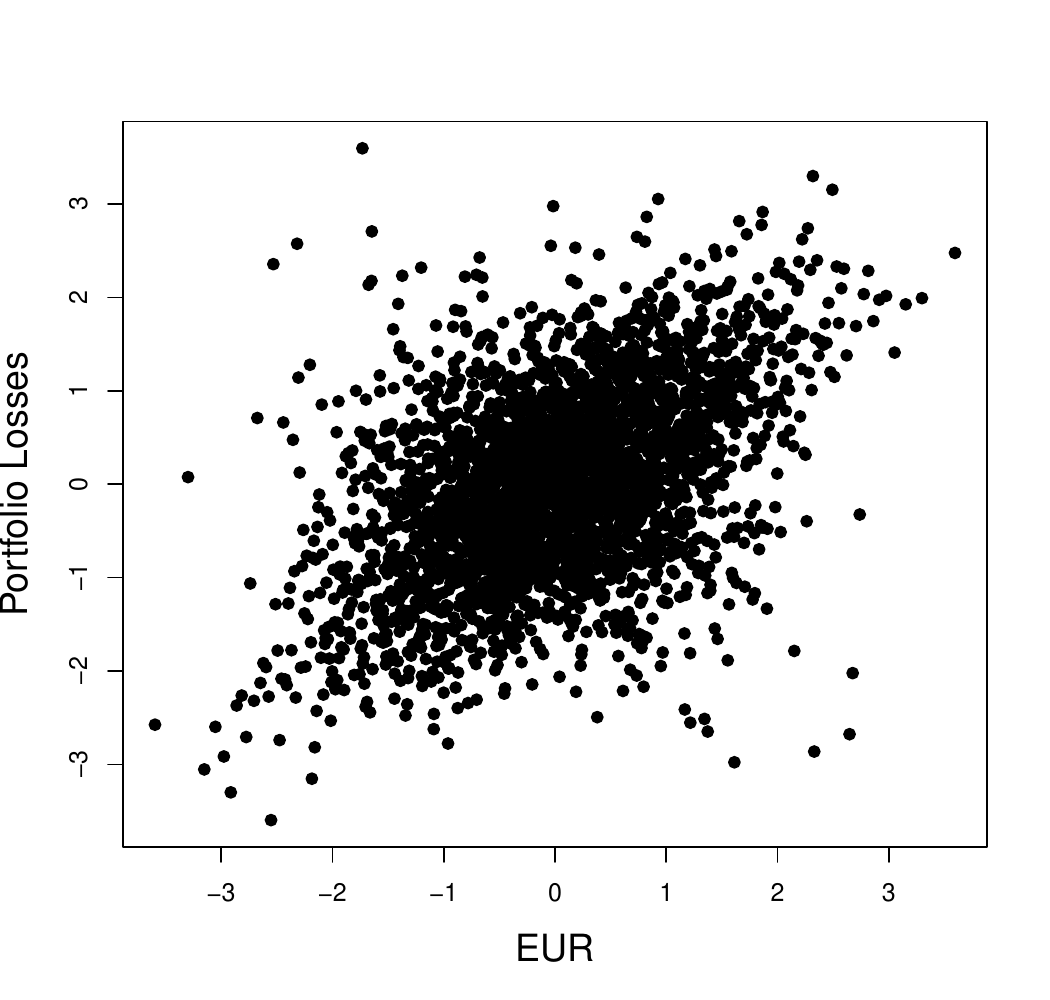}}
\caption{Scatter plots of normal scores between portfolio losses and relative changes in currency exchange rates for portfolio B.}
\label{fig:T_sc}
\end{figure}

Distance-based clustering algorithms are inapplicable here, as we aim to cluster observations according to dependence characteristics. We choose to do clustering via a mixture model, where one of the mixture components exhibits positive dependence while the other exhibits negative dependence. Considering the difficulty of finding a  multivariate parametric model to capture a 5-dimensional dependence structure, it is easier to use only 2-dimensional data to construct the mixture model. In this example, using only CAD and the portfolio losses may be able to generate reasonable clusters. However, when there are more than one risk factor showing multiple regimes, it is not a good idea to cluster with a single risk factor. Under this consideration, we aggregate the information of all risk factors into a single component using principle component analysis (PCA). To make the procedure more general, we adopt kernel PCA (\cite{Scholkopf_etal1998}) rather than the traditional linear PCA. However, in this example, both linear and kernel PCA techniques lead to similar results for these data. 

Once we have the first principle component from the kernel PCA, we transform it as well as the portfolio losses into uniform scale based on the empirical distribution and then construct a mixture of two bivariate t copulas, which is expressed as
$$
\Ub = 
\left\{
\begin{array}{lr}
        \Ub^{(1)}\sim tCop_2(\rho_1, \nu_1), &\text{if } W = 0;\\
        \Ub^{(2)}\sim tCop_2(\rho_2, \nu_2), &\text{if } W = 1,
\end{array}
\right.
$$
where $tCop_2(\rho, \nu)$ denotes the bivariate t copula with $\nu$ degrees of freedom and scale matrix  whose off-diagonal elements are $\rho$ and ones on the diagonal;  $W$ is the latent variable following Bernoulli distribution with success probability $\pi$. The parameters $\{\pi, \rho_1, \rho_2, \nu_1, \nu_2\}$ are estimated via the expectation–maximization (EM) algorithm. We finally cluster the observations by calculating the posterior probabilities $\pbb(W = k|\Ub = \ub_i)$ for $k\in \{0,1\}$, where $\{\ub_1,\ldots,\ub_n\}$ denotes the sample of transformed observations. The resulting two clusters are shown in Figure~\ref{fig:T_clust}.

\begin{figure}[ht]
\centering
\subfigure{\includegraphics[scale=0.25]{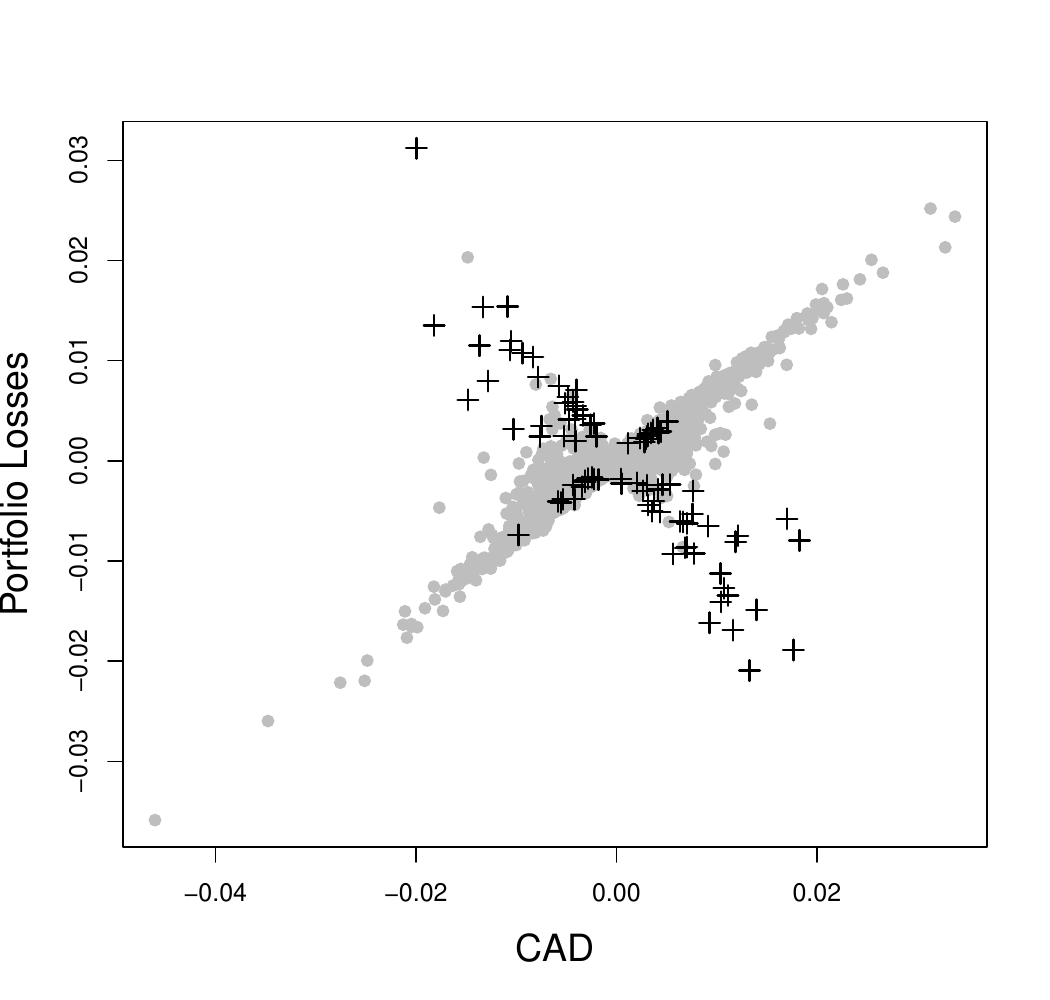}}
\subfigure{\includegraphics[scale=0.25]{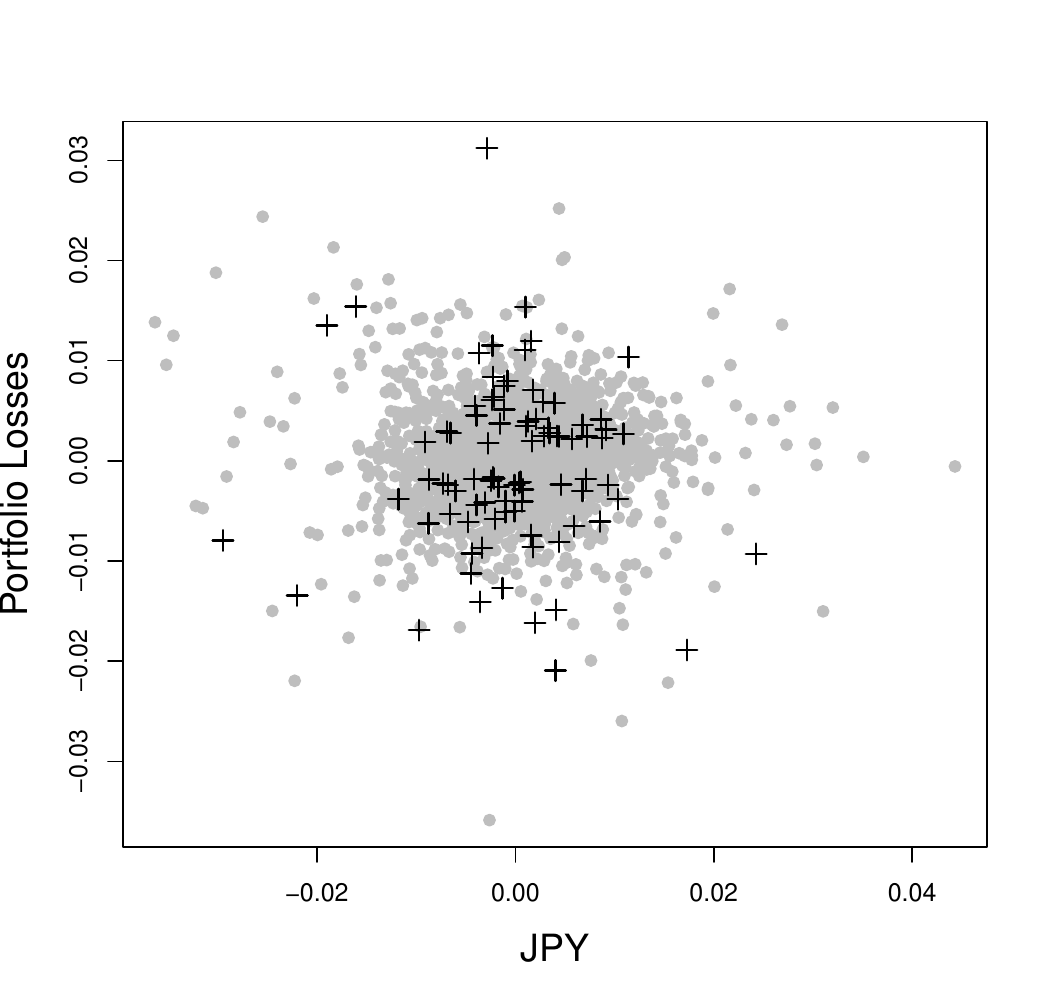}}
\subfigure{\includegraphics[scale=0.25]{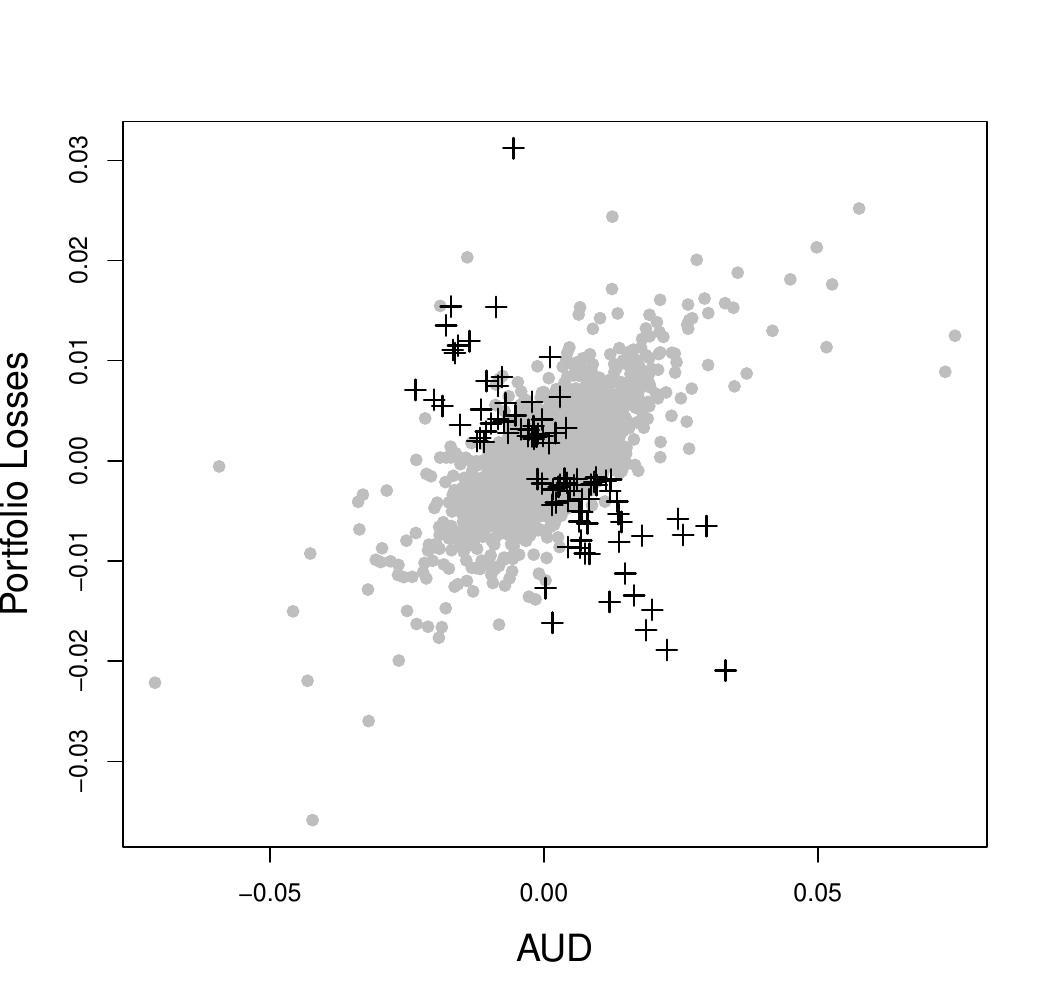}}
\subfigure{\includegraphics[scale=0.25]{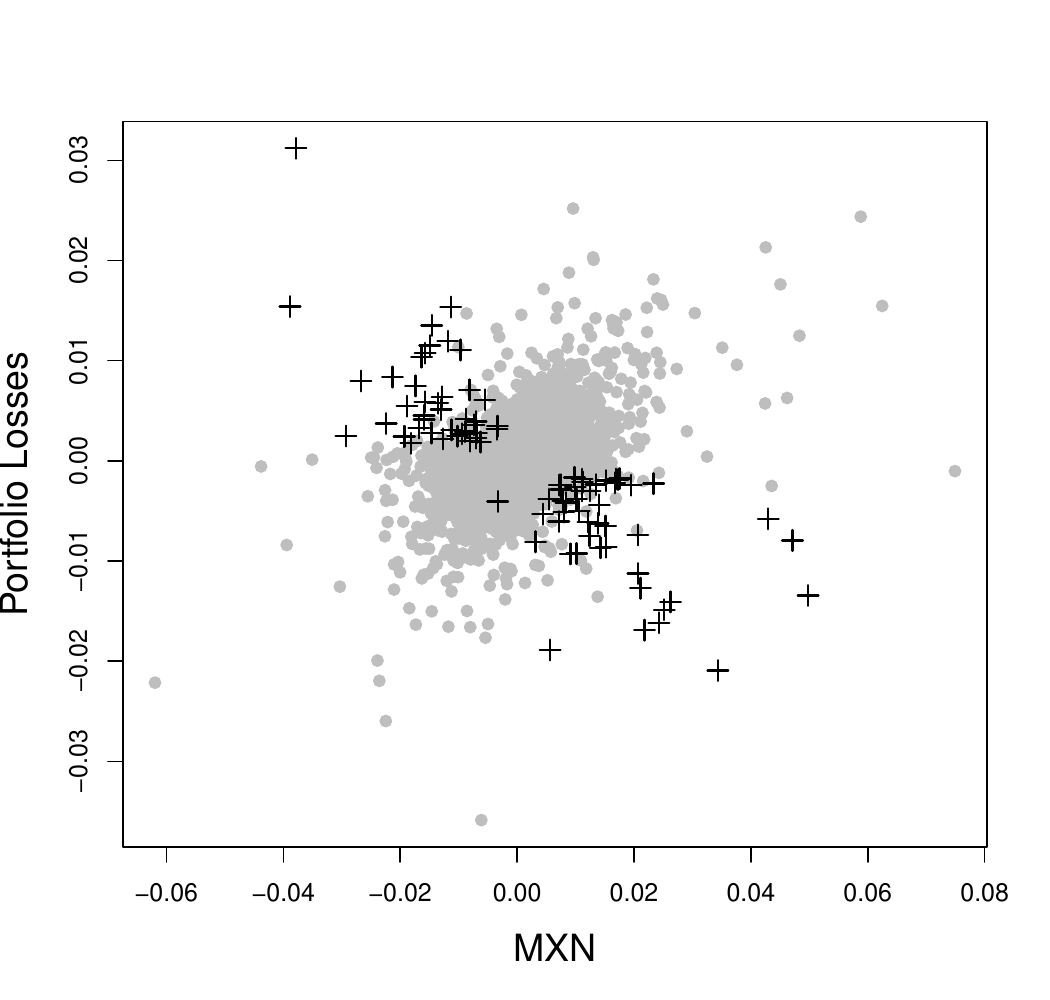}}
\subfigure{\includegraphics[scale=0.25]{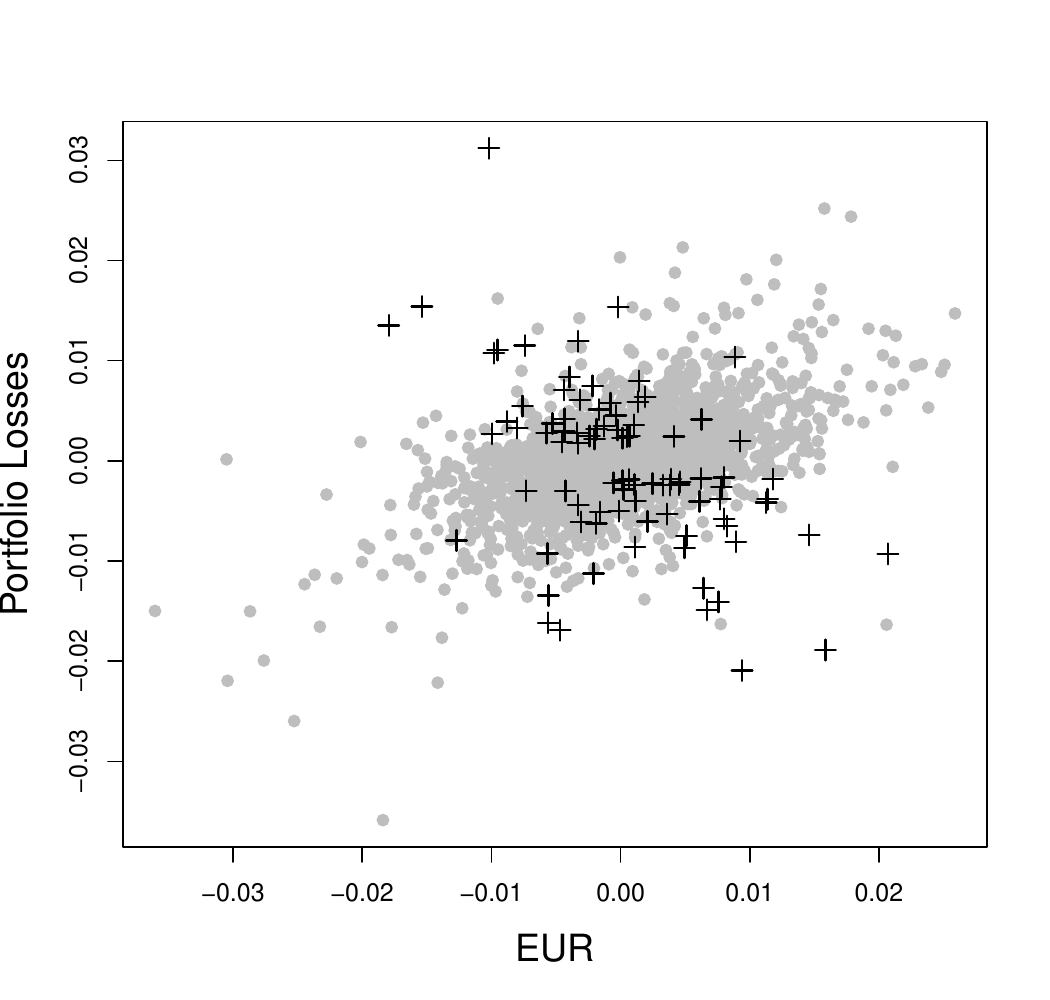}}
\caption{Scatter plots of scaled daily portfolio losses against daily relative changes in currency exchange rates for portfolio B. The round dots are cluster 1 and the cross dots are cluster 2.}
\label{fig:T_clust}
\end{figure}

Figure~\ref{fig:T_clust_ts} shows the time series of portfolio losses with different colours used to differentiate the two regimes. It is worth noting that cluster/regime 2 falls mostly on the dates in the period from December 2015 to February 2016  and then from March 2020 until June 2020, the end of the observation period. From the historical context, both of these time periods correspond to recessions, the mini-recession of 2015-2016 and the start of the COVID-19 pandemic.


\begin{figure}[ht]
\centering
\includegraphics[scale=0.5]{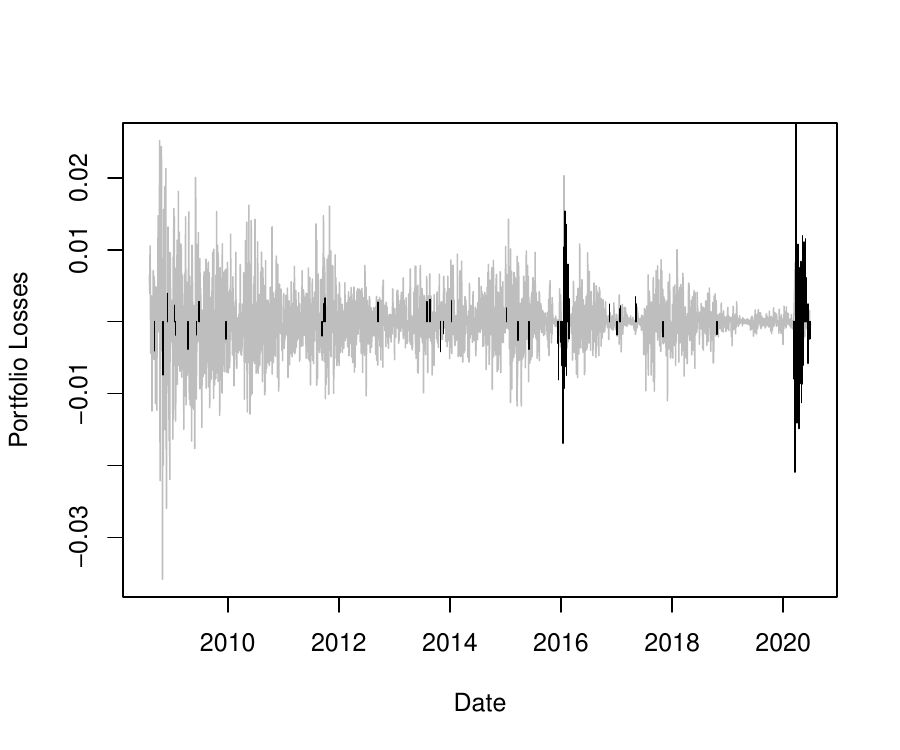}
\caption{Time series plot of the scaled daily portfolio losses for portfolio~B. The grey and black lines indicate time points corresponding to cluster 1 and cluster 2, respectively.}
\label{fig:T_clust_ts}
\end{figure}

Table~\ref{tab:T_dep} reports values of the dependence measures between risk factors and portfolio losses for each of the two clusters. As can be seen from the values of the Kendall's tau correlation coefficient and the estimates of the upper tail dependence coefficient, cluster 1 is characterized by strong positive association and tail dependence between portfolio losses and four of the five risk factors. In contrast, cluster~2 exhibits negative association, and large portfolio losses are tail dependent with negative but large in absolute value relative changes in exchange rates for CAD, AUD, MXN and EUR. 

Marginal parameter estimates for the two clusters/regimes are provided in~Table~\ref{tab:T_mar}. As for Portfolio~A, we found the semi-parametric approach with the generalized Pareto distribution for the tail observations works better when the upper and lower tails have very different tail decays. If this is not the case, the skew-t distribution provided adequate fit to the marginal components.


\begin{table}[H]
\centering
\small
\caption{Kendall's tau correlation coefficient ($\tau$) and estimated tail dependence coefficient ($\hat{\lambda}$) between risk factors and portfolio losses for Portfolio B. For cluster~2, risk factor values were negated in computation of tail dependence coefficient estimates.}
\begin{tabular}{cccccc|ccccc}
\hline
                                                & \multicolumn{5}{c|}{Cluster 1}              & \multicolumn{5}{c}{Cluster 2}              \\
                                                & CAD    & JPY    & AUD    & MXN    & EUR    & CAD    & JPY    & AUD    & MXN    & EUR    \\ \hline
$\tau$  & 0.793  & 0.068  & 0.465  & 0.346  & 0.343  & -0.607 & 0.023 & -0.632 & -0.590 & -0.311 \\
$\hat{\lambda}$ & 0.859  & 0.039  & 0.546  & 0.354  & 0.286  & 0.606  & 0.213  & 0.383  & 0.453  & 0.495  \\
\end{tabular}
\label{tab:T_dep}
\end{table}

\begin{table}[H]
\centering
\small
\caption{Marginal estimates for Portfolio B. The columns correspond to portfolio losses (L) and risk factors given by relative changes in currency exchange rates. Parameters $\alpha$ and $\beta$ are for the fitted skew-t distribution; ${\xi}_L$ and ${\xi}_U$ are shape parameters of the generalized Pareto distribution fitted to the lower and upper tails, respectively.}
\begin{tabular}{ccccccc|cccccc}
\hline
                                                & \multicolumn{6}{c|}{Cluster 1}              & \multicolumn{6}{c}{Cluster 2}              \\
                                           & L     & CAD    & JPY    & AUD    & MXN    & EUR &L    & CAD    & JPY    & AUD    & MXN    & EUR    \\ \hline
$\hat{\alpha}$ & -  & 1.938  & 1.774  & 1.695 & 2.361  & 2.289 & 3.489 & 9.752 & 1.592 & 10.000 & 10.000 & 8.901\\
$\hat{\beta}$   & - & 1.882 & 1.780 & 1.560 & 2.019 & 2.233 &3.365 & 10.000 & 1.723 & 7.884 & 7.059 & 7.050 \\
$1/\hat{\xi}_L$  & 6.757  & -  & -  & -  & -  & -  & -  & - & -  & -  & -  & -  \\
$1/\hat{\xi}_R$  & 25.641  & -  & -  & -  & -  & -  & -  & -  & -  &-  & - & -\\ \hline
\end{tabular}
\label{tab:T_mar}
\end{table}

\subsection{Portfolio C}
Portfolio C is the largest of the three considered portfolios determined by 18 currencies: CAD, JPY, GBP, Swiss franc (CHF), AUD, Danish Krone (DKK), Mexican Peso (MXN), HKD, South African rand (ZAR), Swedish krona (SEK), Norwegian krone (NOK), Turkish lira (TRY), New Zealand dollar (NZD), Singapore dollar (SGD), Thai baht (THB), Indian rupee (INR), Hungarian forint (HUF) and EUR. The scatter plots of normal scores between risk factors and portfolio losses are provided in~\ref{appen::sc}. All the currencies are found to be positively associated with the portfolio losses. 

Tables~\ref{tab:P_dep} and~\ref{tab:P_mar} provide dependence summary statistics and marginal parameter estimates for this portfolio. The estimates of the dependence measures reveal a broad range of dependence characteristics. The risk factor based on MXN currency has the strongest dependence with portfolio losses and is clearly also tail dependent, while majority of other risk factors show only mild degree of dependence or very week dependence as in the case of HKD. The marginal parameter estimates indicate differences in tail behaviour across risk factors and the portfolio loss.   

\begin{table}[H]
\centering
\small
\caption{Kendall's tau correlation coefficient ($\tau$) and estimated upper tail dependence coefficient ($\hat\lambda$) between risk factors and portfolio losses for Portfolio C. For JPY, risk factor values were negated in computation of tail dependence coefficient estimates.}
\begin{tabular}{cccccccccc}
\hline
      & CAD    & JPY    & GBP    & CHF    & AUD    & DKK    & MXN    & HKD    & ZAR    \\ \hline
$\tau$     & 0.237  & -0.066 & 0.183  & 0.099  & 0.310  & 0.184  & 0.844  & 0.073  & 0.389  \\
$\hat{\lambda}$ & 0.248  & 0.261  & 0.203  & 0.086  & 0.290  & 0.158  & 0.811  & 0.000  & 0.348  \\\\\hline
& SEK    & NOK    & TRY    & NZD    & SGD    & THB    & INR    & HUF    & EUR    \\ \hline
$\tau$  & 0.225  & 0.261  & 0.325  & 0.270  & 0.291  & 0.188  & 0.183  & 0.256  & 0.185  \\
$\hat{\lambda}$ & 0.238  & 0.264  & 0.176  & 0.250  & 0.269  & 0.145  & 0.157  & 0.210  & 0.156  \\
\end{tabular}
\label{tab:P_dep}
\end{table}

\begin{table}[H]
\centering
\small
\caption{Marginal estimates for Portfolio C. The columns correspond to portfolio losses (L) and risk factors given by relative changes in currency exchange rates. Parameters $\alpha$ and $\beta$ are for the fitted skew-t distribution; ${\xi}_L$ and ${\xi}_U$ are shape parameters of the generalized Pareto distribution fitted to the lower and upper tails, respectively.}
\begin{tabular}{ccccccccccc}
\hline
     & L & CAD    & JPY    & GBP    & CHF    & AUD    & DKK    & MXN    & HKD    & ZAR    \\ \hline
$\hat{\alpha}$    & 2.011  & 1.984 & 1.771  & 2.153  & 1.870  & 1.698  & 2.339  & 2.182  & -  & 3.055 \\
$\hat{\beta}$   & 1.806 & 1.922 & 1.762 & 2.087 & 1.939 & 1.560 & 2.266 & 1.880 & - & 2.531 \\
$1/\hat{\xi}_L$ & -  & -  & -  & -  & -  & -  & -  & -  & 2.825 &  - \\
$1/\hat{\xi}_R$ & -  & -  & -  & -  & -  & -  & -  & -  & 4.274 &  - \\ \\\hline
& SEK    & NOK    & TRY    & NZD    & SGD    & THB    & INR    & HUF    & EUR    \\ \hline
$\hat{\alpha}$  & 2.149  & 2.226  & 1.787  & 2.261  & 2.120  & 2.109  & -  & - & 2.333 \\
$\hat{\beta}$   & 2.000 & 1.975 & 1.557 & 2.127 & 2.053 & 2.072 & - & - & 2.263\\
$1/\hat{\xi}_L$ & -  & -  & -  & -  & -  & -  & 5.000  & 71.429 & - \\
$1/\hat{\xi}_R$ & -  & -  & -  & -  & -  & -  & 9.709  & 10.417  & - \\ \hline
\end{tabular}
\label{tab:P_mar}
\end{table}

\section{Background on vine copulas}
\label{sec::vine}

In this section, we will provide the relevant background on vine copulas. For more details, refer to \cite{Joe2014} and \cite{Czado2019}. 


As an effective tool for multivariate data analysis, the copula approach allows to model marginal distributions and the dependence structure separately. Mathematically speaking, the copula is a multivariate distribution with all univariate margins being standard uniform. To be more specific, let $\Xb = (X_1,X_2,...,X_d)^\top$ be a $d$-dimensional random vector with joint cumulative distribution function (cdf) $F_{\Xb}$ and marginals $F_i(x) = P(X_i\leq x)$, $i=1,\ldots,d$. The copula of $\Xb$ is defined as
\begin{equation}
    C(u_1,u_2,...,u_d) = \pbb\big(X_1 \leq F_1^{\leftarrow}(u_1), X_2 \leq F_2^{\leftarrow}(u_2),..., X_d \leq F_d^{\leftarrow}(u_d)\big),\quad \ub\in [0,1]^d,
    \label{eq:copula}
\end{equation}
where $F^{\leftarrow}_i$ is the generalized inverse of $F_i$. According to Sklar's theorem (\cite{Sklar1959}), the joint cdf $F_{\Xb}$ of $\Xb$ can be written as
$$F_{\Xb}(x_1,x_2,...,x_d) = C\big(F_1(x_1), F_2(x_2),...,F_d(x_d)\big),\quad \xb\in\rbb^d,$$
where $C$ is the copula defined in~\eqref{eq:copula}. The copula~$C$ is unique if the marginals $F_1,\ldots,F_d$ are continuous.
If, furthermore, $F_{\Xb}$ is absolutely continuous, then its density function can be expressed as
$$f_{\Xb}(x_1, x_2,...,x_d) = c\big(F_1(x_1),F_2(x_2),...,F_d(x_d)\big)\cdot \prod_{j = 1}^d f_j(x_j),$$
where $f_j$ is the density function of $F_j$ ($j=1,\ldots,d$) and $c$ is the density function of $C$.

There exist many parametric copula families that can capture a variety of dependence structures. Popular examples include elliptical, Archimedean and extreme-value copulas; see Chapter 4 in~\cite{Joe2014} for further classes of copulas.


Although copulas allow to separate the modeling of marginals and dependence, in a high-dimensional setting, it may not be possible to choose a single multivariate copula family that adequately captures the entire multivariate dependence structure. Vine copulas allow the construction of multivariate dependence structures using  only bivariate copula families and vine graphs. With a given vine structure, a multivariate density can be decomposed into the product of pair copulas and marginal densities. A vine is a graphical model consisting of a sequence of trees. \cite{BedfordCooke2001} introduce the definition of a regular vine (R-vine), which is useful to identify all possible decompositions of a $d$-dimensional multivariate density function.



\begin{definition}
We say $\VC$ is a regular vine on $d$ elements, with $\EC(\VC) = \EC(\TC_1)\cup...\cup \EC(\TC_{d-1})$ denoting the set of edges of $\VC$, if
\begin{enumerate}
    \item $\VC$ consists of $d-1$ trees: $\{\TC_1,...,\TC_{d-1}\}$;
    \item $\TC_1$ is a connected tree with nodes $\NC(\TC_1) = \{1,...,d\}$ and edges $\EC(\TC_1)$; for $k = 2,...,d-1$, $\TC_k$ is a connected tree with nodes $\NC(\TC_{k}) = \EC(\TC_{k-1})$ and edges $\EC(\TC_k)$;
    \item for $k = 2,...,d-1$ and $\{n_1, n_2\}\in \EC(\TC_k)$, $|n_1\cup n_2| = 1$, where $|\cdot|$ denotes cardinality (i.e., nodes connected by an edge in $\TC_k, k\geq 2$, share a common node in $\TC_{k-1}$).
\end{enumerate}
\label{def:rvine}
\end{definition}

The last property is called the proximity condition. There are two special cases of regular vines: a canonical vine (C-vine) and a drawable vine (D-vine). A regular vine is called a C-vine if there is only one node of degree\footnote{The degree of a node is the number of edges linked to that node.} greater than 1 in each tree. A regular vine is called a D-vine if all nodes have degree smaller than 3. Figures~\ref{fig:cvine} and~\ref{fig:dvine} give examples of a C-vine and a D-vine with 4 variables, respectively.
\begin{figure}[ht]
\centering
\includegraphics[scale=0.5]{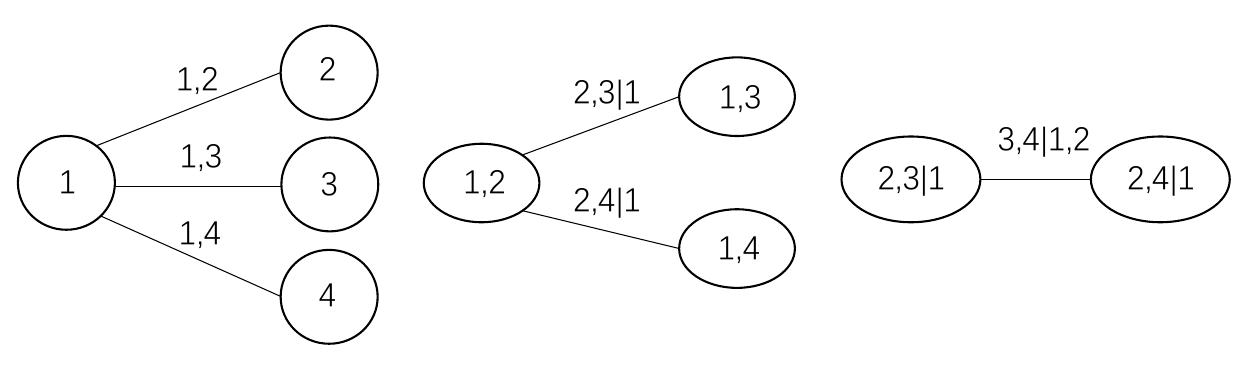}
\caption{A C-vine with 4 variables and 3 trees (from left to right: $\TC_1, \TC_2$ and $\TC_3$).}
\label{fig:cvine}
\end{figure}

\begin{figure}[ht]
\centering
\includegraphics[scale=0.5]{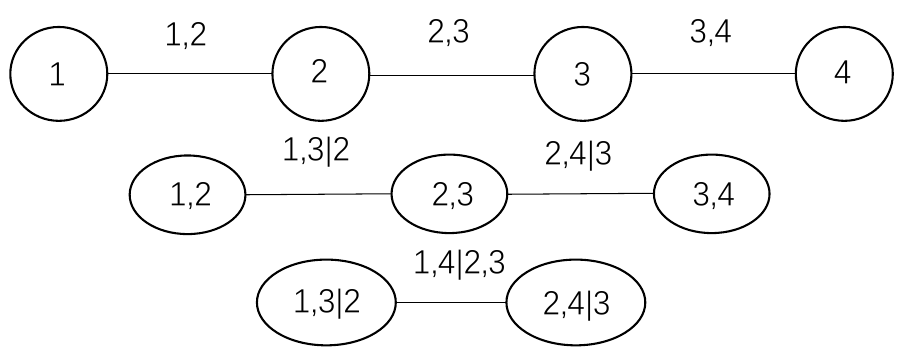}
\caption{A D-vine with 4 variables and 3 trees (from top to bottom: $\TC_1, \TC_2$ and $\TC_3$).}
\label{fig:dvine}
\end{figure}

For a given R-vine tree structure $\VC$ with $\NC(\TC_1) = \{X_1,...,X_d\}$, a vine copula for $\Xb\in \rbb^d$ can be constructed by assigning a bivariate copula to each edge. Let $i,j\in \{1,...,d\}$, $S(i,j)\subset \{1,..,d\}\setminus \{i,j\}$ and $\Xb_{S(i,j)}$ be the subvector of $\Xb$ made up of the variables that are indexed by the elements in set $S(i,j)$. For each edge $[i,j|S(i,j)]\in \EC(\VC)$, suppose $C_{i,j;S(i,j)}$ is the copula associated with the edge, that is the copula of $(X_i, X_j)$ given $\Xb_{S(i,j)}$, and suppose $c_{i,j;S(i,j)}$ is the corresponding density function. By imposing the simplifying assumption\footnote{Let $\tilde{d} = |S(i,j)|$. We say $C_{i,j;S(i,j)}$ satisfies the simplifying assumption if for any $\xb_{S(i,j)}\in \rbb^{\tilde{d}}$, $$C_{i,j;S(i,j)}(u,v|\xb_{S(i,j)}) = C_{i,j;S(i,j)}(u,v),\quad u,v\in [0,1].$$} on $C_{i,j;S(i,j)}$, the joint copula density of $\Xb$ can be written as 
$$c(u_1,u_2,...,u_d) = \prod_{\EC(\VC)}c_{i,j;S(i,j)}\Big(C_{i|S(i,j)}\big(u_i|\ub_{S(i,j)}\big), C_{j|S(i,j)}\big(u_j|\ub_{S(i,j)}\big)\Big),$$
where $C_{j|S(i,j)}\big(u_j|\ub_{S(i,j)}\big) = \pbb\big(F_j(X_j) \leq u_j|F_m(X_m) = u_m, m\in S(i,j)\big)$ for $j = 1,...,d$. Consequently, the joint density function of $\Xb$ can be expressed as
\begin{equation}
    f_{\Xb}(x_1,x_2,...,x_d) = \prod_{j = 1,..,d}f_j(x_j)\cdot\prod_{\EC(\VC)}c_{i,j;S(i,j)}\Big(F_{i|S(i,j)}\big(x_i|\xb_{S(i,j)}\big), F_{j|S(i,j)}\big(x_j|\xb_{S(i,j)}\big)\Big),\quad \xb\in\rbb^d,
    \label{eq:vcop}
\end{equation}
where $F_{j|S(i,j)}\big(x_j|\xb_{S(i,j)}\big)$ is the conditional cdf of $X_j$ given $\Xb_{S(i,j)} = \xb_{S(i,j)}$.


\cite{Morales2011} showed that for $d$ variables there are $\dfrac{d!}{2}\cdot 2^{(d-3)(d-2)/2}$ R-vine tree structures. When $d$ is large, the number of R-vine tree structures is enormous. Thus, constructing a vine structure via enumerated optimization is impossible. An alternative is to use a sequential construction of $\TC_1,...,\TC_{d-1}$. One of the most widely-used construction algorithms is the minimum/maximum spanning tree algorithm (for more details, see \cite{Joe2014}), which aims to choose a tree such that pairs of variables have the strongest dependence. The pair copula families are selected independently via maximizing the Akaike information criterion (AIC) or Bayesian information criterion (BIC). The parameters of pair copulas are estimated sequentially via maximum likelihood. R packages \texttt{VineCopula} (\cite{VineCopula}) and \texttt{rvinecopulib} (\cite{rvinecopulib}) provide functions to select vine structures and fit vine copula models.

\section{Estimation}
\label{sec::est}



Let $(\Xb, L)$ be a $(d+1)$-dimensional random vector containing $d$ risk factors $\Xb$ and the associated portfolio loss~$L$. From the definition of the stress scenario at threshold~$\ell$ in~\eqref{eq:RST_L}, a direct estimator of $\mb(\ell)$ is given by
$$\widehat{\mb}(\ell) = \argmax_{\xb\in \rbb^d} \hat{f}(\xb|L\geq \ell),$$
with the estimator of the conditional density computed via
\begin{equation}
\hat{f}(\xb|L\geq \ell) =  \frac{\int_{\ell}^{\infty} \hat{f}_{\Xb, L}(\xb, z)\rmd z}{1-\hat{F}_L(\ell)},
\label{eq:integral}
\end{equation}
where
$\hat{f}_{\Xb,L}$ is an estimate of the joint density of $(\Xb, L)$ and $\hat{F}_L$ is an estimate of the cdf of $L$. To estimate $f_{\Xb, L}$ we consider modelling $(\Xb, L)$ using vine copulas.

The integration in the numerator of~\eqref{eq:integral} can be avoided. The idea is similar to that used in the vine copula regression context; see, e.g., \cite{KrausCzado2017}, \cite{Schallhorn_etal2017} and \cite{ChangJoe2019}. Note that the conditional density of $\Xb$ given $L\geq \ell$ can also be expressed as 
\begin{equation}
    f(\xb|L\geq \ell) = \frac{f_{\Xb}(\xb)}{1-F_L(\ell)}\cdot \Big[1-F_{L|\Xb}\big(\ell|\xb\big)\Big],\quad \xb\in \rbb^d,
\label{eq:vreg}
\end{equation}
where $F_{L|\Xb}$ denotes the condition cdf of $L$ given $\Xb$. In the literature on vine copula regression, to obtain $F_{L|\Xb}(\cdot|\xb)$ without integration, the node that contains random variable~$L$ as the conditioned variable is constrained to be a leaf node\footnote{A node is a leaf node in a tree if its degree is one.} in all trees (for more details, see~\ref{appen::vine_reg}). By constructing the vine structure in this way, Algorithm~3.1 in \cite{ChangJoe2019} provides a recursive procedure to calculate the conditional cdf $F_{L|\Xb}(\ell|\xb)$.

Initial simulation studies showed that the estimator in~\eqref{eq:integral} based on the conditional density calculation via~\eqref{eq:vreg} tends to overestimate true values of stress scenarios. This can easily be explained by comparing the naive estimator of the stress scenario described above, which is equivalent to 
$$\widehat{\mb}(\ell) = \argmax_{\xb\in\rbb^d} \hat{f}_\Xb(\xb)\cdot\Big[1 - F_{L|\Xb}\big(\ell|\xb\big)\Big],$$
with the functional we actually aim to estimate:
$$\mb(\ell) = \argmax_{\xb\in\rbb^d} f(\xb|g(\Xb)\geq \ell) = \argmax_{\xb\in\rbb^d} f_{\Xb}(\xb)\cdot\oneb_{\{\xb: g(\xb)\geq \ell\}}.$$
Instead of multiplication by the indicator function, the estimator involves a factor that is between zero and one, ignoring the fact that portfolio loss $L$ is fully determined by the risk factors. This suggests that estimation can likely be improved by explicitly modelling function $g$ linking risk factors to the portfolio loss.

Possible approaches to estimate function $g$ include simple linear regression, vine copula regression and machine learning techniques. In our simulation studies and applications, we use the vine copula regression, which can model nonlinear functions as well as heteroscedasticity.

We should note that even when the estimate $\hat{g}$ can approximate function $g$ accurately, the effect of risk factors on portfolio losses may not be properly captured, especially when there is a large number of risk factors influencing the portfolio and some risk factors are only weakly dependent with the portfolio loss. When dependence between a risk factor $X_i$ and $L$ is not correctly modelled, the corresponding estimate of the stress scenario component $\widehat{m}_i(\ell)$ will deviate from the true value even though the prediction $\hat{g}(\xb)$ from the regression is close to the true value $g(\xb)$.


Another possibility to improve the behaviour of the initial estimator in~\eqref{eq:integral} is to explore the connection 
between densities of risk factors conditional on stress event $\{L\ge \ell\}$ versus $\{L=\ell\}$. Under certain conditions on the functions $f_\Xb(\xb)$ and $g(\xb)$, we obtain the following equality:
$$\argmax_{\xb\in\rbb^d} f(\xb|L\geq \ell)=\argmax_{\xb\in\rbb^d} f_{\Xb|L}(\xb|\ell),$$ 
where $f_{\Xb|L}(\cdot|\ell)$ is the conditional density of $\Xb$ given that $L = \ell$. The following proposition gives the conditions.

\begin{proposition}
Consider a random vector $\Xb$ on $\rbb^d$ with a density function $f_\Xb$. Let $g$ be a real-valued function on $\rbb^d$, and $\ell\in\rbb$. Assume that $f_\Xb$ and $g$ have continuous first derivatives, and for all feasible points $\xb$ (i.e., $g(\xb)-\ell\geq 0$), $\nabla g(\xb)\neq \zerob$ and $\nabla f_\Xb(\xb)\neq \zerob$. Then we have
$$\argmax_{\xb\in\rbb^d} f\big(\xb|g(\Xb)\geq \ell\big) = \argmax_{\xb\in\rbb^d} f_{\Xb|g(\Xb)}\big(\xb|\ell\big),$$
where $f(\cdot\mid g(\Xb)\geq \ell)$ is the conditional density of $\Xb$ given $g(\Xb)\geq \ell$ and $f_{\Xb|g(\Xb)}(\cdot|\ell)$ is the conditional density of $\Xb$ given $g(\Xb)= \ell$.
\label{prop:gt_to_eq}
\end{proposition}


\begin{proof}
Finding the $\argmax$ of conditional density $f\big(\xb|g(\Xb)\geq \ell\big)$ is actually based a constrained optimization problem expressed as:
\begin{equation}
    \min_{\xb\in\rbb^d} -f_{\Xb}(\xb),\quad \text{subject to } g(\xb) - \ell\geq 0.
    \label{eq:NLP}
\end{equation}
By introducing a squared slacking variable $\theta^2 = g(\xb) - \ell$, the problem in~\eqref{eq:NLP} can be reformulated as
\begin{equation}
    \min_{\xb\in \rbb^d,\theta\in \rbb} -f_{\Xb}(\xb),\quad \text{subject to } g(\xb) - \ell - \theta^2 = 0.
\label{eq:NLP_eq}
\end{equation}
Let $\xb^*$ be the optimum in~\eqref{eq:NLP}. If $(\theta^*)^2 = g(\xb^*) - \ell$, $(\xb^*, \theta^*)$ is an optimum in~\eqref{eq:NLP_eq}. Due to the fact that $\xb^*$ is definitely a feasible point, we have $\nabla g(\xb^*)\neq \zerob$, and thus there exists a unique Lagrange multiplier $\lambda^*\in \rbb$ such that
\begin{equation}
    -\nabla f_{\Xb}(\xb^*) + \lambda^* \nabla g(\xb^*) =0;
    \label{eq:sta_con}
\end{equation}
and
\begin{equation}
   \lambda^*\theta^* = 0\Rightarrow \lambda^*{\sqrt{g(\xb^*) - \ell}} = 0. 
   \label{eq:comp_slack}
\end{equation}
Furthermore, the fact that $\nabla f_{\Xb}(\xb^*)\neq \zerob$ indicates $\lambda^*\neq 0$ from~\eqref{eq:sta_con} and thus $g(\xb^*)-\ell = 0$ from~\eqref{eq:comp_slack}. Note that $\{\xb:g(\xb) - \ell = 0\}$ is a subset of $\{\xb:g(\xb) - \ell \geq 0\}$, hence $\xb^* = \argmax f_{\Xb|g(\Xb)}\big(\xb| \ell\big)$.
\end{proof}

In the problem of RST, the threshold $\ell$ is usually large. When there are some components of $\Xb$ that are strongly or moderately dependent with $L$, large threshold $\ell$ will make the constrained region to be extreme and thus the density function $f_{\Xb}(\xb)$ to be monotone for these components. This suggests that the condition $\nabla f_{\Xb}(\xb)\neq \zerob$ for all feasible points $\xb$ can be satisfied.

With this connection, rather than to estimate $\argmax f(\xb|L\geq \ell)$, we can estimate $\argmax f_{\Xb|L}(\xb| \ell)$, which can be easily obtained from the estimated joint density of $(\Xb, L)$. In our simulation studies, this estimator is found to have good performance as well as being more computationally efficient than the other two copula-based estimators discussed above.

\begin{remark}
{\rm Note that under the setting when $L$ is a function of $\Xb$, the joint density of $(\Xb, L)$ should be degenerate and it is not sensible to estimate it. However, when we estimate the joint density with vine copulas, estimation is done pairwisely, so an estimated joint density can be obtained. In our simulation and application studies, we found the estimator based on $\argmax f_{\Xb|L}(\xb|\ell)$ also shows good performance.} 
\end{remark}

\begin{remark}
{\rm Although the copula-based methodologies can deal with independent and dependent cases at the same time,  to reduce computational complexity in medium and large portfolios, we can choose to remove independent risk factors before applying the framework, as the corresponding stress scenarios for the independent risk factors are just their marginal modes. To be more specific, suppose $\Xb_1$ contains risk factors which are strongly dependent with portfolio losses, while $\Xb_2$ contains risk factors that have only very weak dependence with portfolio losses. 
We can split our estimator $\widehat{\mb}(\ell)$ into $\widehat{\mb}_1(\ell)$ and $\widehat{\mb}_2(\ell)$. The second part $\widehat{\mb}_2(\ell)$ can be obtained by estimating the marginal mode of $\Xb_2$, while the first part $\widehat{\mb}_1(\ell)$ is calculated using one of the copula-based methods above using data on $(\Xb_1, L)$. Under this situation, $L$ is not determined by $\Xb_1$ anymore and the joint density of $(\Xb_1, L)$ exists.}
\end{remark}

\section{Simulation studies}
\label{sec::sim}
In this section, we do simulation studies to examine the performance of the proposed estimation procedure in this paper and compare with the estimator proposed in \cite{Glasserman_etal2015} (denoted as GKK). Two models are considered in the simulation studies. The first model is a bivariate t distribution, where the assumptions of the GKK method are fully satisfied. The second model is a 4-dimensional distribution with an R-vine copula and different skew-t distributed marginals based on the fit to Portfolio A data, which mimic the behavior of real financial portfolios. For each simulation setting, we evaluate the following three copula-based estimators of stress scenarios:
\begin{itemize}
    \item Estimator that equates the conditional densities under conditioning on $\{L\ge\ell\}$ and $\{L=\ell\}$:
    $$\widehat{\mb}(\ell) = \argmax_{\xb\in\rbb^d} \frac{\hat{f}_{\Xb,L}(\xb, \ell)}{\hat{f}_L(\ell)}\hspace{3cm} (\CM_1);$$
    \item Estimator that uses an estimate of function $g$ obtained via vine regression:
    $$\widehat{\mb}(\ell) = \argmax_{\xb\in\rbb^d} \hat{f}_{\Xb}(\xb)\cdot\onebb\{\hat{g}(\xb)\geq \ell\}\hspace{2cm} (\CM_2);$$
    \item The naive estimator:
    $$\widehat{\mb}(\ell) = \argmax_{\rbb^d} \hat{f}_{\Xb}(\xb)\cdot \big[1 - \hat{F}_{L|\Xb}(\ell|\xb)\big]\hspace{2cm} (\CM_3).$$
\end{itemize}

To make comparison among different estimators, we use the following three summary measures: mean percentage error (MPE) and root mean squared percentage error (RMSPE) for each univariate component of the $d$-dimensional stress scenario, and the average $L_2$-norm of differences (ML2) as a combined measure. Specifically, for a scalar functional~$\theta$, if $\hat{\theta}_1,...,\hat{\theta}_r$ denote its estimates in each of $r$ simulated samples, the MPE and RMSPE are given by
$$\text{MPE} = \frac{1}{r}\sum_{j=1}^r \frac{\hat{\theta}_j - \theta}{\theta}\cdot 100\%\quad\text{and}\quad\text{RMSPE} = \sqrt{\frac{1}{r}\sum_{j=1}^r \frac{(\hat{\theta}_j - \theta)^2}{\theta^2}}\cdot 100\%.$$
For a vector-valued functional~$\tb$ with estimates denoted by $\hat{\tb}_1,...,\hat{\tb}_r$ from the $r$ simulated samples, the ML2 measure is given by $$\text{ML2} = \frac{1}{r}\sum_{j=1}^r||\hat{\tb}_j - \tb||_2.$$

As in any estimation problem, even if an estimator is unbiased, sampling variability will lead to its value being above or below the true value of the parameter or functional that it is designed to estimate. In the context of stress scenario estimation, this means that in some cases the portfolio loss corresponding to an estimate of the stress scenario at threshold~$\ell$ may actually not exceed the specified threshold and this may be of practical concern. To assess this aspect of stress scenario estimation, we also report the percentage of the total number of replications~$r$ that resulted in portfolio losses greater or equal to threshold~$\ell$, denoted as $E_r$.  


\subsection{Bivariate t distribution}\label{ss1}

As a first data generating process for risk factors $\Xb$, we consider a bivariate t distribution with mean vector $\mub = \zerob$, the degrees of freedom parameter $\nu = 4$ and covariance matrix $\Sigma = \left(\begin{matrix}
   1 & 0.5   \\
   0.5 & 1   \\
   \end{matrix} \right)
$. The portfolio loss $L$ is set to a linear combination of components of $\Xb$: $L = 0.7X_1 + 0.3X_2$. The threshold $\ell$ is chosen as the $0.99$-quantile of the distribution of $L$. We generate $r = 500$ random samples of size $n = 3000$, resulting in 500 stress scenario estimates for each method. Based on these estimates, the sampling density plots are displayed in Figure~\ref{fig:t_d2_den} and the corresponding MPE, RMSPE and ML2 summary statistics are provided in Table~\ref{sim:t_d2}.

\begin{figure}[ht]
\centering
\subfigure{\includegraphics[scale=0.55]{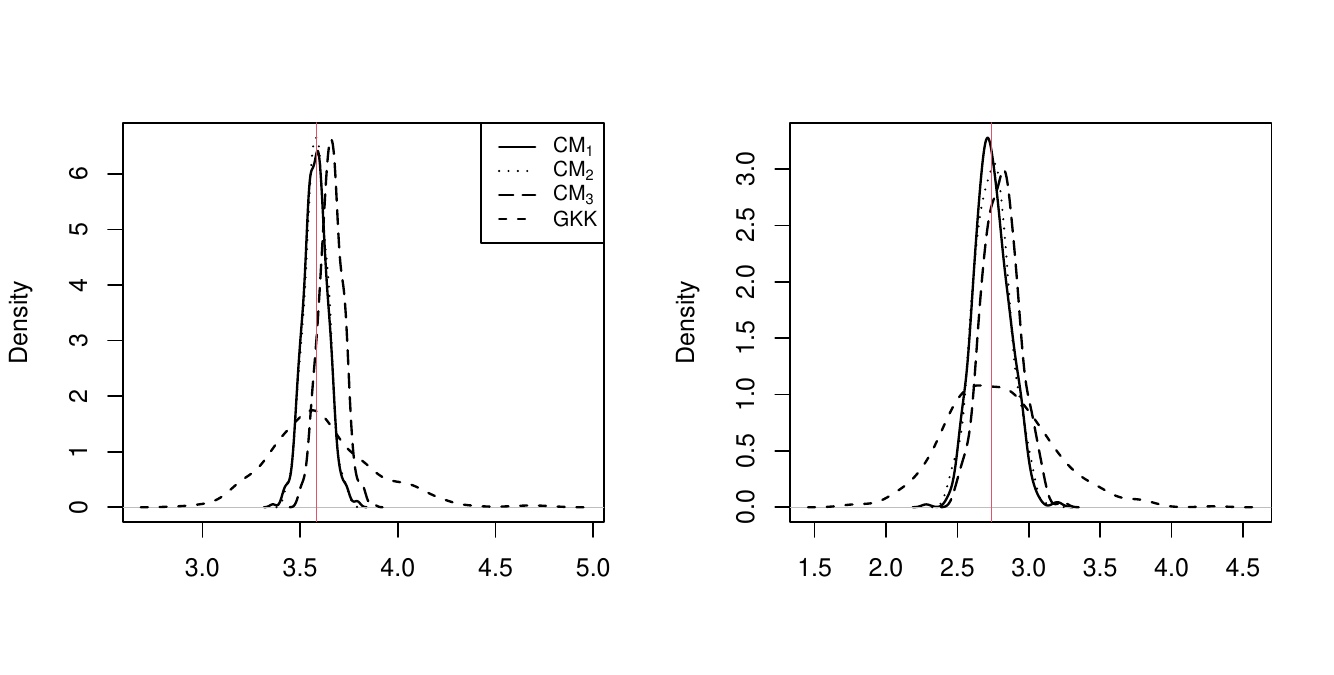}}
\caption{Sampling densities of the two marginal components of stress scenario estimates in the setting of Section~\ref{ss1}, where risk factors follow a bivariate t distribution. The vertical lines indicate the true values of stress scenario marginal components.}
\label{fig:t_d2_den}
\end{figure}

\begin{table}[H]
\centering
\small
\caption{Summary statistics for stress scenario estimates in the setting of Section~\ref{ss1}, where risk factors follow a bivariate t distribution.}
\begin{tabular}{c|cc|cc|cc}
\hline
              & \multicolumn{2}{c|}{MPE(\%)}       & \multicolumn{2}{c|}{RMSPE(\%)}   & ML2  & $E_r$ \\
& $X_1$    & $X_2$       & $X_1$    & $X_2$   & ($\times 100$)  & (\%)\\ \hline\hline
 $\CM_1$    & -0.098  & {\bf 0.005}     & 1.743   & {\bf 4.534}  & {\bf 11.852} & 47.0\\
  $\CM_2$    &{\bf -0.044}    & 0.089  & {\bf 1.656} & 4.601 & 12.012 & 47.6\\
$\CM_3$ & 2.124  & 2.368    & 2.733  & 5.270  & 15.639 & 95.4\\\hline\hline
GKK          & 0.828 & 1.687 & 7.618   & 13.379 & 38.776  & 50.4\\ \hline
\end{tabular}
\label{sim:t_d2}
\end{table}

From the sampling density plots, we observe that all three copula-based methods have similar variability. The naive estimator $\CM_3$ shows a visible positive bias for both marginal components, while $\CM_1$ and $\CM_2$ estimators perform best in terms of marginal MPE and RMSPE measures as well as the combined measure ML2. The GKK estimator demonstrates good performance in terms of bias. This is to be expected since the simulation setting satisfies all of the assumptions under  which the GKK estimator was designed. However, due to a relatively large threshold~$\ell$, the GKK estimator also exhibits a much larger variance compared to the copula-based estimators.


While a bias is usually an undesired characteristic of a statistical estimator, a positive bias of the $\CM_3$ estimator may actually make this estimator a conservative option from the risk management perspective. In 95\% of the simulated samples, estimated stress scenarios led to a portfolio loss above the specified threshold~$\ell$; see values of $E_r$ in Table~\ref{sim:t_d2}. In contrast, the other three estimators resulted in stress scenarios for which in only about $50\%$ of simulated samples, portfolio losses exceeded the threshold. It should be noted that in this context, it is especially important to have an estimator with a small variance so as to avoid significant deviations from the desired level of how extreme the portfolio loss is to be at the selected stress scenario.


\subsection{Meta-vine distribution with skew-t marginals}\label{ss2}

We also consider a data generating process that comes from a model fitted to a real dataset. As a dataset, we took portfolio~A data  described in Section~\ref{spA} and fitted univariate skew-t marginals (see~\ref{appen::margin} for details) and an R-vine copula for the dependence structure. We refer to the resulting distribution as a meta-vine distribution with skew-t marginals. The contour plots of the fitted pair copulas on normal scores are given in Figure~\ref{fig:sim_cont}. For simulation of risk factors, we first generate random samples on the uniform scale from the fitted R-vine copula, and then transform them to the original scale via the quantile transform of the fitted skew-t marginals. As a final step, to evaluate portfolio losses based on the simulated values of risk factors, portfolio function~$g$ is set to the fitted value from a linear regression of the portfolio loss on risk factors, leading to the map of the form:
$$g(\xb) = -1.944x_{1} - 0.018 x_{2} + 0.011 x_{3} + 0.011 x_{4},\qquad \xb\in\rbb^4.$$
While other regression techniques could be used, for this particular portfolio we found the linear structure to be realistic with R-squared coefficient close to 0.99.

\begin{figure}[ht]
\centering
\includegraphics[scale=0.3]{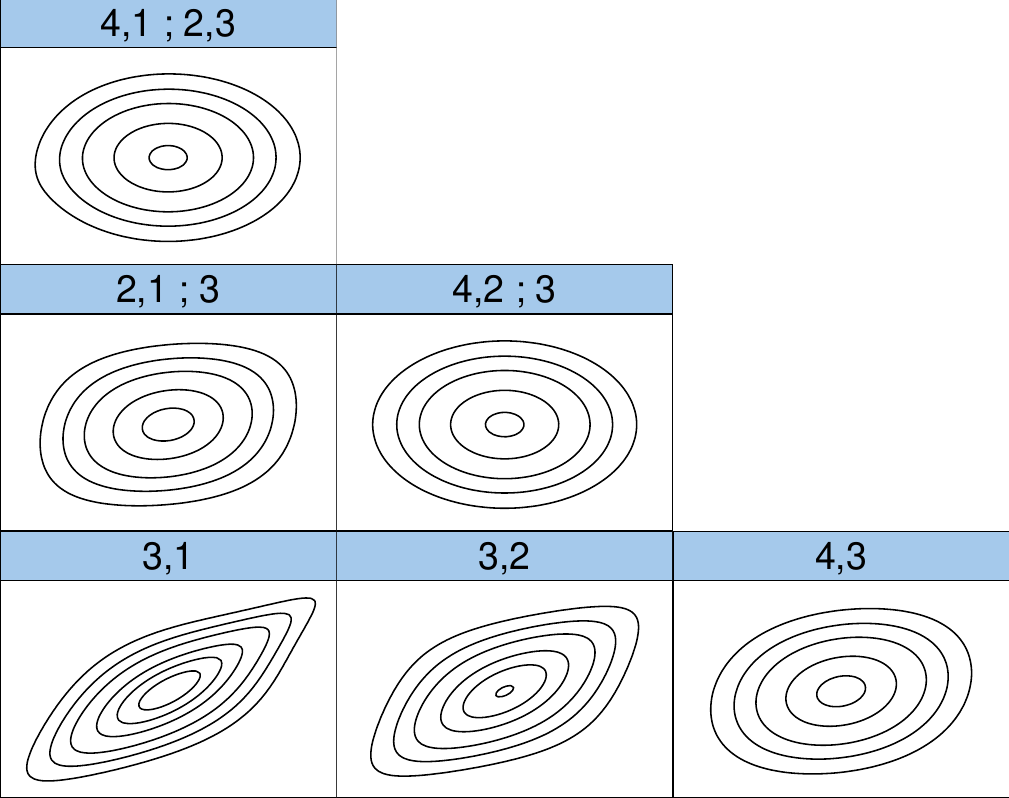}
\caption{Contour plots of fitted pair copulas on normal scores for portfolio A data (from top to bottom: $\TC_3, \TC_2$ and $\TC_1$). Variables $X_1, X_2, X_3$ and $X_4$ are denoted by 1,2,3 and 4, respectively.}
\label{fig:sim_cont}
\end{figure}

\subsubsection{Results}
For this model, we generate 100 random samples of size $n = 3000$. The threshold $\ell$ is taken as the $0.99$-quantile of actual losses for portfolio A, giving $\ell = 0.03$. The true value of the stress scenario at threshold $\ell$ is computed numerically and is equal to $\mb(\ell)=(-1.54\times 10^{-2}, -5.41\times 10^{-3}, -1.05\times 10^{-2}, -1.96\times 10^{-5})$. 

For this simulation study, we also include an estimator that is obtained by assuming function $g$ is known:
$$\widehat{\mb}(\ell) = \argmax_{\xb\in\rbb^4} \hat{f}_{\Xb}(\xb)\cdot\onebb\{-1.944x_{1} - 0.018x_{2} + 0.011x_{3} + 0.011x_{4}\geq \ell\},$$
and denote it by $\CM_*$. This ``oracle" estimator, which uses the knowledge of the true form of function $g$, can give us information about the magnitude of the error introduced by having to estimate function~$g$. 

The results of the simulation study are summarized in Table~\ref{sim:PA}. First of all, we can see that estimator~$\CM_*$ outperforms all the other estimators for three of the four risk factors, thus confirming that the knowledge of function~$g$ is important for accurate estimation of stress scenarios. Of course, in practice, function~$g$ is not always known and its values can only be obtained for a sample of risk factor values based on internal (black-box) portfolio valuation models. 

The three copula-based estimators show very similar performance on the basis of the root mean squared percentage error and combined measure ML2. The best performing estimator in this setting, especially in terms of bias, is $\CM_2$. This is the estimator that incorporates explicit modelling of function~$g$ into the estimation procedure. Perhaps not surprisingly, the GKK estimator does not do well in this setting due to the fact that the assumption of an elliptical distribution for risk factors does not hold here.

When comparing the performance of estimators across the four marginal components, the most accurate and precise results are achieved for the first component $X_1$, which contributes most to the portfolio loss, and least accurate results are obtained for the last component $X_4$. By reviewing scatter plots and correlation coefficients between $X_i$'s ($i = 1,2,3,4$) and $L$, it is found that $X_1$ exhibits strongest dependence with $L$, while $X_4$ displays weakest dependence (with an absolute value of the Kendall's tau coefficient of approximately 0.09). So the performance of stress scenario estimators tends to deteriorate with weakening of dependence between risk factors and portfolio losses. This pattern is particularly apparent for the GKK estimator; an implicit model assumption for the GKK estimator is that $X_i$'s and $L$ are (tail) dependent when the underlying distribution is heavy-tailed.

\begin{table}[H]
\centering
\small
\caption{Summary statistics of stress scenario estimates in the setting of Section~\ref{ss2}, where risk factors follow a 4-dimensional meta-vine distribution with skew-t marginals. The values for estimator~$\CM_*$ are provided in grey to highlight the fact that this estimator would not generally be applicable in practice unless the form of function $g$ is available.}
\begin{tabular}{c|cccc|cccc|cc}
\hline
  & \multicolumn{4}{c|}{MPE(\%)}       & \multicolumn{4}{c|}{RMSPE(\%)}   & ML2   & $E_r$\\
& $X_1$    & $X_2$   & $X_3$    & $X_4$      & $X_1$    & $X_2$   & $X_3$    & $X_4$   & ($\times 100$)   & (\%)  \\ \hline\hline
{\color{gray}$\CM_*$}      & {\color{gray}0.007}  & {\color{gray} 1.111}  & {\color{gray}2.672}  & {\color{gray}-8.051}   & {\color{gray} 0.030}   & {\color{gray} 9.328}  & {\color{gray} 10.584} & {\color{gray} 28.962}  & {\color{gray} 0.097} & {\color{gray}100.0} \\\hline\hline
$\CM_1$    & 0.044  & 3.987  & 5.314  & {\bf -4.952}   & 0.109  & 9.992  & 11.287 & 29.473  & 0.110  & 70.0\\
$\CM_2$    & {\bf 0.041}  & {\bf 1.540}  & {\bf 2.611}  & -8.288   & {\bf 0.109}  & {\bf 9.369}  & {\bf 10.625} & 29.019 & {\bf 0.098}  & 69.0\\
$\CM_3$ & 0.672  & 2.756  & 3.179  & -8.122   & 0.680   & 9.728  & 10.858 & {\bf 29.008}  & 0.101 & 99.0\\\hline\hline
GKK          & -0.099 & 30.948 & 50.497 & -2697.350 & 7.938  & 38.539 & 55.226 & 45091.702 & 0.668  & 43.0 \\ \hline
\end{tabular}
\label{sim:PA}
\end{table}

\section{Application}
\label{sec::app}
In this section, we apply the proposed copula-based methods on the three real-life portfolios of currencies introduced in Section~\ref{sec::data}. 

The copula-based methods rely on a joint distribution fitted to the data by first modelling the univariate marginal components and then the dependence structure, in our case through a vine copula. The modelling and fit of the marginal components were reported and discussed in Section~\ref{sec::data}. In this section we will only comment on the fitted copula models. 

While it is not possible to assess accuracy of estimated stress scenarios, in our presentation of the results, we will focus on two aspects of the stress scenario estimation problem which include sampling variability and calibration in terms of balancing the plausibility and severity requirements. 

Due to the lack of large-sample theory for estimators of vine copula parameters and consequently for the copula-based stress scenario estimators, estimation uncertainty for stress scenarios cannot be quantified using asymptotic approximations. We therefore adopt a 
bootstrapping scheme. Specifically, we resample observations with replacement from the original dataset using the block bootstrap algorithm proposed in~\cite{PolitisRomano1994}, where the block sizes are randomly generated from a geometric distribution with mean $n^{1/3}$; here $n = 3077$, the original sample size.  The process is repeated 500 times to generate bootstrap samples of length~$n$. Using the resulting stress scenario estimates from each sample, we are able to construct confidence intervals for univariate components.

To assess what we refer to as calibration, we evaluate the fitted portfolio map function $\hat g$ at the estimated stress scenario $\widehat\mb(\ell)$. In the ideal situation, we would like the stress scenario estimate at threshold~$\ell$ to be perfectly calibrated giving portfolio loss at $\widehat\mb(\ell)$ equal to the desired threshold $\ell$. This ensures that the selected stress scenario meets the severity requirement ($L=g(\Xb)\ge \ell$), yet remains the most plausible data point within the risk region $\{\xb:g(\xb)\ge\ell\}$. In the application, $g$ has to be substituted with its estimate $\hat g$, which we compute using vine regression. Approximate 95\% confidence intervals are provided to account for estimation uncertainty associated with $\hat g$.


\subsection{Portfolio A}


Figure~\ref{fig:A_T1} illustrates the first tree of the R-vine fitted to Portfolio~A data. From this plot we can see that, in the first tree of the selected R-vine, the portfolio loss $L$ is only connected with the risk factor based on currency CAD. This is consistent with the results in Table~\ref{tab:A_dep}, where this risk factor is shown to have strongest dependence with portfolio losses. The other lower-order trees as well as the selected pair copulas with their estimated parameter values are provided in~\ref{appen::tree_detail}. For this portfolio, we estimate stress scenarios at threshold $\ell$ equal to the $0.99$ and $0.999$-empirical quantiles of the portfolio losses. 

\begin{figure}[ht]
\centering
\includegraphics[scale = 0.5]{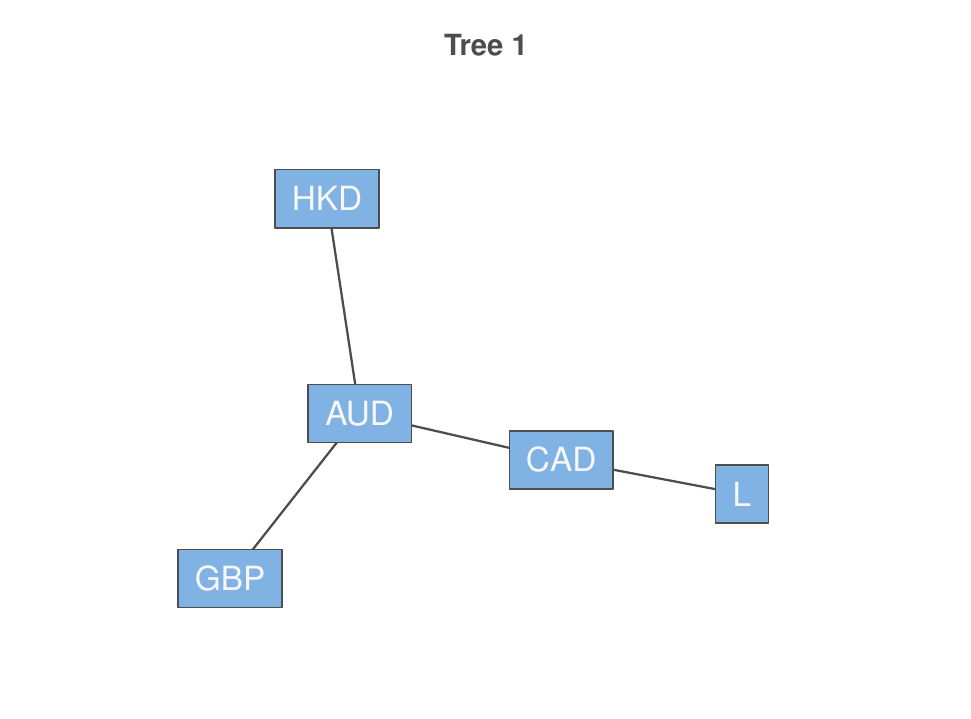}
\caption{The first tree in the R-vine selected for the portfolio A data.}
\label{fig:A_T1}
\end{figure}

Figure~\ref{fig:A_CI} displays estimated stress scenarios for each risk factor in portfolio~A against the fitted value of the portfolio loss, $\hat g(\widehat \mb(\ell))$, for two values of threshold~$\ell$  under the three copula-based methods as well as the GKK method. The horizontal segments indicate 95\% confidence intervals for each stress scenario estimate. The confidence intervals for copula-based methods are obtained using block bootstrap described earlier.
For the GKK estimator of stress scenarios, \cite{Glasserman_etal2015} propose to compute approximate confidence regions by shifting the confidence regions for the conditional mean obtained from the asymptotics of the empirical likelihood. The main observations we draw from this figure are in line with what we have seen earlier in the simulation studies. The copula-based estimators have much lower sampling variability compared to the GKK estimator as indicated by the width of confidence intervals. The estimators $\CM_1$ and $\CM_2$ also tend to be well calibrated, in particular, the $\CM_2$ estimator for which the fitted portfolio map function $\hat g$ returns exactly the specified threshold~$\ell$ at the stress scenario estimate. The naive estimator $\CM_3$ leads to the (fitted) portfolio loss at the stress scenario estimate $\widehat \mb(\ell)$ that is considerably above threshold~$\ell$. On the other hand, for this portfolio, the GKK estimator results in the (fitted) portfolio loss at the stress scenario estimate that is below threshold~$\ell$. This may be due to underestimation (in absolute value) of the stress scenario component corresponding to the risk factor based on the Canadian dollar (CAD), which is the main driver of portfolio losses here. Not surprisingly, under copula-based methods, the stress scenario estimates for risk factor linked to currency HKD remain nearly the same for the two threshold values, owing to the fact that this risk factor is only weakly associated with portfolio losses as indicated by its Kendall's tau correlation coefficient of just $-0.095$ (see Table~\ref{tab:A_dep}).

\begin{figure}[ht]
\centering
\subfigure{\includegraphics[scale=0.4]{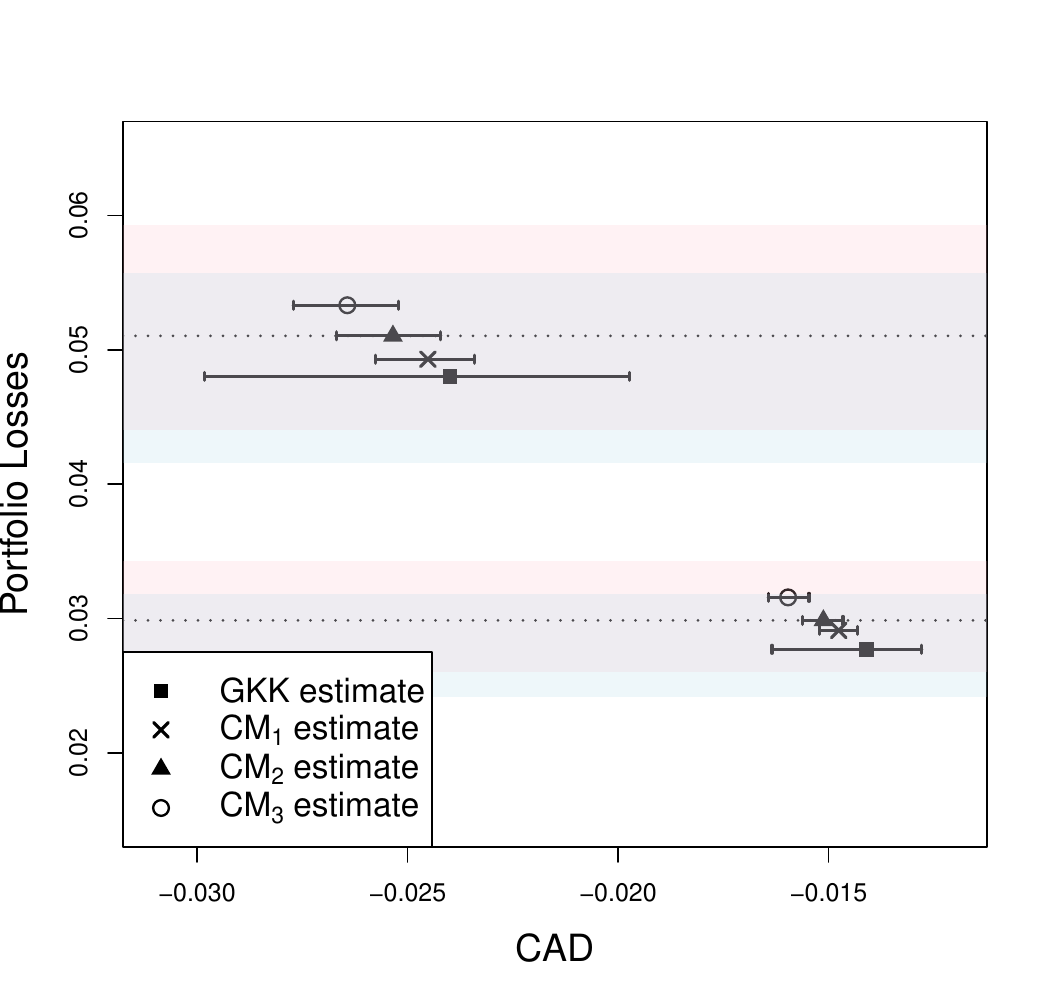}}
\subfigure{\includegraphics[scale=0.4]{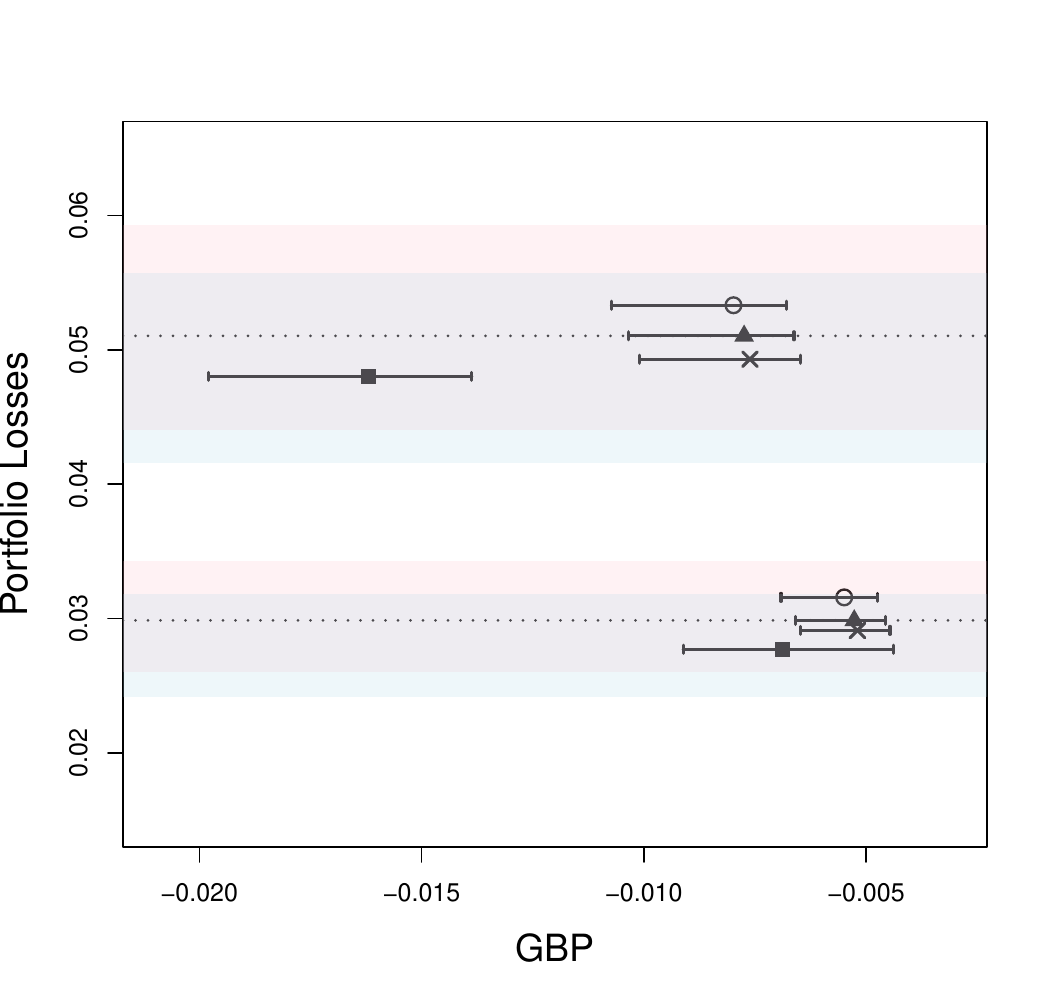}}
\subfigure{\includegraphics[scale=0.4]{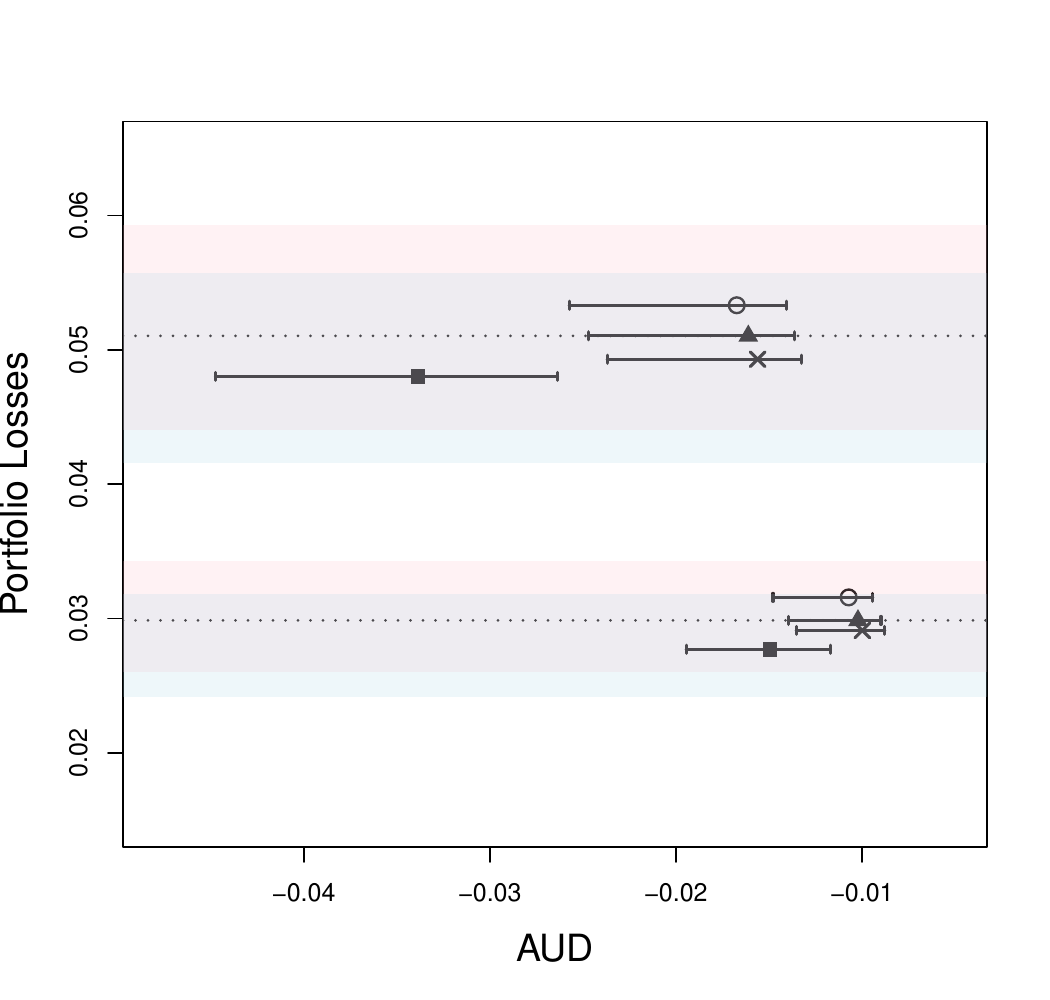}}
\subfigure{\includegraphics[scale=0.4]{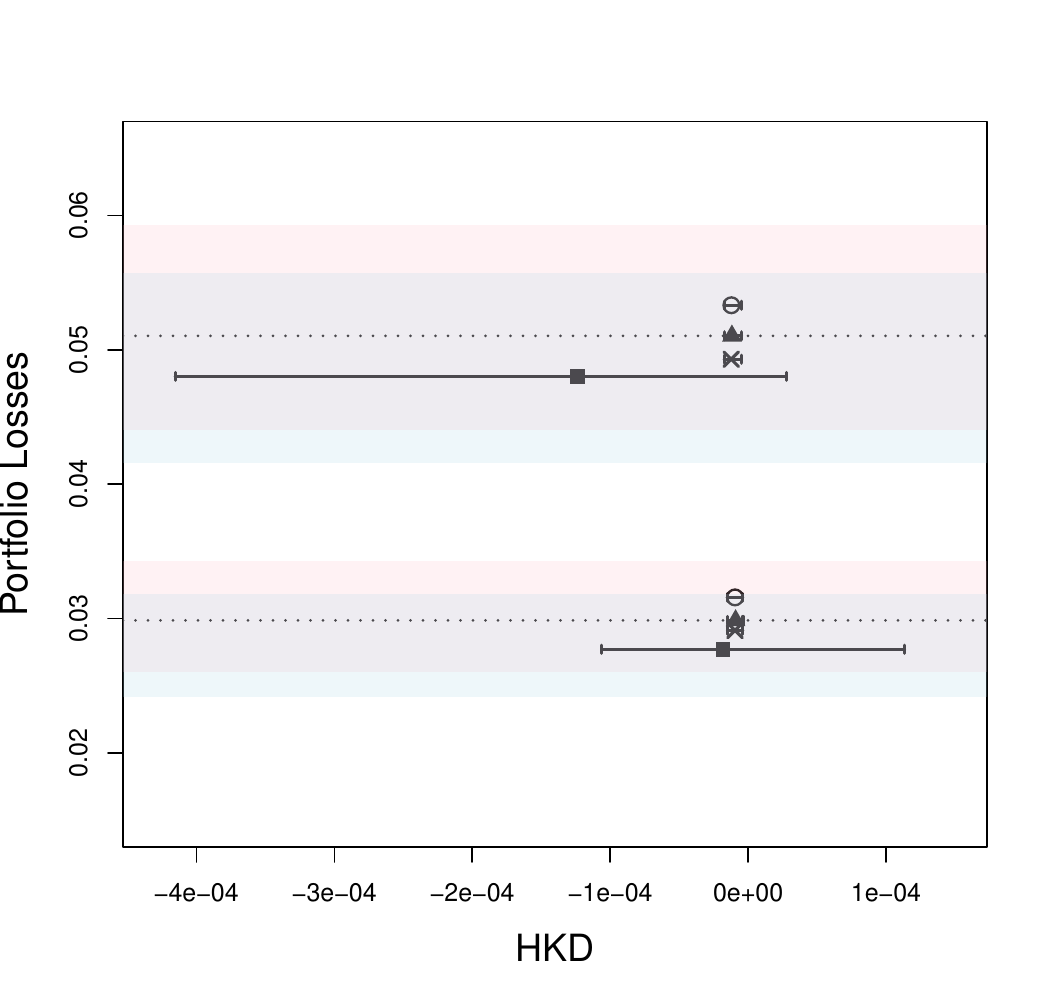}}
\caption{Marginal stress scenario estimates $\widehat m_i(\ell)$ along with 95\% confidence intervals (shown by horizontal segments) for Portfolio A. The horizontal dotted lines indicated thresholds $\ell = 0.0299$ and $\ell = 0.0510$ (corresponding to risk levels $p=0.01$ and $p=0.001$). The pink and blue shaded bands show the $95\%$ confidence intervals for predicted portfolio losses under the $\CM_2$ and GKK methods.}
\label{fig:A_CI}
\end{figure}

In addition to plotting individual risk factors against portfolio losses, we also provide in Figure~\ref{fig:A_3d} three-dimensional data clouds of risk factors together with estimates of stress scenarios. As the copula-based estimates are fairly close to each other, we only include the GKK and $\CM_2$ estimates at two threshold levels. For the lower threshold, both stress scenario estimates (indicated by the red symbols on the plots) are well within the data clouds but they are also close to each other. However, the differences in the two estimates (purple symbols) increase at the higher threshold as the stress scenarios are pushed toward the edges of the data clouds. 


\begin{figure}[ht]
\centering
\subfigure{\includegraphics[scale=0.52]{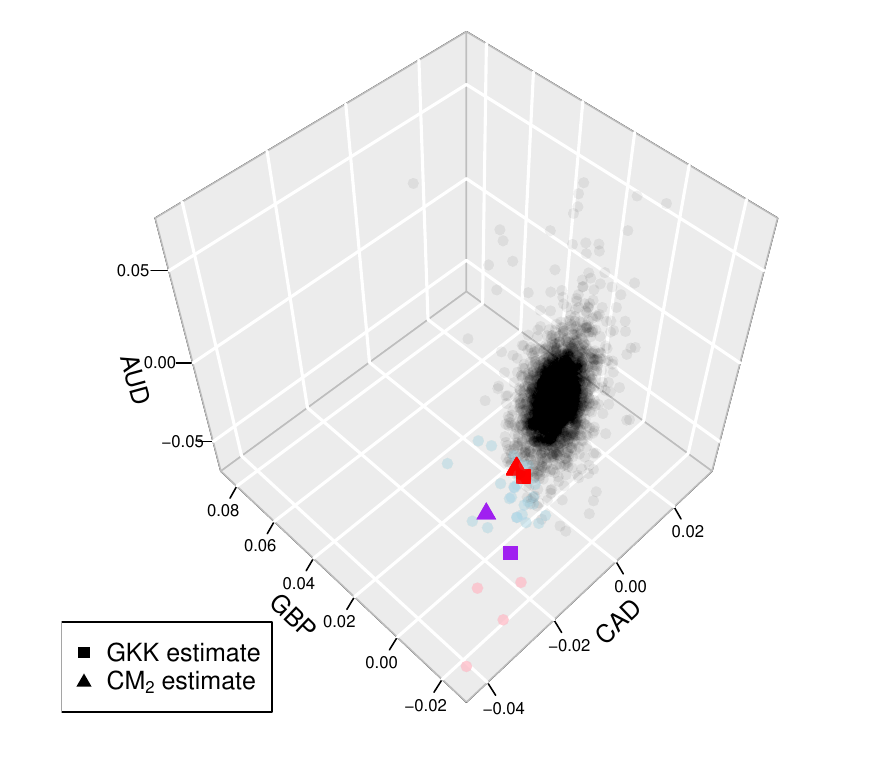}}
\subfigure{\includegraphics[scale=0.52]{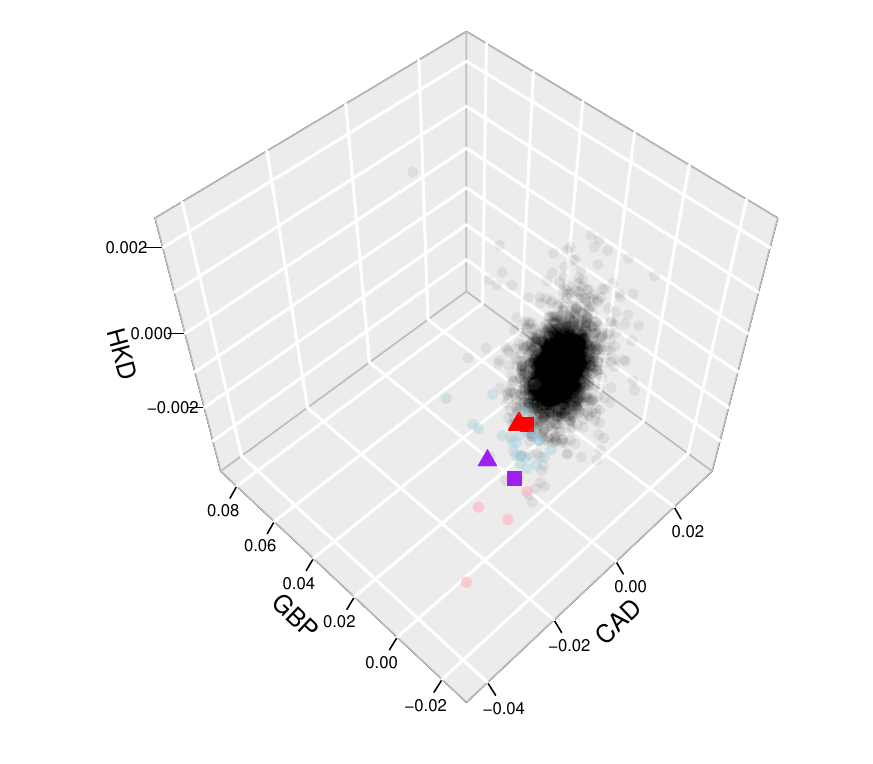}}
\subfigure{\includegraphics[scale=0.52]{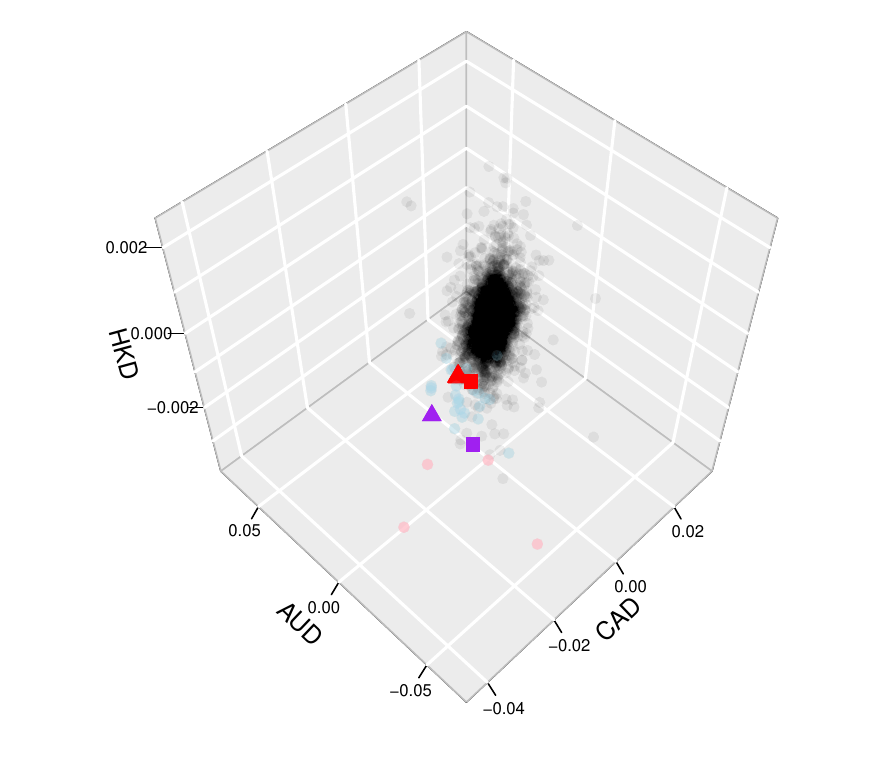}}
\subfigure{\includegraphics[scale=0.52]{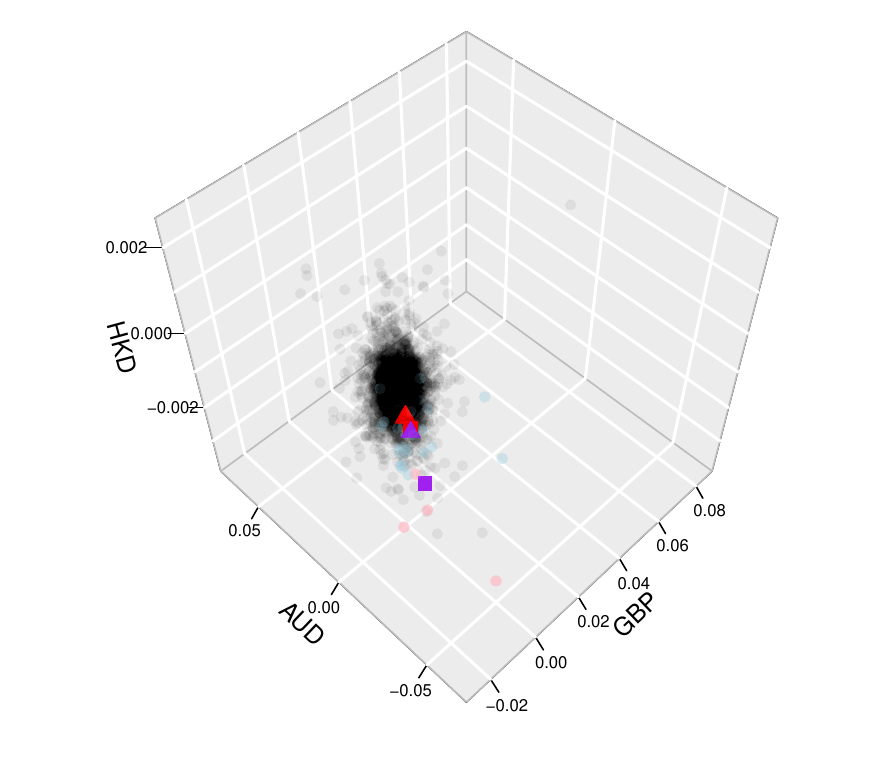}}
\caption{Three-dimensional clouds of risk factors for Portfolio A. The pink points indicate observations whose portfolio losses are greater than $\ell = 0.0510$; the lightblue points indicate observations whose portfolio losses are greater than $\ell = 0.0299$. The red triangles and squares correspond to stress scenario estimates at $\ell = 0.0299$ and the purple triangles and squares correspond to stress scenario estimates at $\ell = 0.0510$.}
\label{fig:A_3d}
\end{figure}

\subsection{Portfolio B}
\label{sec::T_portfolio}

As argued in Section~\ref{siB}, the data in portfolio~B exhibits features of a two-regime process, for which it is difficult to find a parametric model that can well capture the data structure. This prompted us to consider clustering the data, leading to one cluster with the data points for which risk factors tend to have positive association with portfolio losses and the other cluster with the points for which risk factors show negative association with portfolio losses. In our analysis, in addition to presenting stress scenario estimates for each of the two clusters, we also include estimates that ignore the two-regime nature of these data and treat them as coming from a single stationary distribution over the observation period. We report results of the stress scenario estimation for two threshold values, corresponding to the $0.99$ and $0.999$-empirical quantiles of the portfolio losses. Figure~\ref{fig:B_CI} shows stress scenario estimates at the lower threshold under the copula-based $\CM_2$ method and the GKK method. The results at the higher threshold are provided in~\ref{appen::B_CI_s}. The confidence bands for the fitted portfolio loss evaluated at the stress scenario estimates are computed using all observations. 

\begin{figure}[H]
\centering
\subfigure{\includegraphics[scale=0.4]{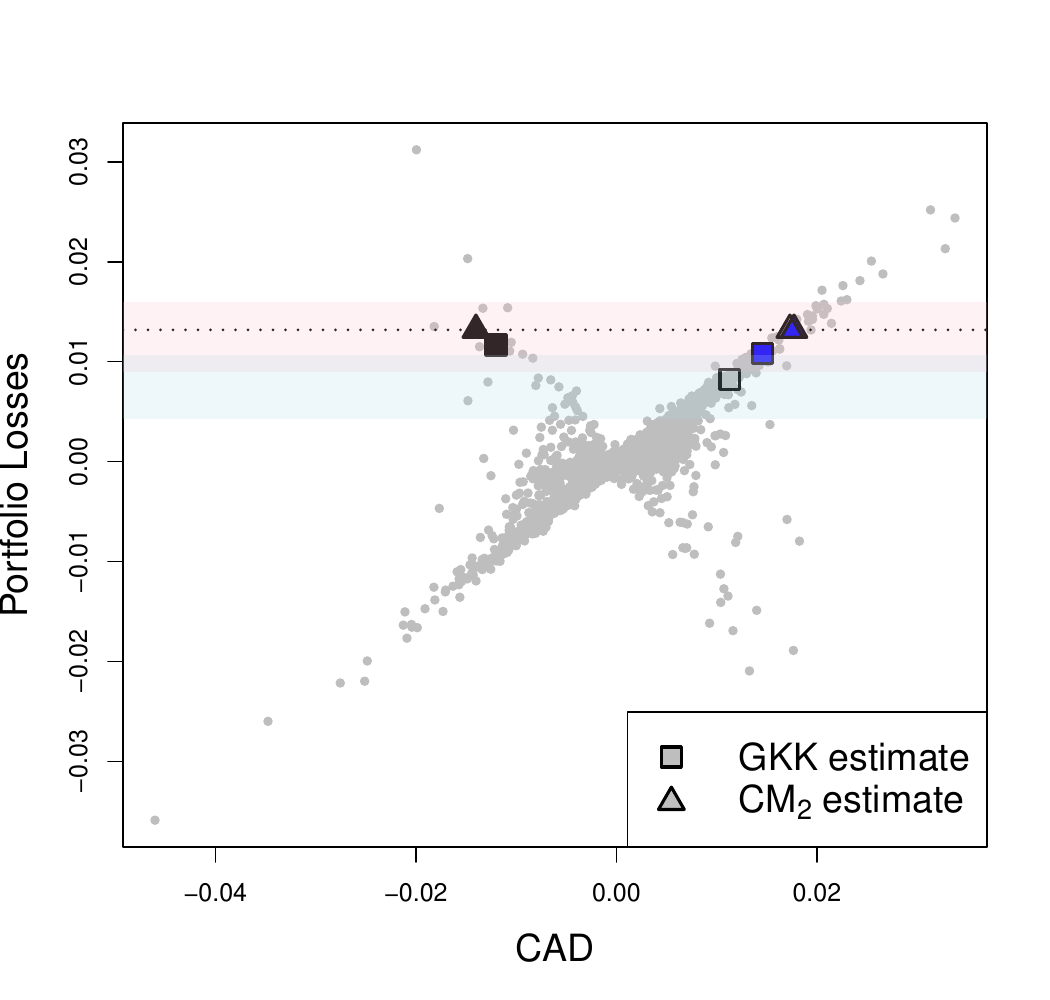}}
\subfigure{\includegraphics[scale=0.4]{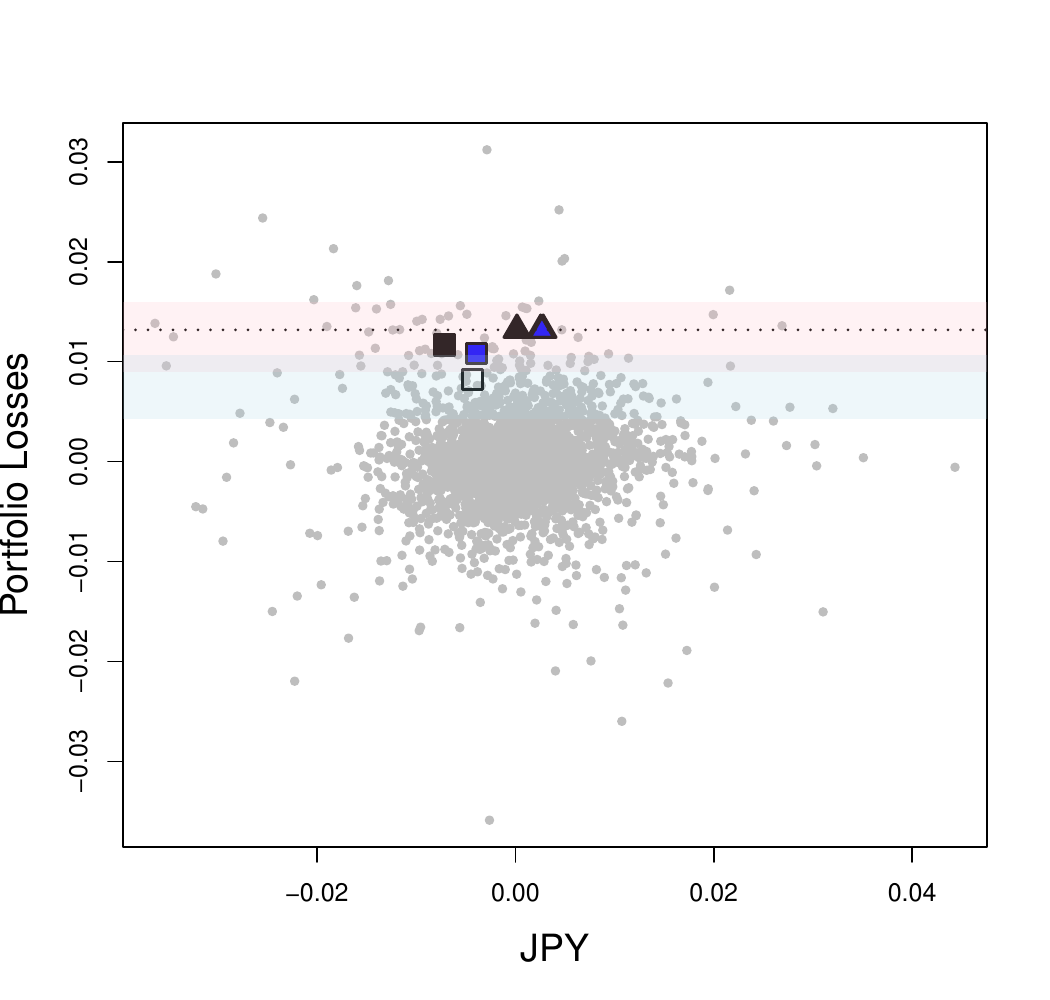}}
\subfigure{\includegraphics[scale=0.4]{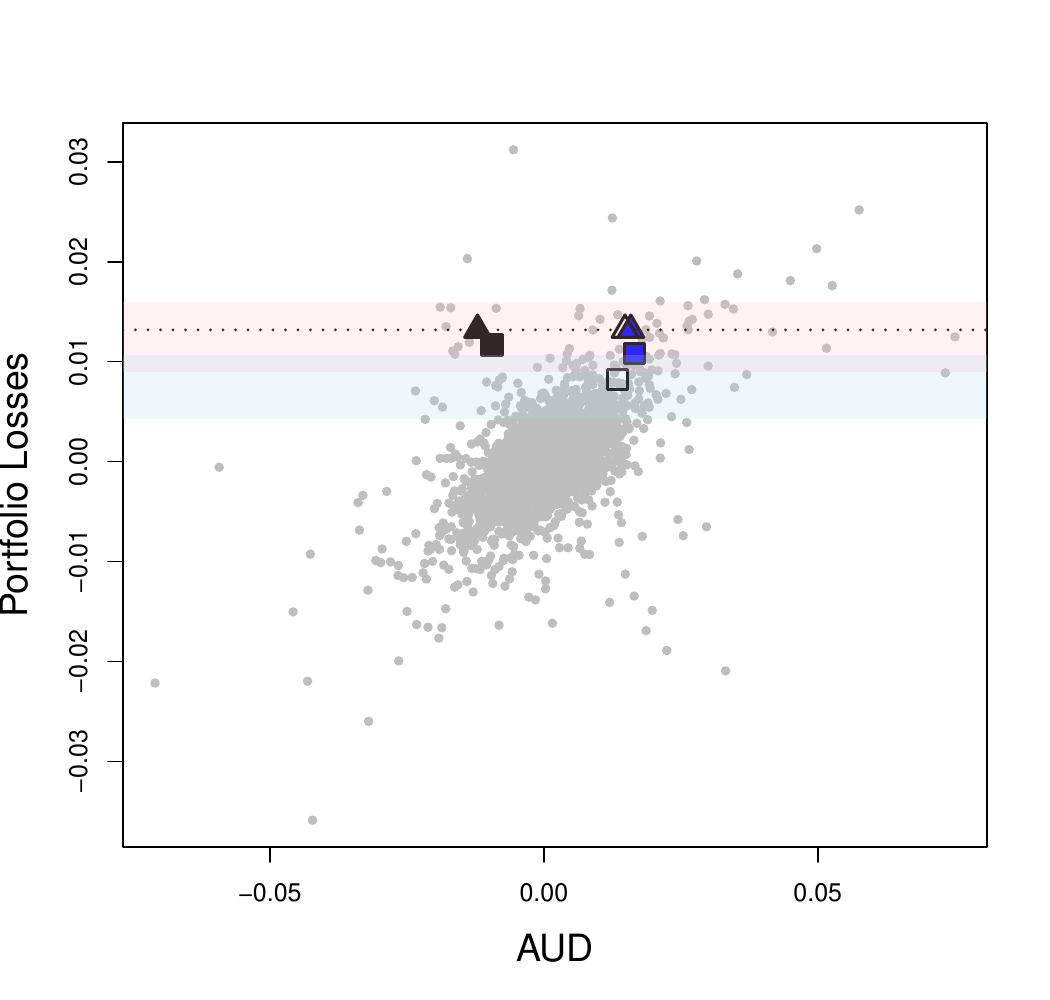}}
\subfigure{\includegraphics[scale=0.4]{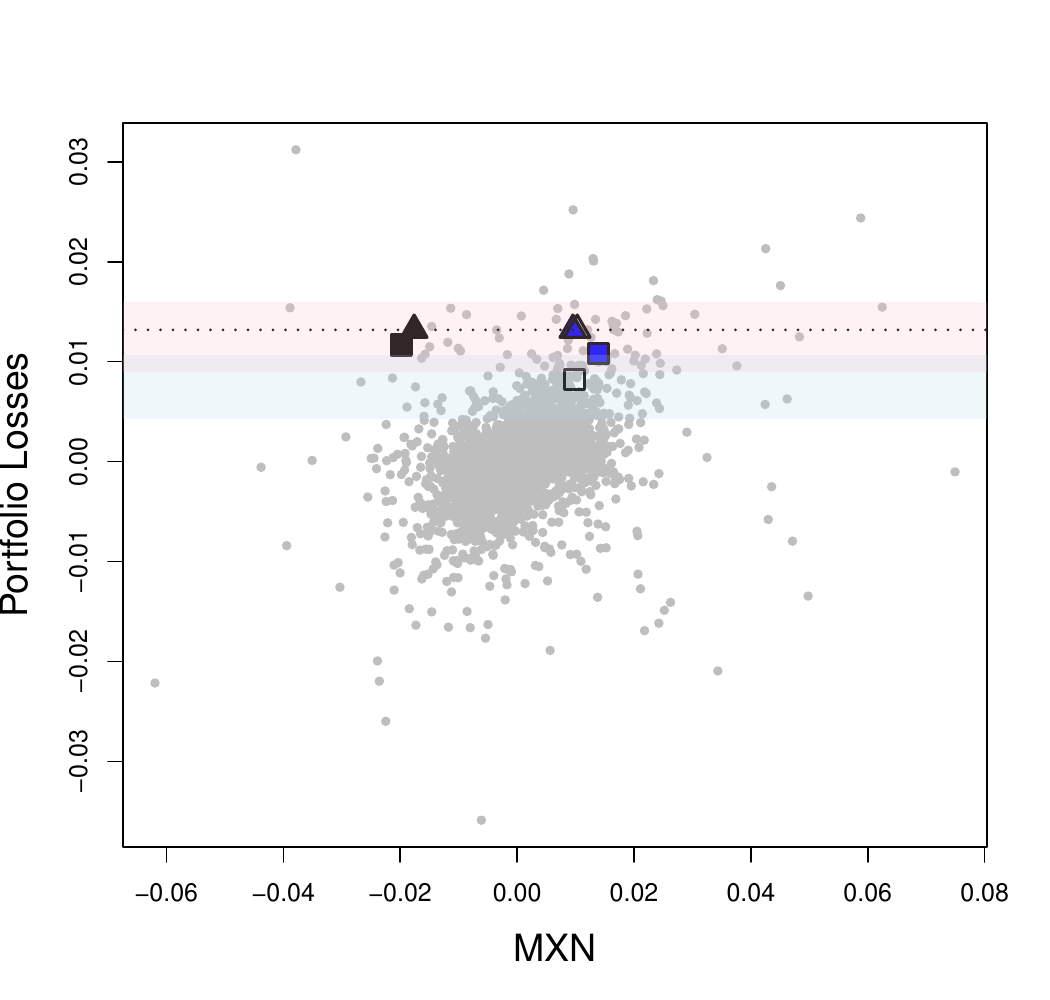}}
\subfigure{\includegraphics[scale=0.4]{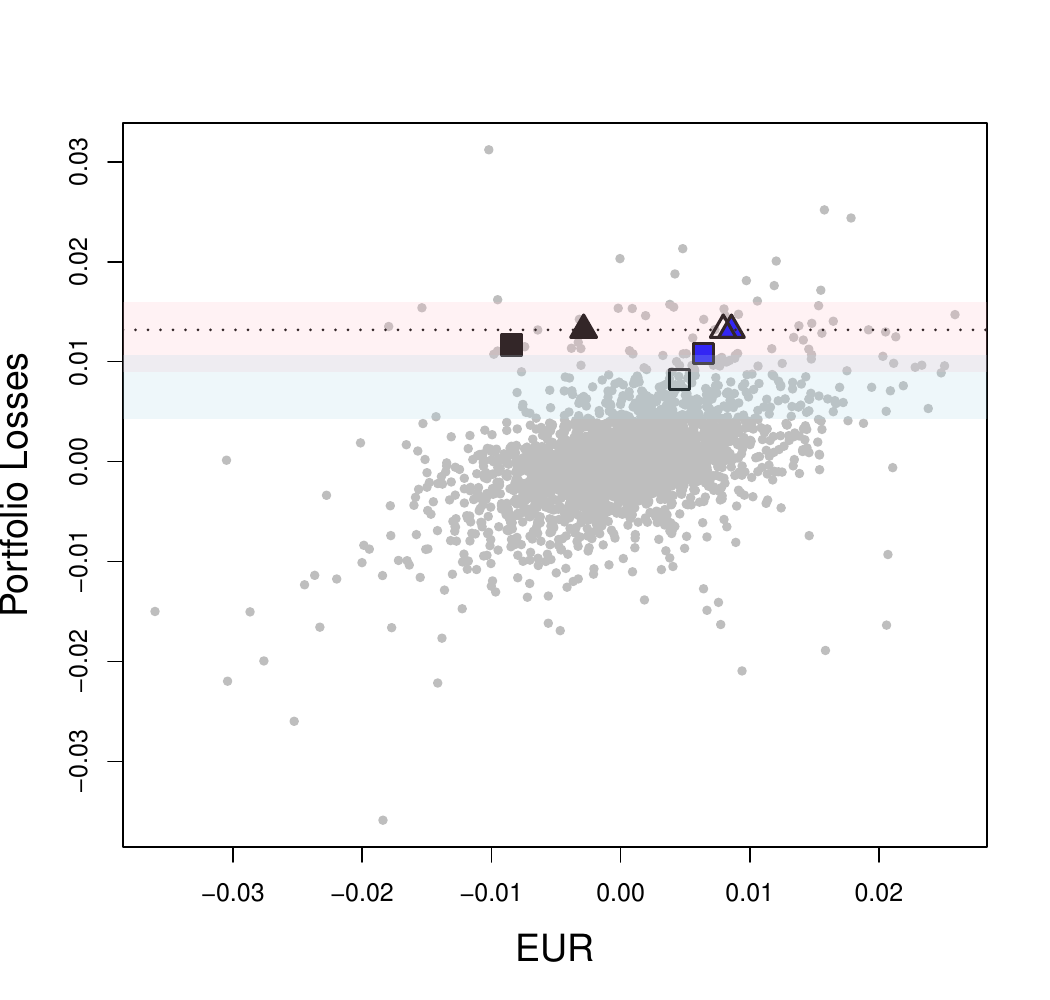}}
\caption{Plots of stress scenario estimates together with the data points for Portfolio B. The horizontal dotted lines indicate the threshold, $\ell = 0.0132$. The blue, black and unfilled symbols show, respectively, stress scenario estimates for cluster 1, cluster 2 and all of the data combined. The pink and blue shaded bands show the $95\%$ confidence intervals for predicted portfolio losses for the $\CM_2$ and GKK stress scenario estimates with all observations included.}
\label{fig:B_CI}
\end{figure}

From these plots, we see that the estimated stress scenarios for cluster~1 tend to fall along the main diagonal for the risk factors having strong positive association with portfolio losses, in particular, for CAD and AUD and to some extent for EUR. For cluster~2, in contrast, the stress scenario estimates lie close to the counter-diagonal, signalling negative association with portfolio losses.  

When we look at the stress scenarios estimated using all observations (shown with non-filled symbols), we can observe that for the copula-based methods, the estimates will be close to those of cluster~1. This is not surprising as cluster~1 data dominates the sample and hence has the most influence on pair-copula choices and their maximum likelihood parameter estimates. This can also be confirmed by comparing the contour plots of fitted pair copulas in Figure~\ref{fig:B_contour} for cluster~1 alone (left panel) and for all observations combined (right panel). We can see that using all observations leads to a similar copula model as that using observations in cluster 1. However, using all observations as opposed to treating the two clusters of data separately has a different effect on 
the GKK stress scenario estimates. The combined stress scenario estimates fall in between the two cluster-based estimates (although much closer to cluster~1 values) due to averaging of risk factor values of different signs from different clusters. This in turn leads to the associated (fitted) portfolio loss to fall even further below the specified threshold~$\ell$. In fact, the 95\% confidence band for this fitted portfolio value at the estimated combined scenario does not contain the threshold value.

\begin{figure}[ht]
\centering
\subfigure[Cluster 1]{\includegraphics[scale=0.45]{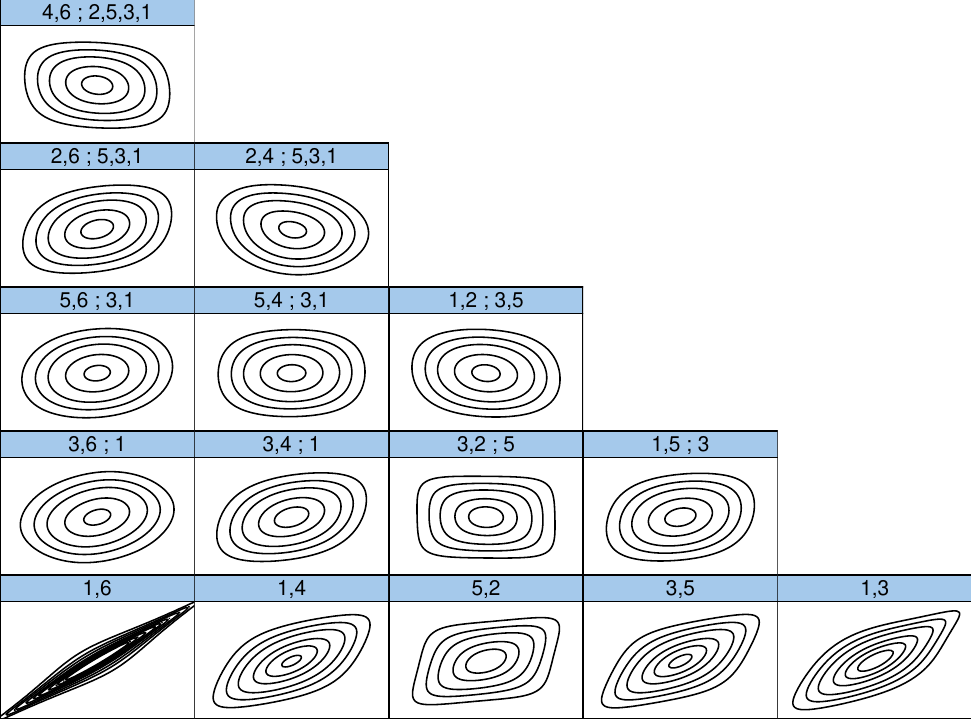}}
\subfigure[All observations]{\includegraphics[scale=0.45]{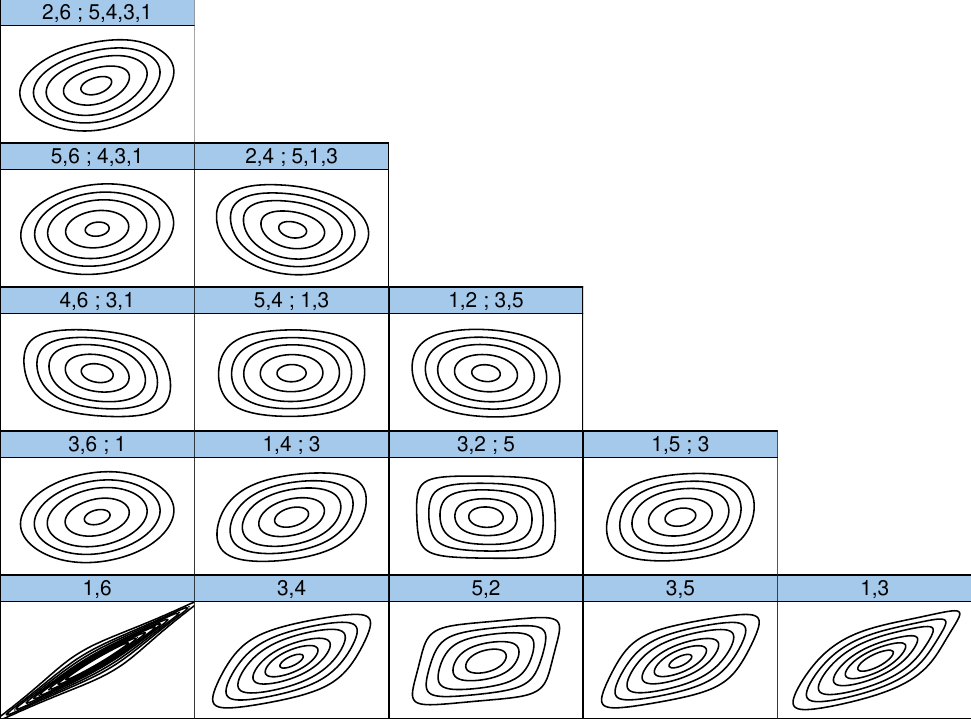}}
\caption{Contour plots of fitted pair copulas for Portfolio B (from top to bottom: $\TC_5, \TC_4, \TC_3, \TC_2, \TC_1$). Risk factor variables based on currencies CAD, JPY, AUD, MXN, EUR and portfolio loss $L$ are abbreviated as 1,2,3,4,5 and 6, respectively.}
\label{fig:B_contour}
\end{figure}

For further illustration, in Figure~\ref{fig:T_3d} and~\ref{appen::B_3d}, we show three-dimensional plots of data clouds and stress scenario estimates corresponding to pairs of risk factors and the associated portfolio losses. From this perspective, we can see that there is a better agreement between the two considered methods for cluster~1 that contains more data points. However, the difference in stress scenario estimates under the two methods is more pronounced for cluster~2 with fewer data points and also for the higher threshold. 

This portfolio serves as a nice illustration of situations where data are far away from the assumption of elliptical symmetry. The idea of clustering the data prior to scenario estimation to identify multiple regimes serves at least two purposes. First of all, it is meaningful from the applied perspective to consider multiple stress scenarios under each regime rather than a single scenario. In fact, one may also consider specifying different stress thresholds in each regime to match severity of losses under each regime. Secondly, from the statistical modelling and estimation perspective, clustering would allow to make use of methods developed under more restrictive assumptions (e.g., assuming an elliptical distribution for risk factors).  

\begin{figure}[ht]
\centering
\subfigure{\includegraphics[scale=0.5]{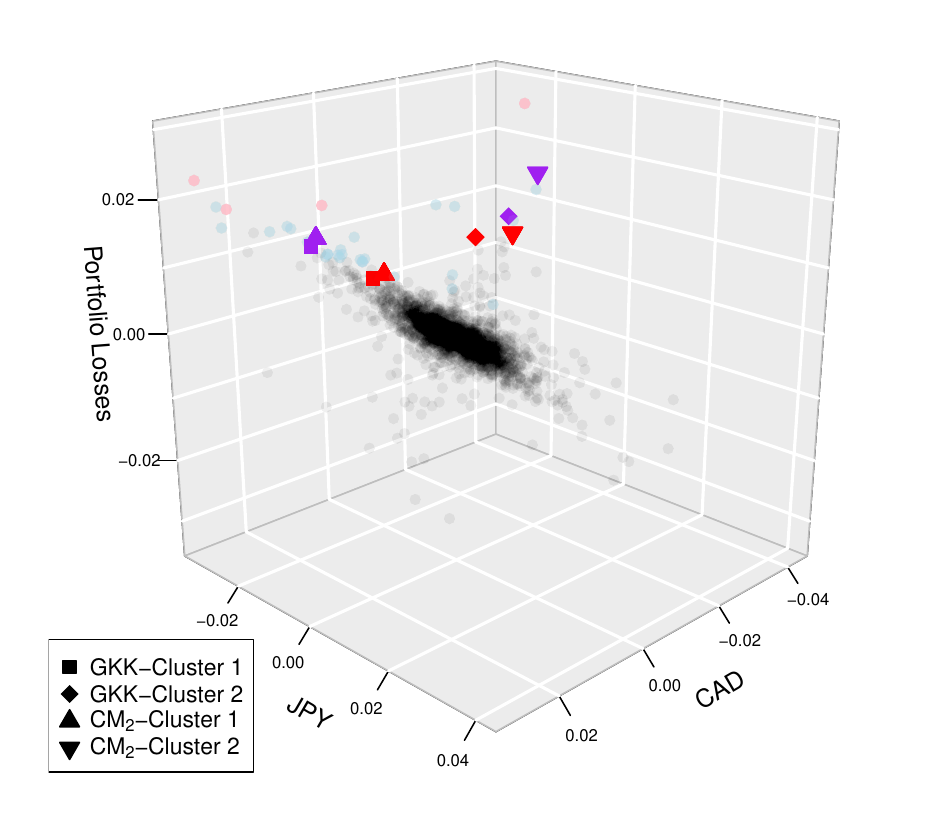}}
\subfigure{\includegraphics[scale=0.5]{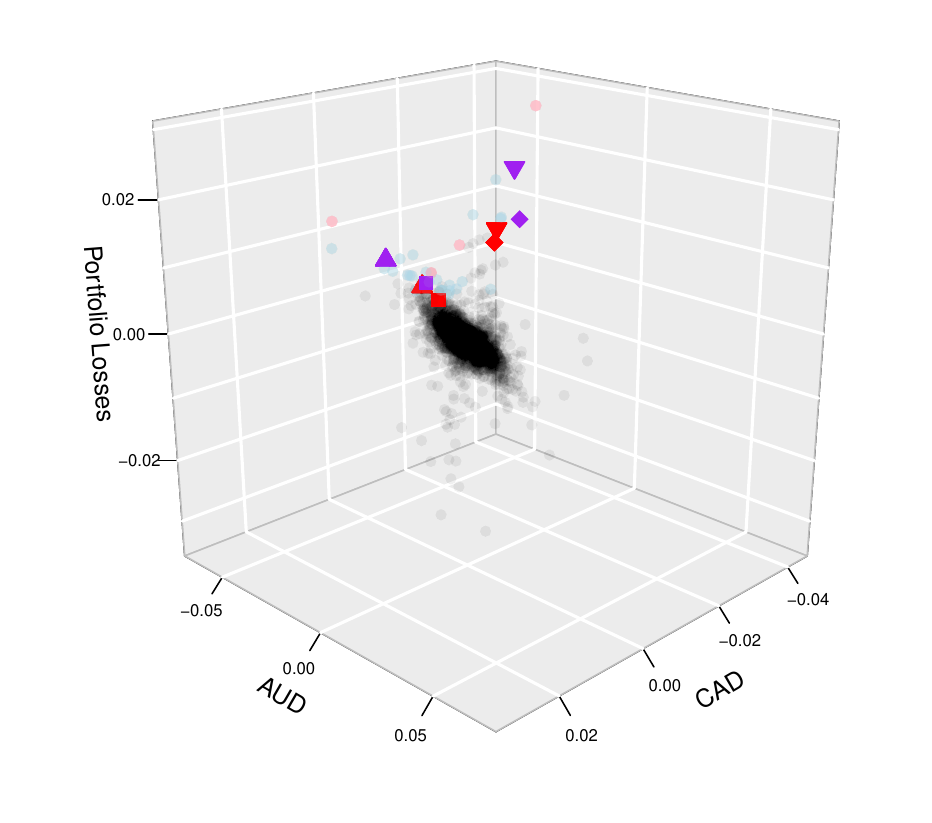}}
\caption{Plots of pairs of risk factors and associated portfolio losses for Portfolio B. The light blue and pink points indicate observations for which portfolio losses are greater than $\ell = 0.0132$ and $\ell = 0.0212$, respectively. The red and purple triangles/squares show stress scenario estimates at thresholds $\ell = 0.0132$ and $\ell = 0.0212$, respectively.}
\label{fig:T_3d}
\end{figure}

\subsection{Portfolio C}

Figure~\ref{fig:P_T1} illustrates the first tree of the R-vine selected for the data in portfolio~C. The portfolio loss variable~$L$ is only connected with the risk factor based on currency MXN; this is consistent with this risk factor having the strongest dependence with portfolio losses as supported by the value of Kendall's tau in Table~\ref{tab:P_dep}. The vine graph in this tree also tends to connect risk factors associated with currencies that are connected geographically or economically (e.g., EUR, GBP, CHF and HUF for Europe; NOK, SEK and DKK for Scandinavian countries; THB, HKD and SGD for Asia).

For this portfolio we only evaluate the $\CM_2$ estimator among the copula-based methods due to its better calibration properties; its values are also close to those based on the $\CM_1$ estimator.


\begin{figure}[ht]
\centering
\includegraphics[scale = 0.5]{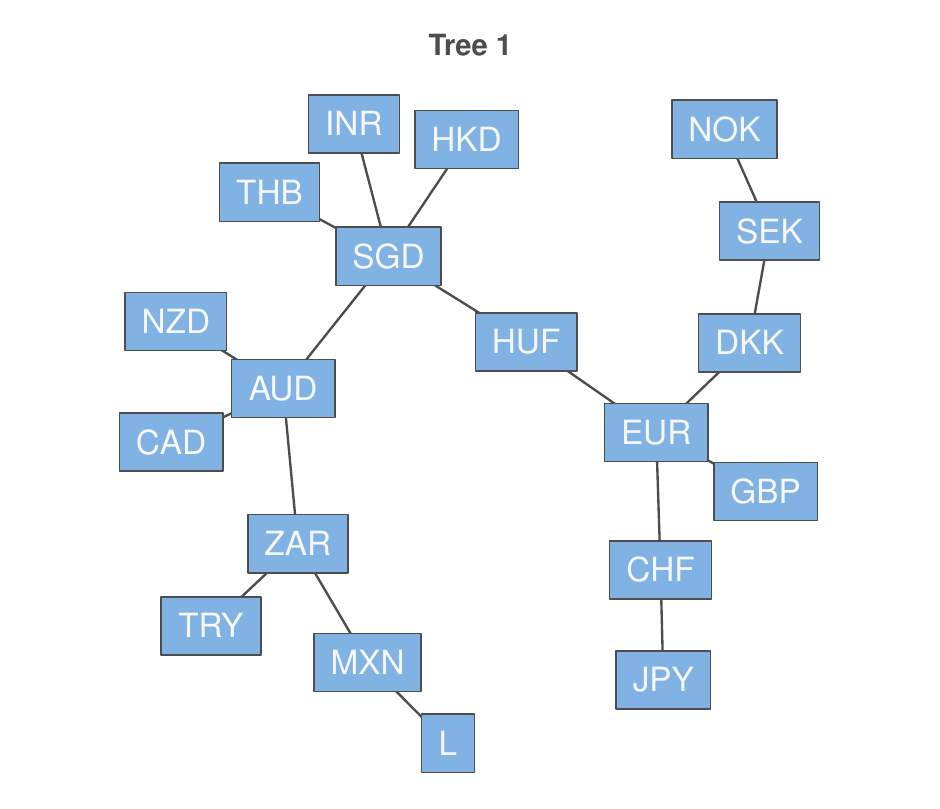}
\caption{The first tree in the R-vine selected for the portfolio C data.}
\label{fig:P_T1}
\end{figure}

We should note that, in this portfolio, we have 18 risk factors, which will result in an 18-dimensional and potentially non-convex objective function for the optimization procedure and increase the computational complexity. This problem is more severe for $\CM_2$, whose conditional density involves $\hat{g}$ from vine regression. However, exploring optimization algorithms is not the focus of this work. In our implementation, we use R package \texttt{DEoptim} (\cite{DEoptim}) to do the optimization. In order to obtain the global maximum, we consider 10 sets of initial populations and 4000 iterations. The final stress scenario estimate is obtained by selecting the one that maximizes the corresponding conditional density function among the 10 estimates resulting from different starting values. Here we only estimate stress scenarios at only one threshold value equal to the $0.99$-empirical quantiles of the portfolio losses. Figure~\ref{fig:C_est} displays the $\CM_2$ stress scenario estimates against the fitted value of the portfolio loss. The GKK stress scenario estimates are also included for comparison.


For this portfolio, the marginal components of the GKK stress scenario estimates are larger in absolute value than the copula-based estimates. This difference tends to be more pronounced for risk factors that show moderate dependence and tail dependent  with portfolio losses (refer to Table~\ref{tab:P_dep}) and for which there are a few data points with portfolio losses just above the threshold yet risk factors take on extreme extreme values (e.g., GBP and NOK). Due to inherent non-robustness of the sample mean, the GKK stress scenario estimates are pulled towards these risk factor extremes. This then results in a fitted portfolio loss at the stress scenario to be above the threshold value. For risk factors showing either very strong or very weak dependence with portfolio losses, the copula-based and GKK stress scenario estimates are fairly close to each other. This is not surprising for these data as overall they do not appear to deviate substantially from the assumption of elliptical symmetry.   

\begin{figure}[H]
\centering
\subfigure{\includegraphics[scale=0.26]{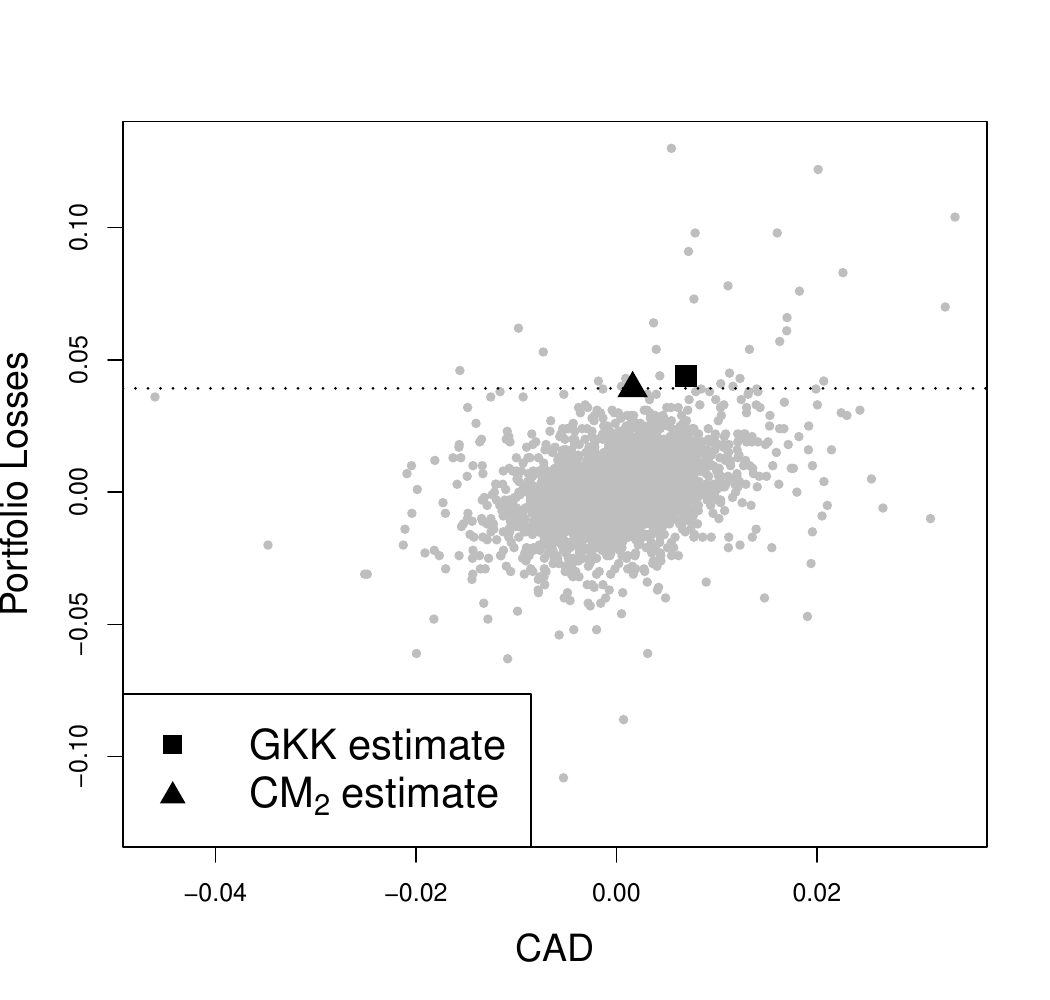}}
\subfigure{\includegraphics[scale=0.26]{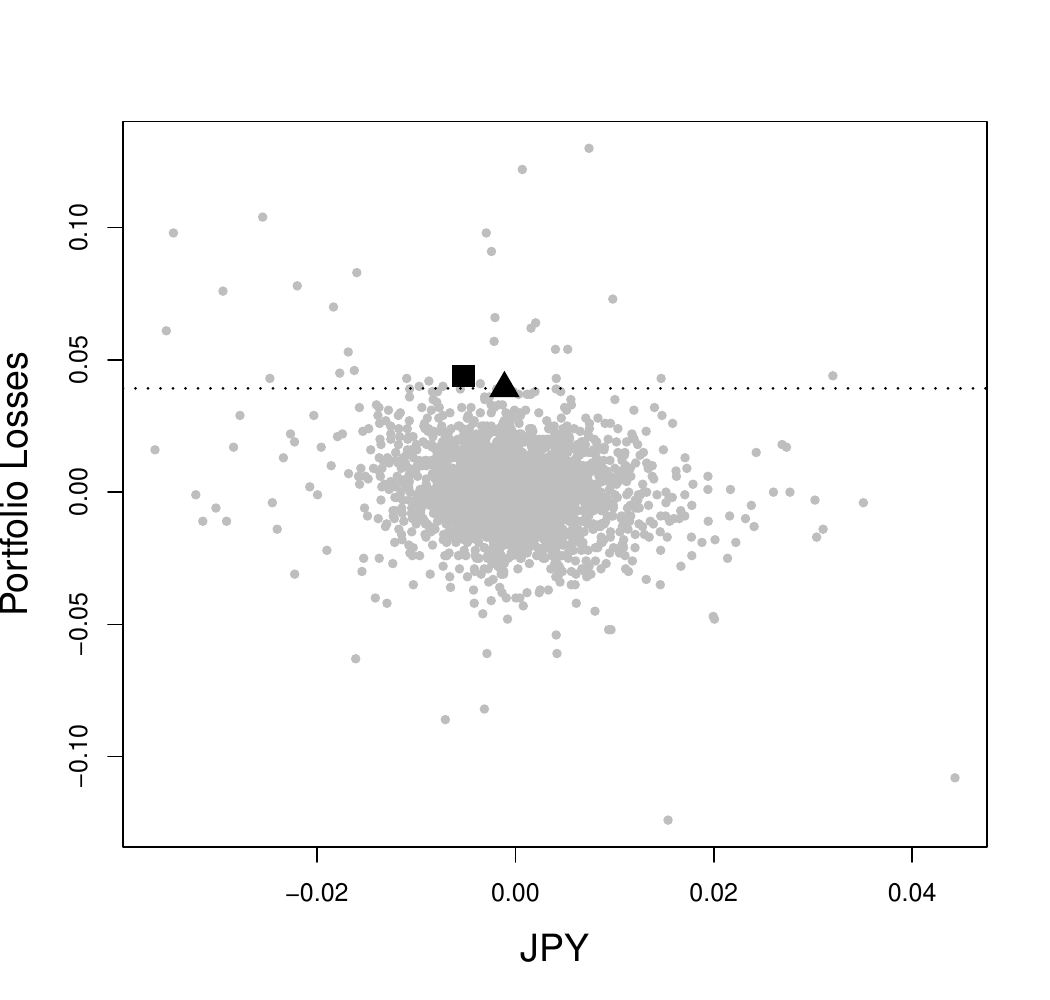}}
\subfigure{\includegraphics[scale=0.26]{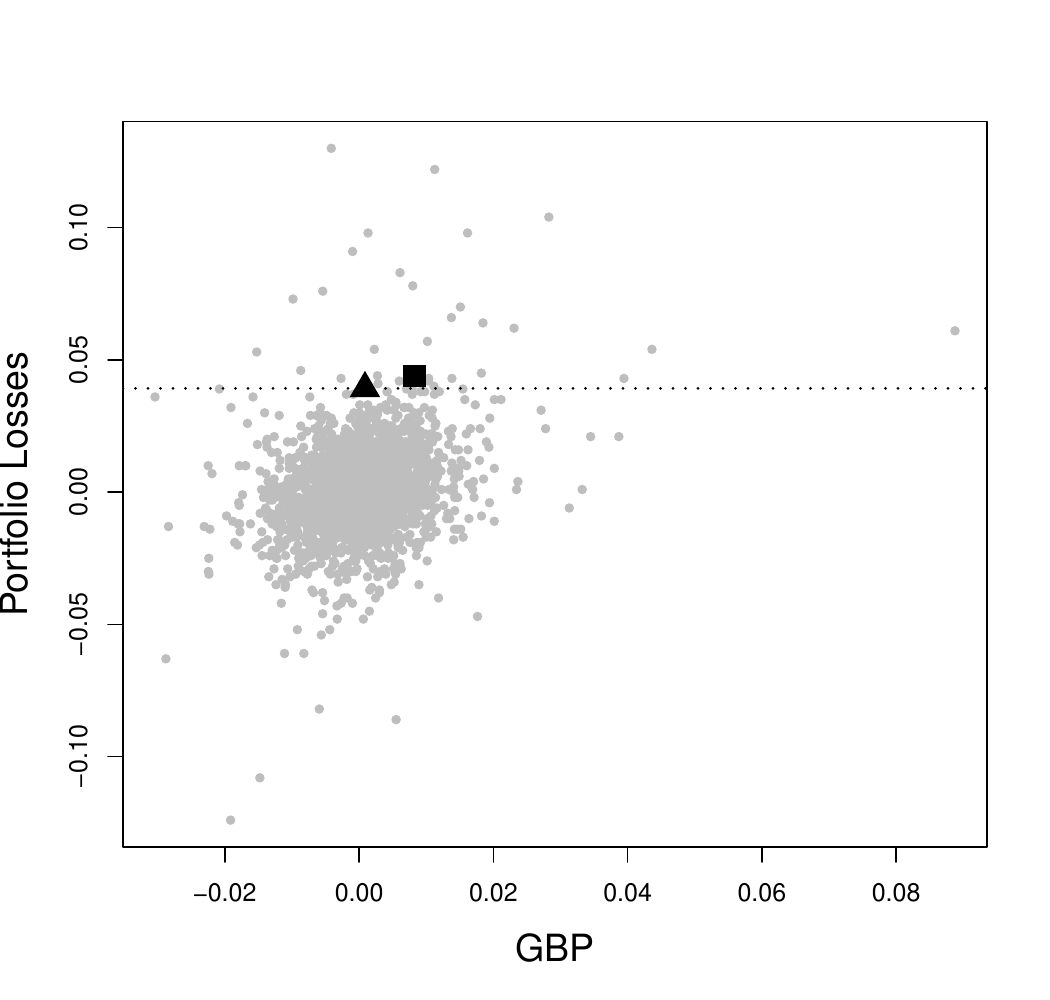}}
\subfigure{\includegraphics[scale=0.26]{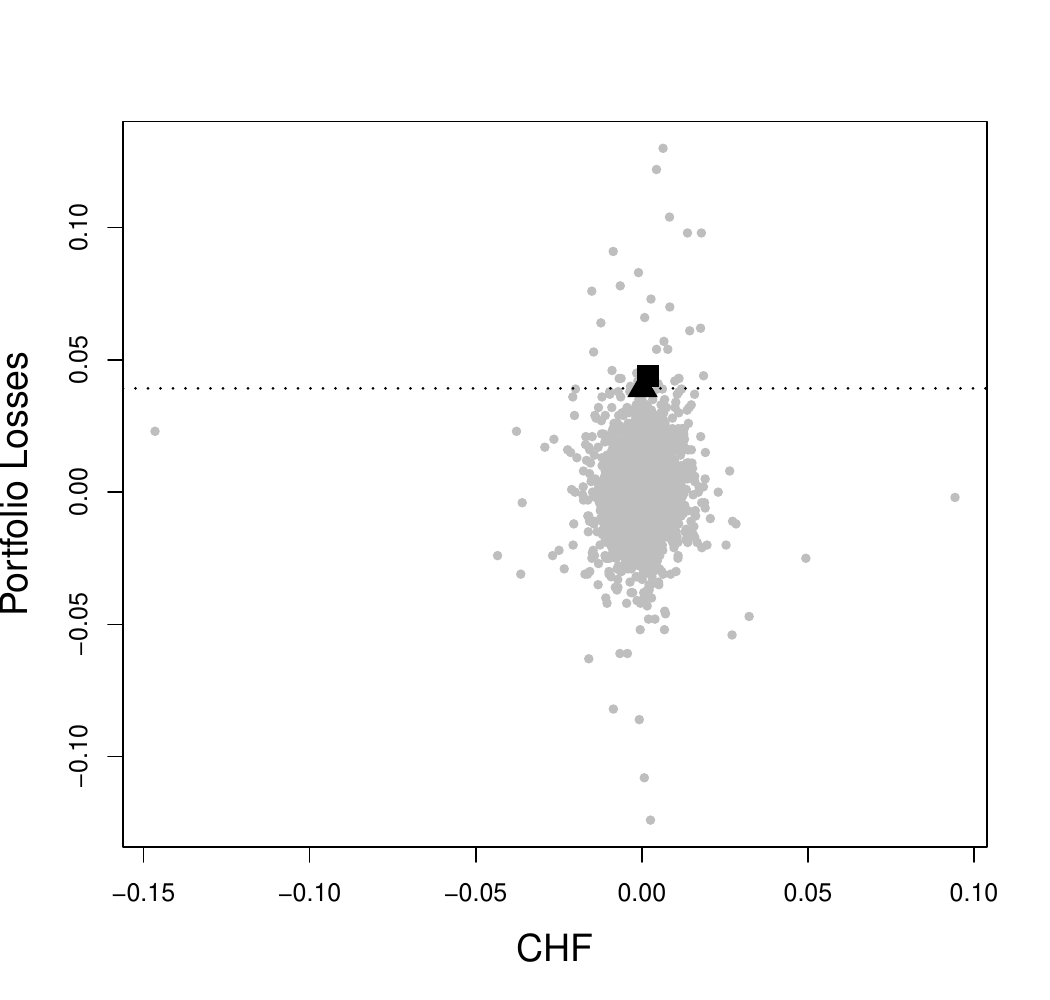}}
\subfigure{\includegraphics[scale=0.26]{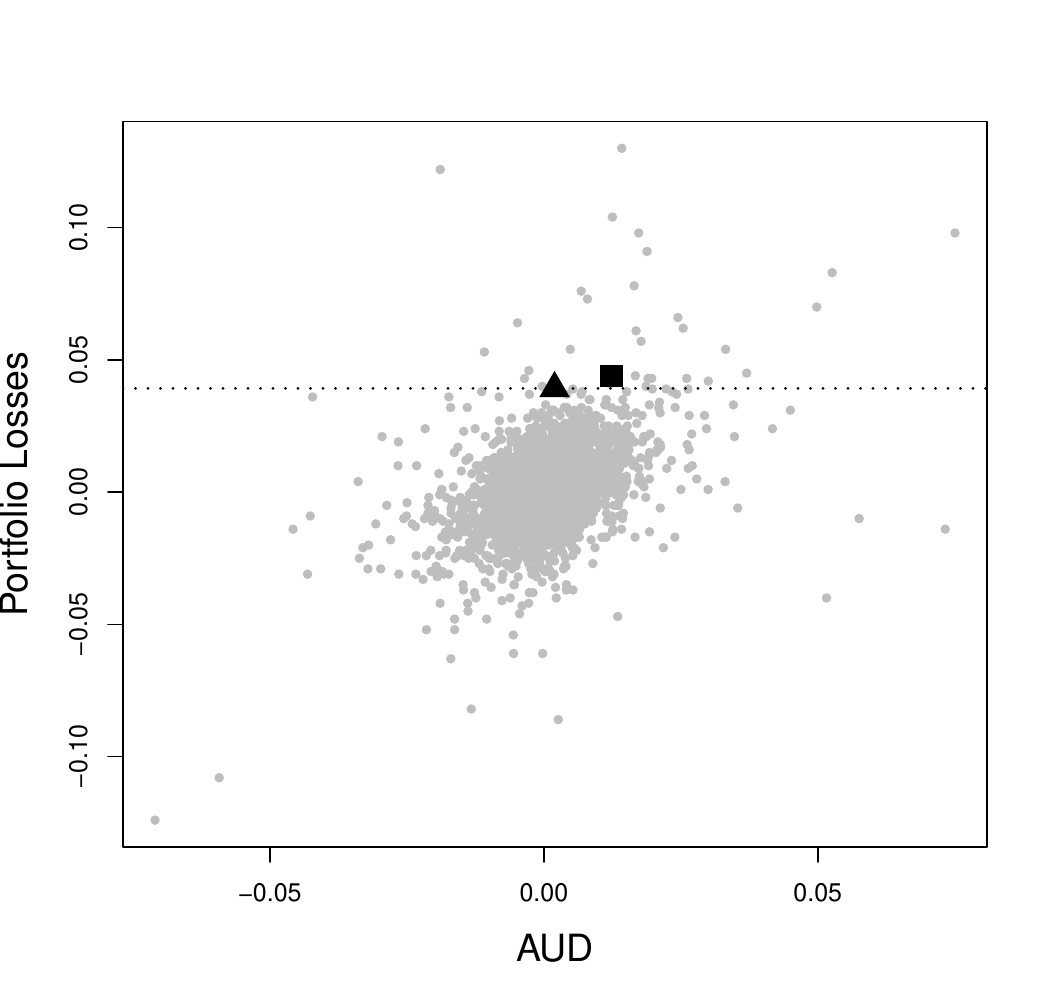}}
\subfigure{\includegraphics[scale=0.26]{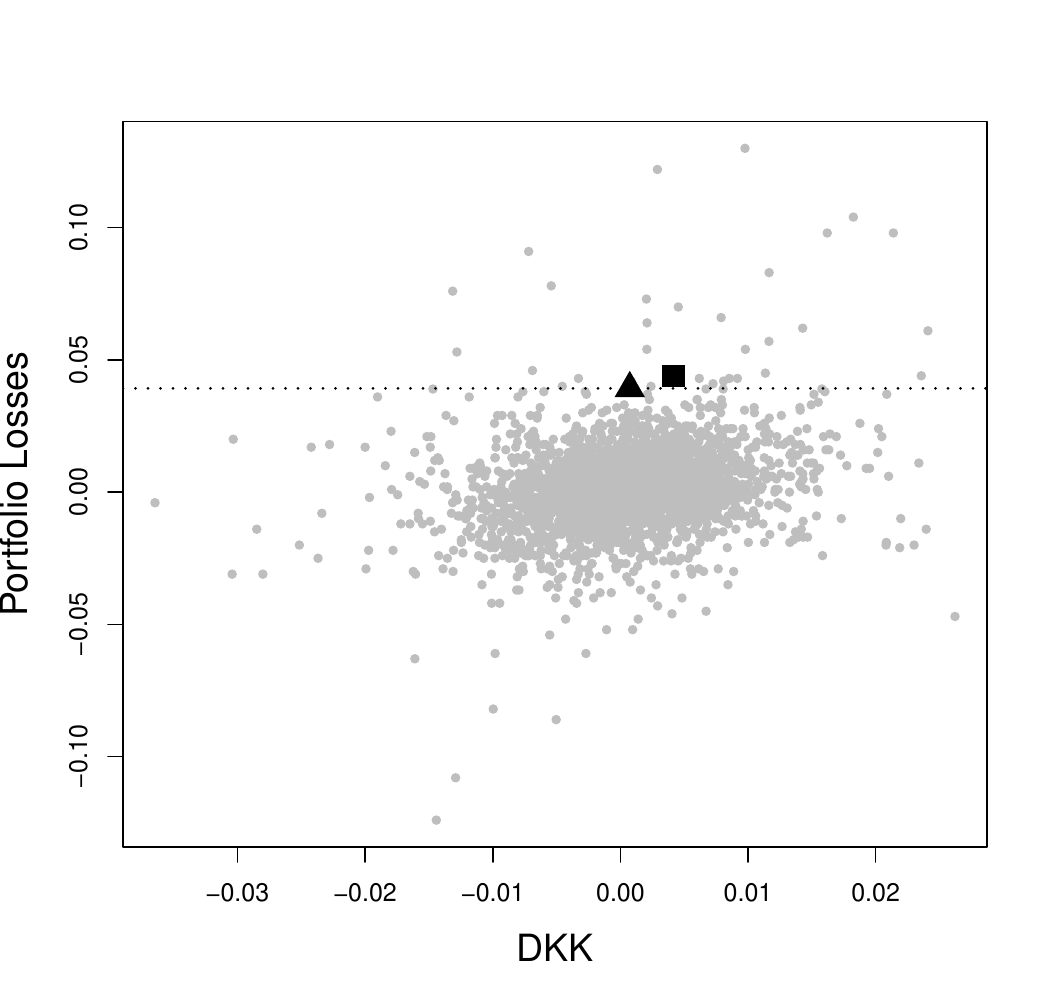}}
\subfigure{\includegraphics[scale=0.26]{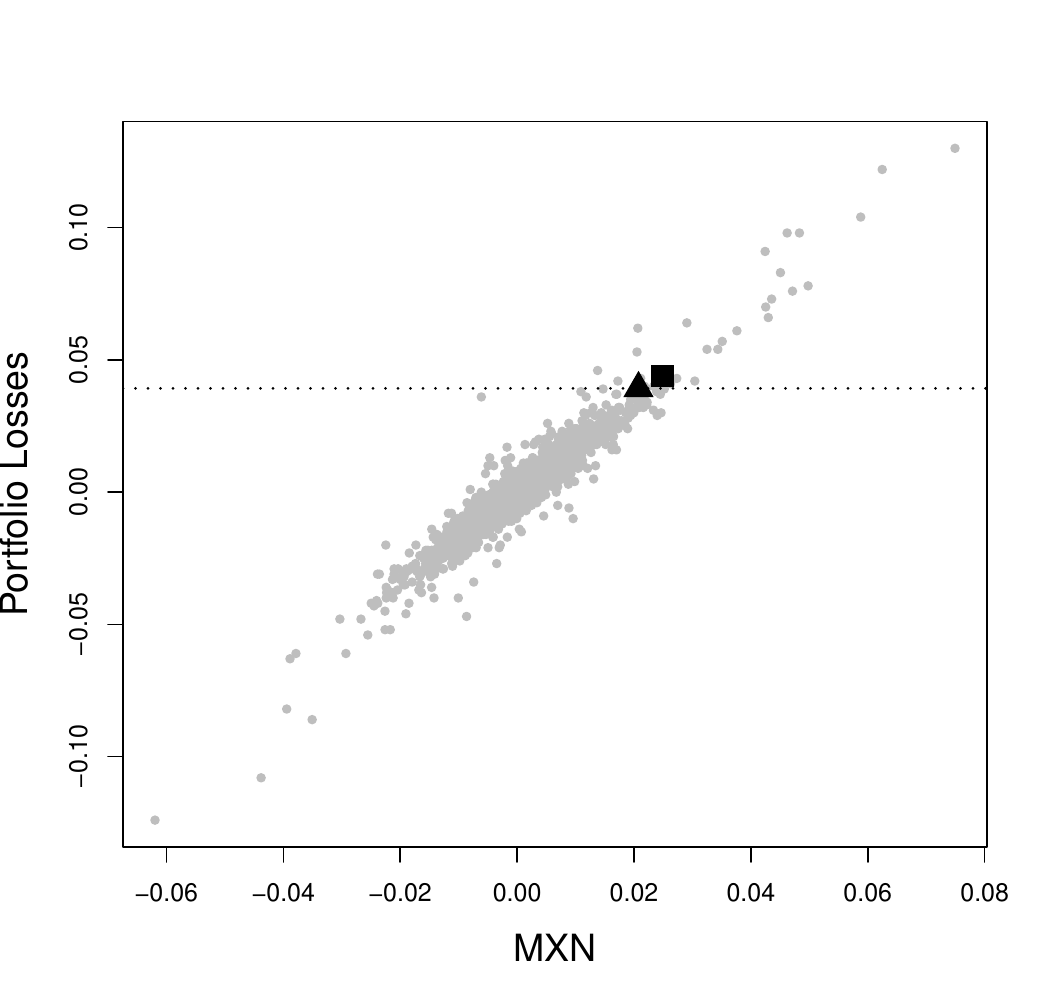}}
\subfigure{\includegraphics[scale=0.26]{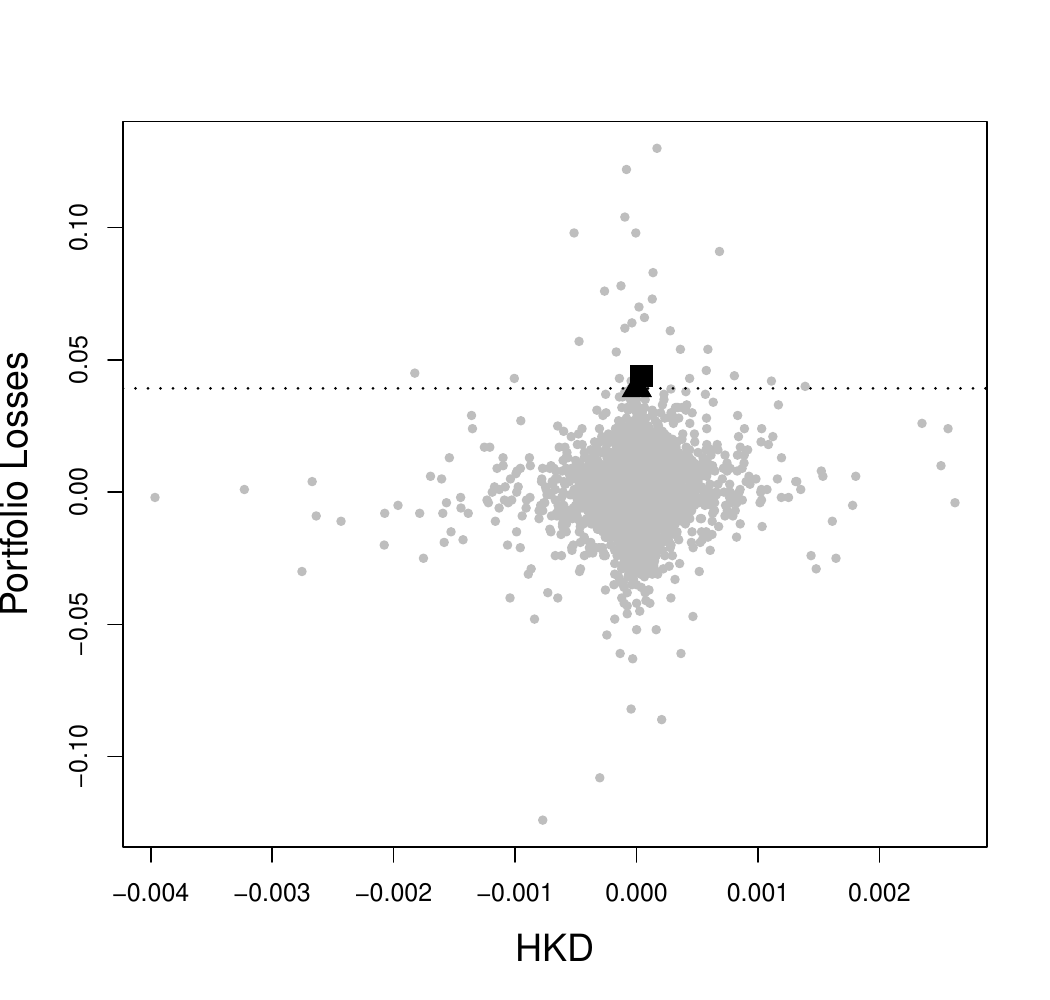}}
\subfigure{\includegraphics[scale=0.26]{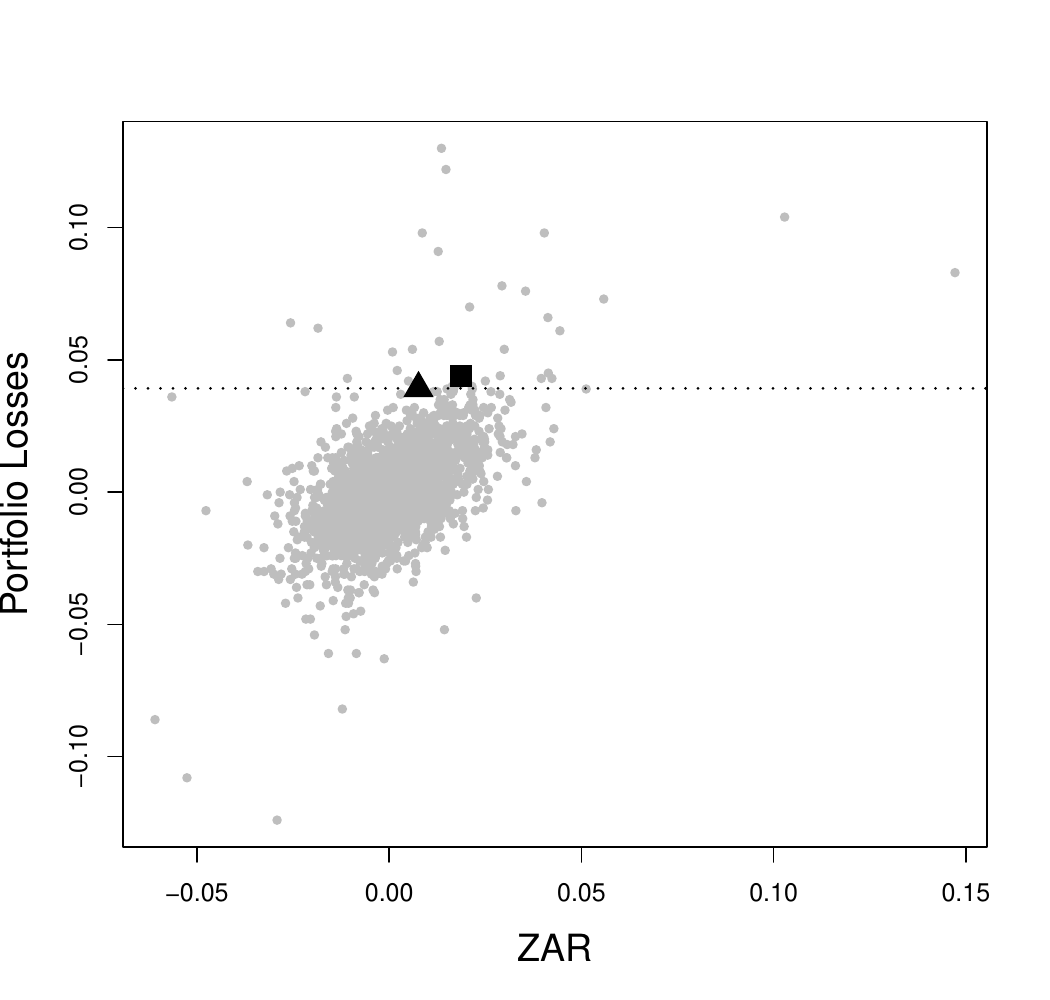}}
\end{figure}
\begin{figure}[H]
\centering
\subfigure{\includegraphics[scale=0.26]{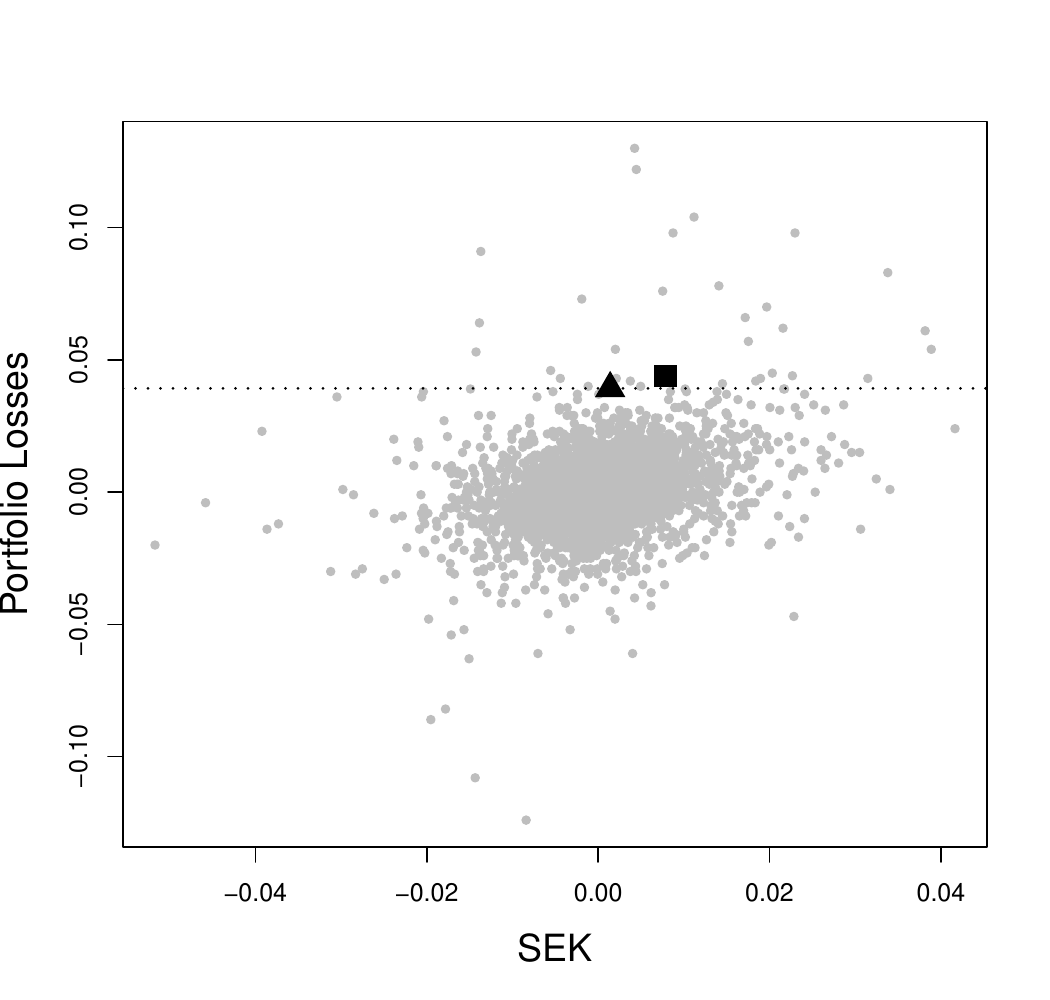}}
\subfigure{\includegraphics[scale=0.26]{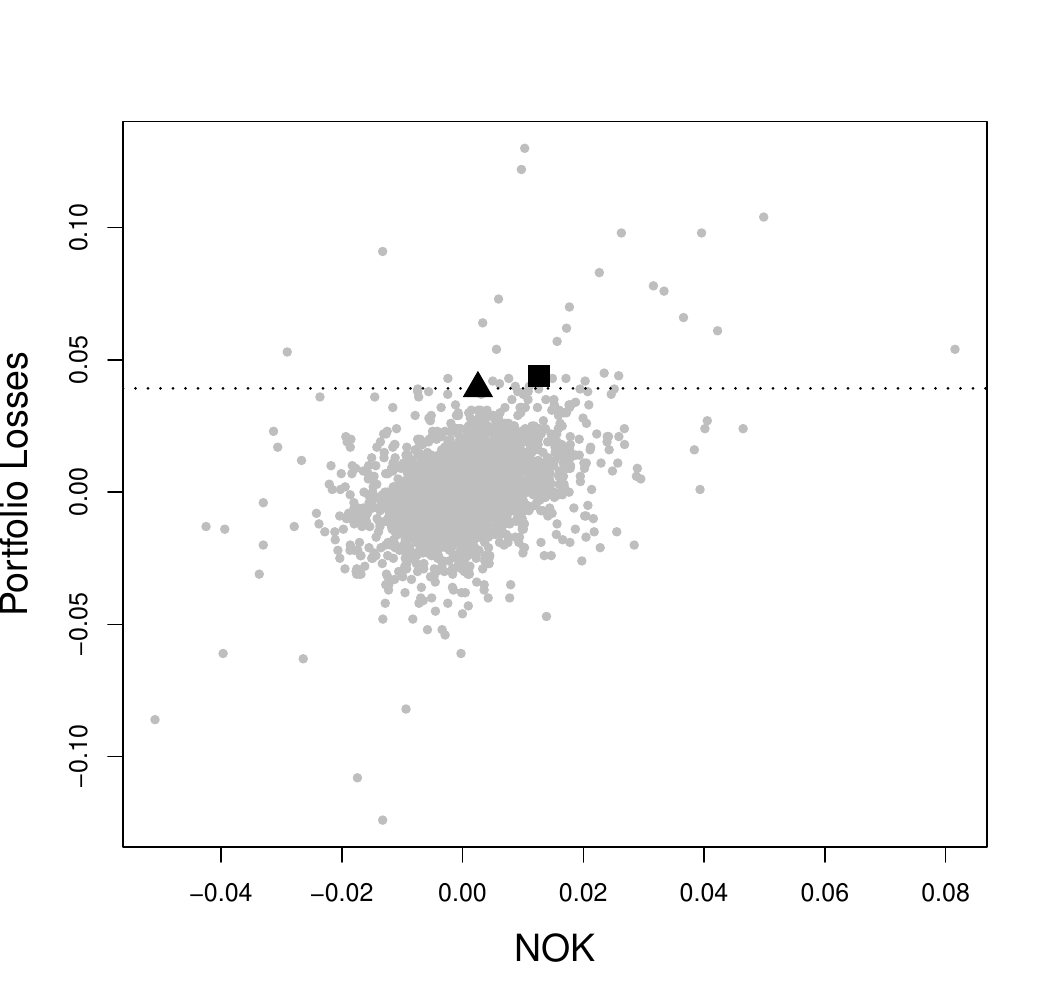}}
\subfigure{\includegraphics[scale=0.26]{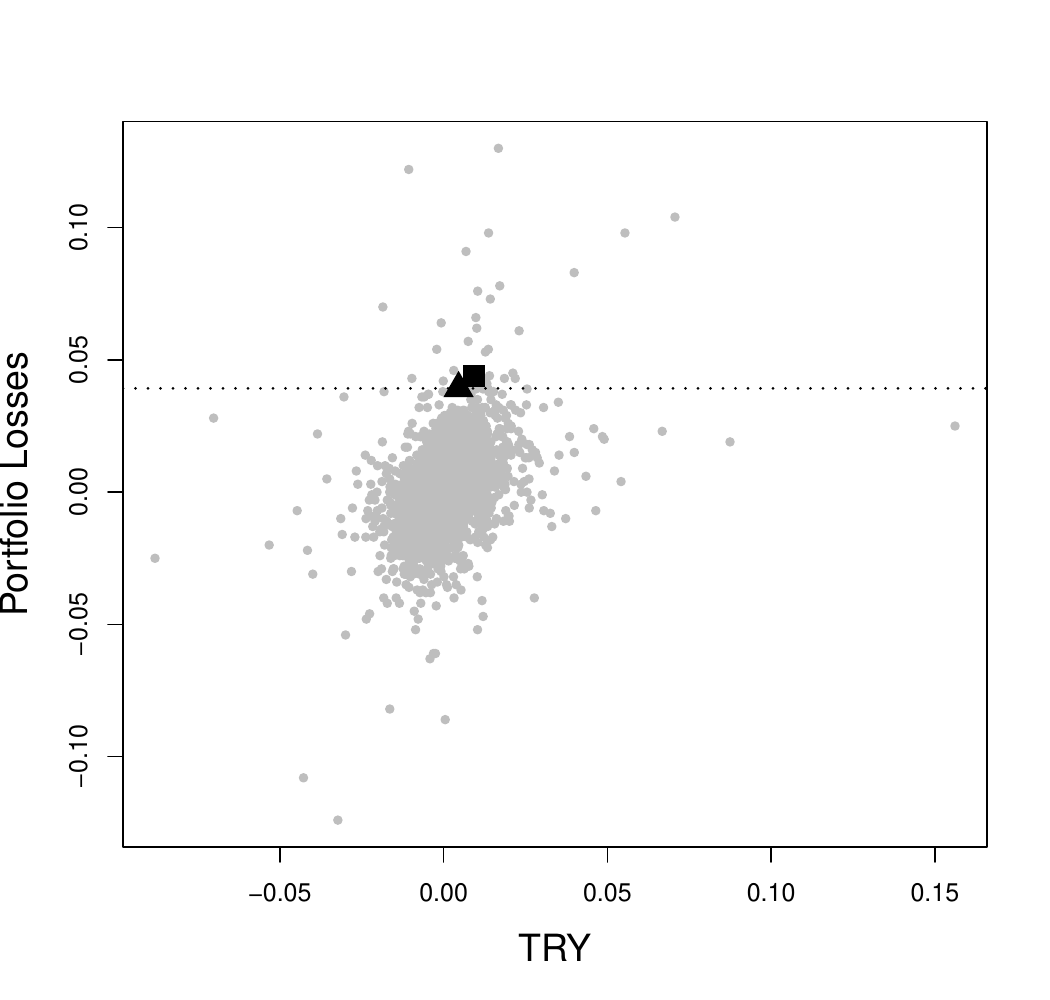}}
\subfigure{\includegraphics[scale=0.26]{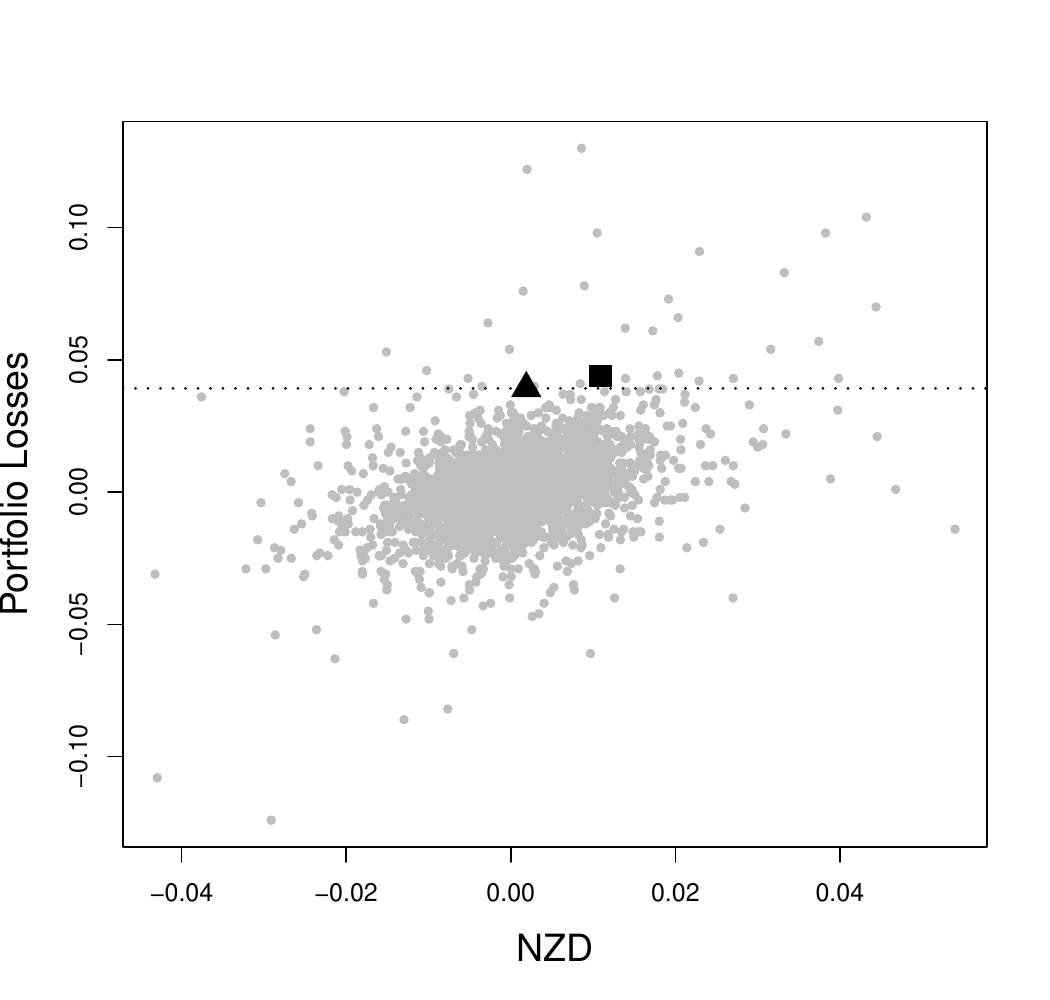}}
\subfigure{\includegraphics[scale=0.26]{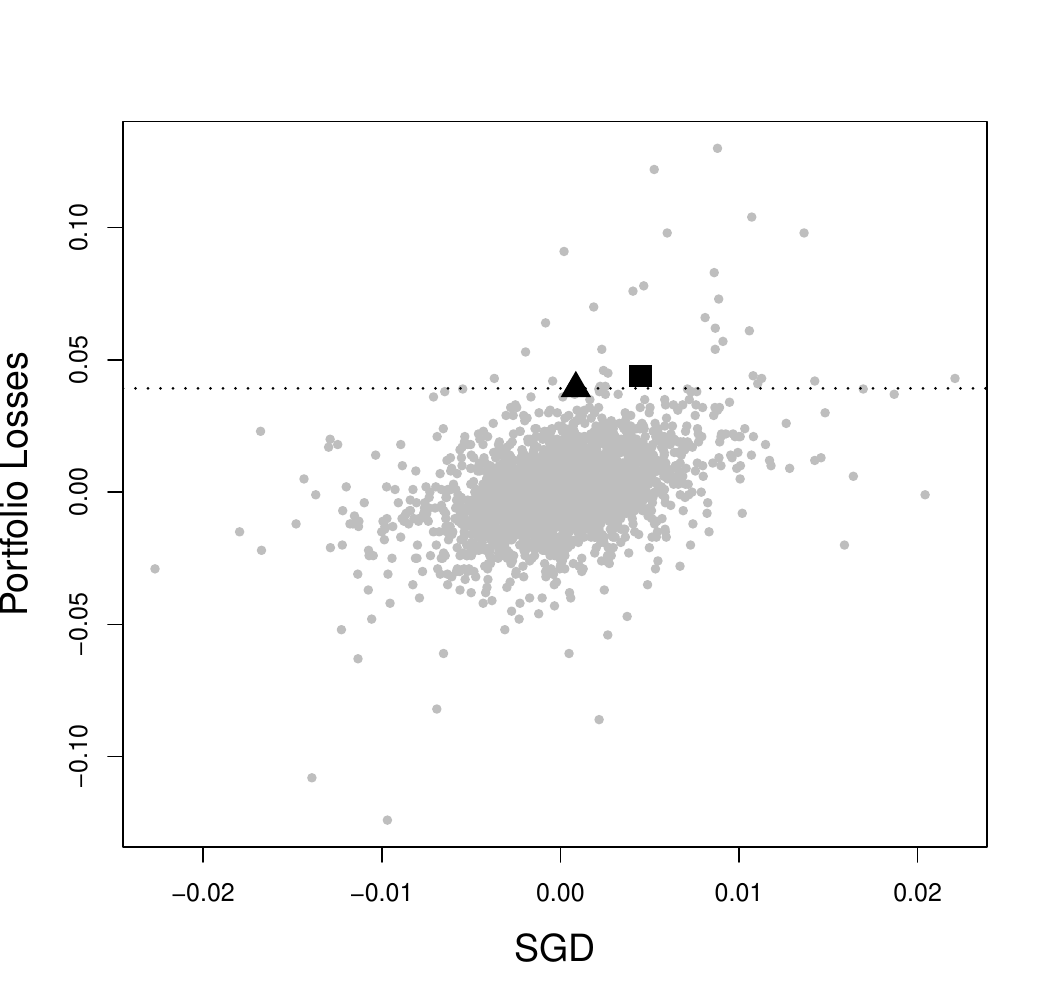}}
\subfigure{\includegraphics[scale=0.26]{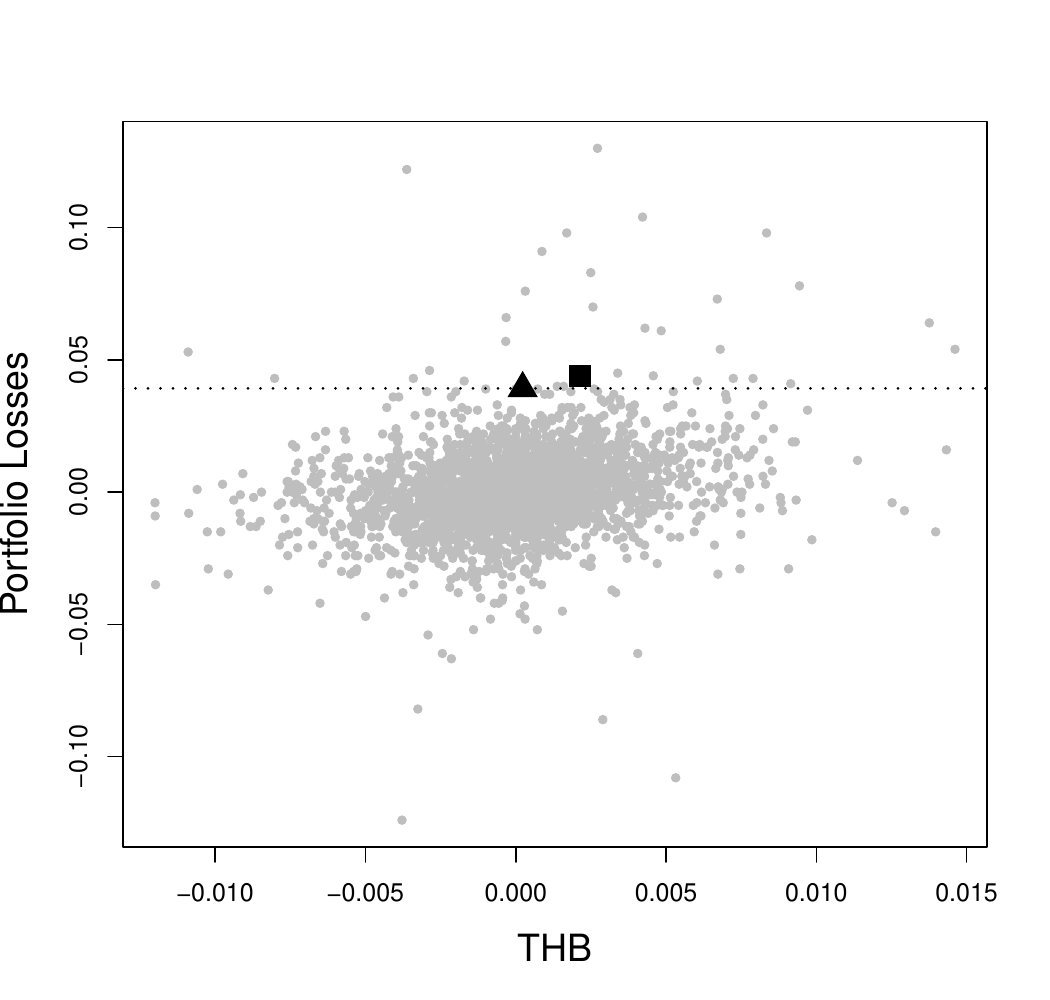}}
\subfigure{\includegraphics[scale=0.26]{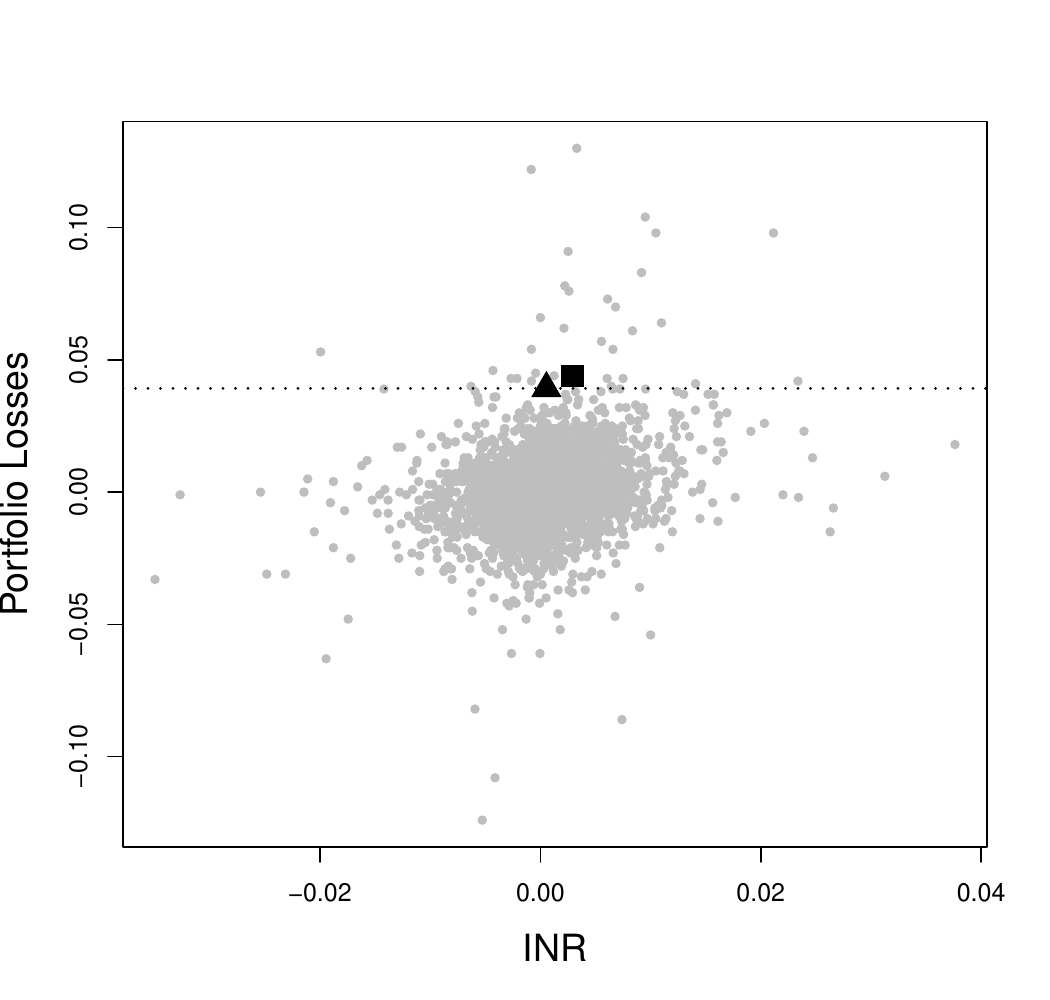}}
\subfigure{\includegraphics[scale=0.26]{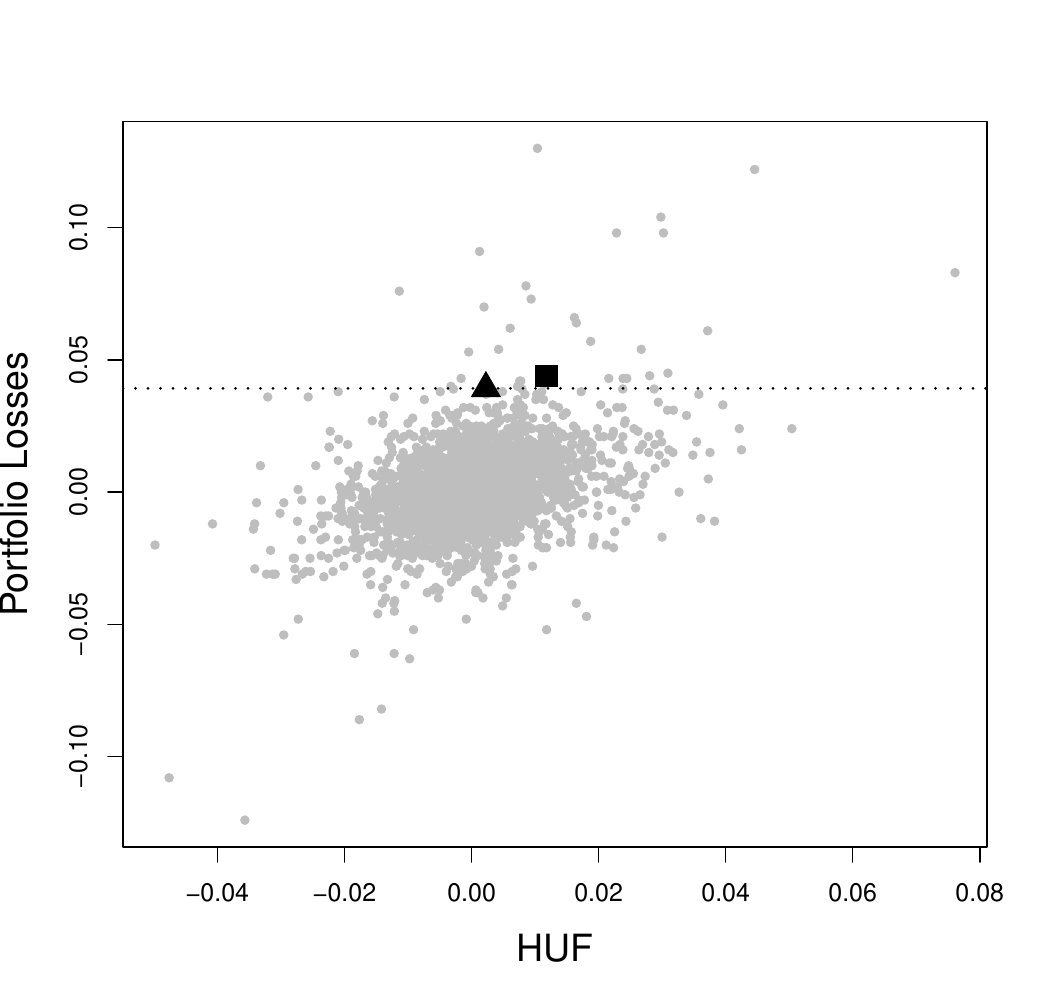}}
\subfigure{\includegraphics[scale=0.26]{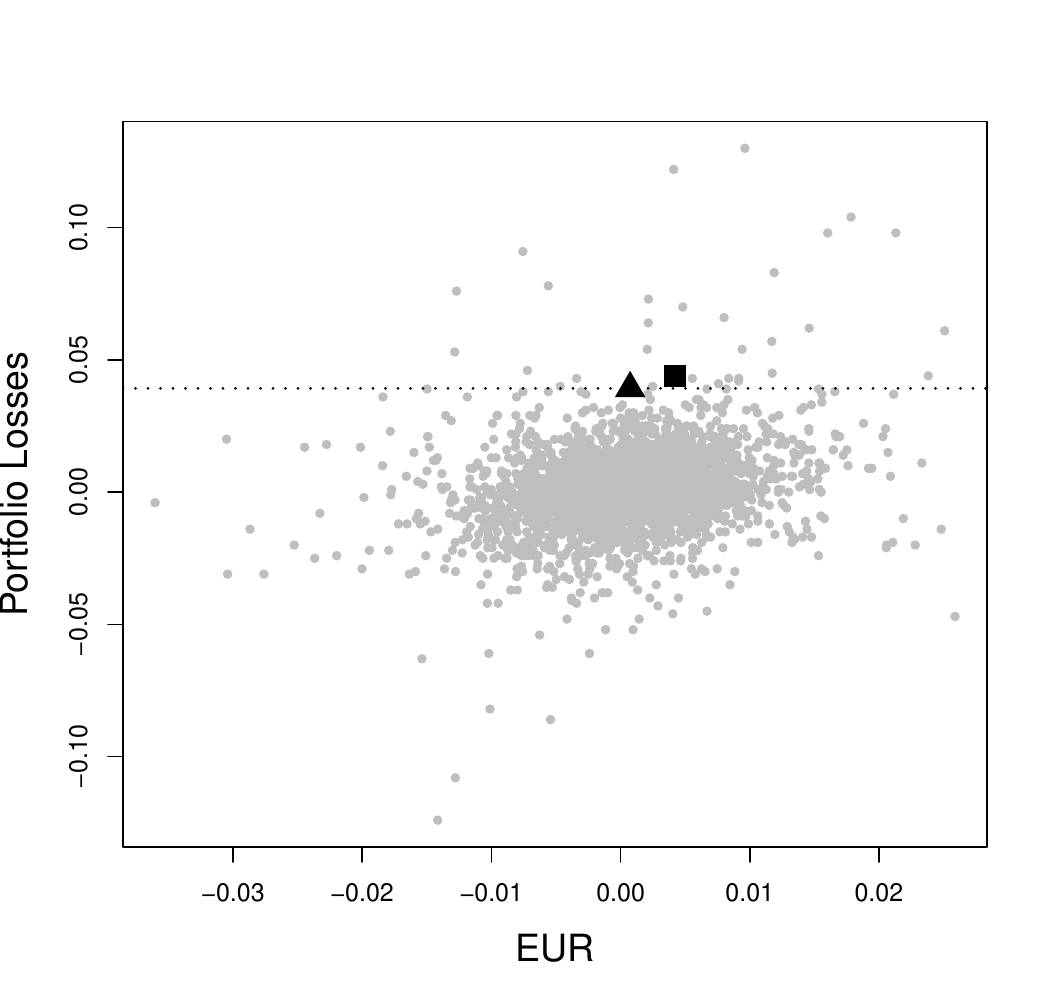}}
\caption{Plot of stress scenario estimates together of observations for Portfolio C. The horizontal lines stand for the threshold $\ell$.}
\label{fig:C_est}
\end{figure}

%
%


In addition to point prediction, we also provide confidence intervals for the fitted portfolio losses at stress scenario estimates. We construct the intervals with both vine regression and linear regression. For this portfolio, the value of R-squared coefficient of the linear model is about 0.98, suggesting that a linear portfolio map function provides a good approximation here for the true function~$g$. For clarity of illustration, Figure~\ref{fig:C_CI} zooms in on the stress scenario estimates for risk factor based on CAD currency. Taking sampling variability of the fitted function $g$ into account, this figure suggests that the GKK stress scenario estimate leads to a higher portfolio loss than specified by the threshold for the stress scenario and hence is likely not to reach the plausibility criterion.
%

\begin{figure}[ht]
\centering
\subfigure[Vine regression]{\includegraphics[scale=0.4]{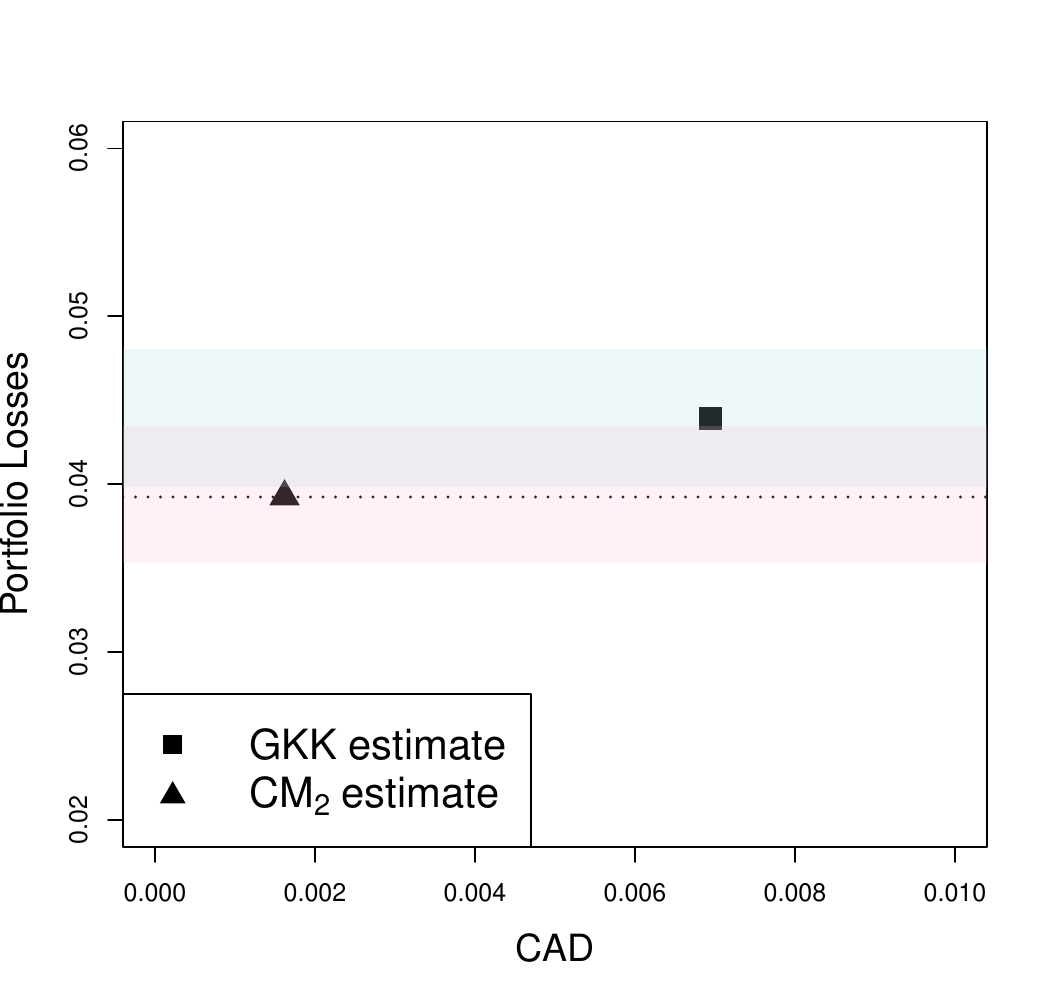}}
\subfigure[Linear regression]{\includegraphics[scale=0.4]{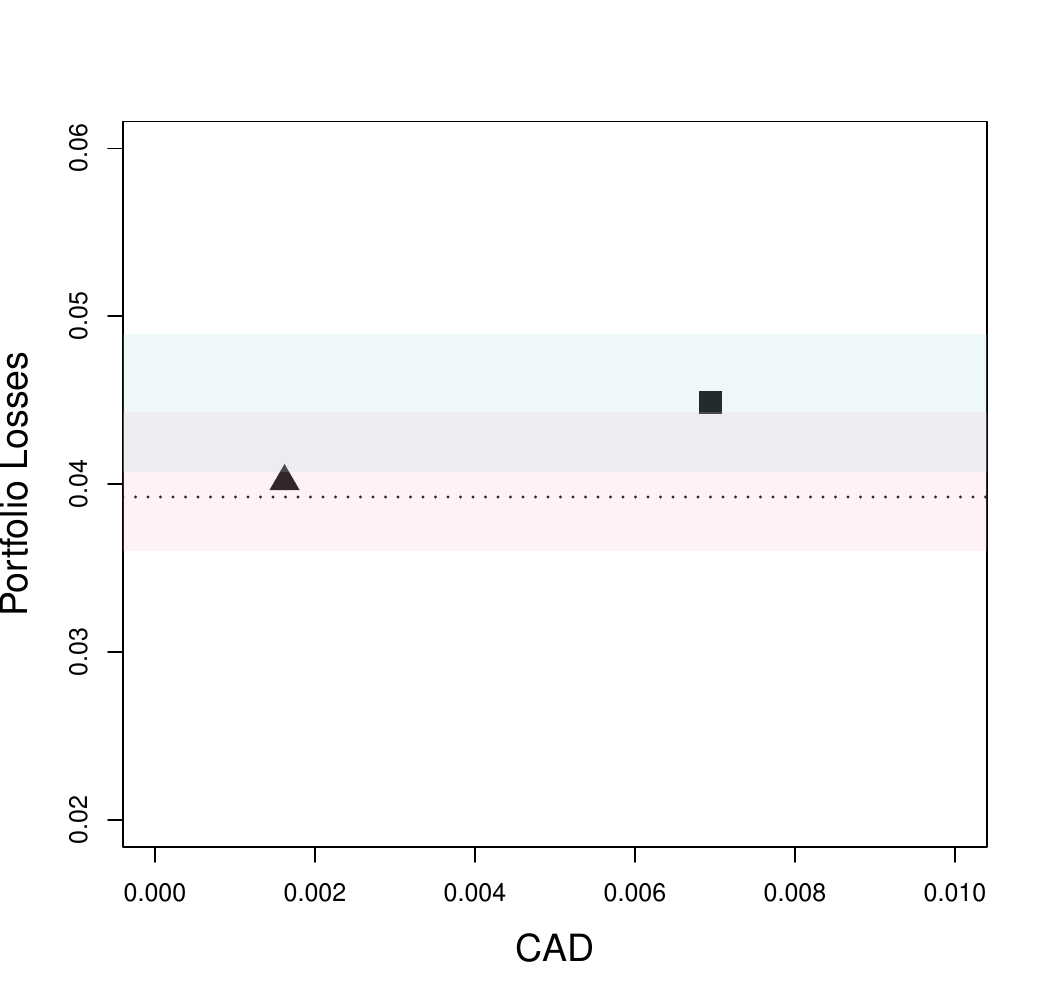}}
\caption{Plot of estimates for Portfolio C. The pink and blue shaded bands are the $95\%$ confidence intervals of predicted portfolio losses for $\CM_2$ and GKK stress scenario estimates. The dashed horizontal lines stand for the threshold $\ell$.}
\label{fig:C_CI}
\end{figure}

\section{Conclusion}
\label{sec::con}

Reverse stress testing is used by banks and insurance companies as part of their risk management toolkit to ensure financial stability of their operation. In this paper we develop a pragmatic solution for reverse stress testing through the use of a fully parametric yet flexible framework for multivariate modelling using vine copulas. The vine copula-based methodology allows to capture a wide range of stochastic properties of the data for both the marginal behaviour and the multivariate dependence structure and can be applied to large portfolios of several dozens of risk factors.  


With a constructed vine copula, three stress scenario estimators (denoted as $\CM_1$, $\CM_2$ and $\CM_3$) are investigated. Through the simulation studies, the initial naive estimator $\CM_3$ is found to be poorly calibrated leading to (fitted) portfolio loss values at estimated stress scenarios to be above the threshold specified for the stress scenario. Addressing limitations of $\CM_3$, estimator $\CM_2$ incorporates an estimate of function $g$ obtained via vine regression; this estimator is found to have the best performance in terms of bias and variance. However, $\CM_2$ also has the highest computationally complexity. Although the estimator $\CM_1$ that equates conditional densities under conditioning on $\{L\geq \ell\}$ and $\{L = \ell\}$ shows slightly worse performance than $\CM_2$, it is the most computationally efficient among the copula-based estimators. All the three copula-based estimators are superior in their performance to the estimator in~\cite{Glasserman_etal2015}. The methods are applied to three real life financial portfolios of currencies that exhibit a range of stochastic properties and include a small portfolio of just four risk factors and a large portfolio with 18 risk factors. Prediction of the corresponding portfolio losses demonstrates that the proposed methods are capable of generating stress scenarios that meet the criteria of severity and plausibility.

As a final remark, we acknowledge that the model fitting procedure for the proposed stress scenario estimators uses the entire dataset thus potentially compromising the fit in the tail in favour of that in the central part of the data.  Given this consideration, a potential modification to the methodology could involve the use of data specifically from the tail region to construct vine copulas (see, e.g., \cite{Kiriliouk_etal2023}). Moreover, the optimization step involved in determining stress scenarios based on joint density poses a computational challenge for the proposed methods, especially when the number of risk factors is larges. A possible solution is to decrease the dimensionality of the problem by clustering risk factors into several subsets so that risk factors in different subsets are conditionally independent; graphical models (\cite{EngelkeHitz2020}) could be used as a tool.

%

\section*{Acknowledgements}
The authors acknowledge financial support of the UBC-Scotiabank Risk Analytics Initiative and Natural Sciences and Engineering Research Council of Canada.

\newpage
\bibliographystyle{plainnat} 
\bibliography{biblio}{}

\begin{thebibliography}{38}
\providecommand{\natexlab}[1]{#1}
\providecommand{\url}[1]{\texttt{#1}}
\expandafter\ifx\csname urlstyle\endcsname\relax
  \providecommand{\doi}[1]{doi: #1}\else
  \providecommand{\doi}{doi: \begingroup \urlstyle{rm}\Url}\fi

\bibitem[{Basel Committee on Banking Supervision}(2009)]{Basel2009}
{Basel Committee on Banking Supervision}.
\newblock Principles for sound stress testing practices and supervision, 2009.

\bibitem[Bedford and Cooke(2001)]{BedfordCooke2001}
T.~Bedford and R.M. Cooke.
\newblock Probability density decomposition for conditionally dependent random
  variables modeled by vines.
\newblock \emph{Annals of Mathematics and Artificial Intelligence},
  32:\penalty0 245--268, 2001.

\bibitem[Bedford and Cooke(2002)]{BedfordCooke2002}
T.~Bedford and R.M. Cooke.
\newblock Vines: a new graphical model for dependent random variables.
\newblock \emph{The Annals of Statistics}, 30:\penalty0 1031--1068, 2002.

\bibitem[Breuer and Csisz{\'a}r(2013)]{BreuerCsiszar2013}
T.~Breuer and I.~Csisz{\'a}r.
\newblock Systematic stress tests with entropic plausibility constraints.
\newblock \emph{Journal of Banking \& Finance}, 37:\penalty0 1552--1559, 2013.

\bibitem[Breuer et~al.(2009)Breuer, Jandacka, Rheinberger, and
  Summer]{Breuer_etal2009}
T.~Breuer, M.~Jandacka, K.~Rheinberger, and M.~Summer.
\newblock How to find plausible, severe, and useful stress scenarios.
\newblock \emph{International Journal of Central Banking}, 5:\penalty0
  205--224, 2009.

\bibitem[Chang and Joe(2019)]{ChangJoe2019}
B.~Chang and H.~Joe.
\newblock Prediction based on conditional distributions of vine copulas.
\newblock \emph{Computational Statistics \& Data Analysis}, 139:\penalty0
  45--63, 2019.

\bibitem[{Committee of European Banking Supervision}(2009)]{Committee2009}
{Committee of European Banking Supervision}.
\newblock Guidelines on stress testing ({CP}32). {E}uropean {B}anking
  {A}uthority, 2009.

\bibitem[Cooke et~al.(2015)Cooke, Joe, and B.Chang]{Cooke_etal2015}
R.M. Cooke, H.~Joe, and B.Chang.
\newblock Vine regression.
\newblock \emph{Resources for the Future Discussion Paper}, pages 15--52, 2015.

\bibitem[Czado(2019)]{Czado2019}
C.~Czado.
\newblock \emph{Analyzing Dependent Data with Vine Copulas: A Practical Guide
  With R}.
\newblock Springer International Publishing, 2019.

\bibitem[Embrechts et~al.(1997)Embrechts, Kl{\"u}ppelberg, and
  Mikosch]{Embrechts_etal1997}
P.~Embrechts, C.~Kl{\"u}ppelberg, and T.~Mikosch.
\newblock \emph{Modelling Extremal Events: for Insurance and Finance},
  volume~33.
\newblock Springer Berlin Heidelberg, 1997.

\bibitem[Engelke and Hitz(2020)]{EngelkeHitz2020}
S.~Engelke and A.S. Hitz.
\newblock Graphical models for extremes.
\newblock \emph{Journal of the Royal Statistical Society Series B: Statistical
  Methodology}, 82:\penalty0 871--932, 2020.

\bibitem[{Financial Services Authority}(2009)]{FSA2009}
{Financial Services Authority}.
\newblock Stress and scenario testing: {F}eedback on {CP}08/24 and final rules.
  {P}olicy statement 09/20, 2009.

\bibitem[Glasserman et~al.(2015)Glasserman, Kang, and
  Kang]{Glasserman_etal2015}
P.~Glasserman, C.~Kang, and W.~Kang.
\newblock Stress scenario selection by empirical likelihood.
\newblock \emph{Quantitative Finance}, 15\penalty0 (1):\penalty0 25--41, 2015.

\bibitem[Grundke(2011)]{Grundke2011}
P.~Grundke.
\newblock Reverse stress tests with bottom-up approaches.
\newblock \emph{Journal of Risk Model Validation}, 5:\penalty0 71, 2011.

\bibitem[Grundke and Pliszka(2018)]{GrundkePliszka2018}
P.~Grundke and K.~Pliszka.
\newblock A macroeconomic reverse stress test.
\newblock \emph{Review of Quantitative Finance and Accounting}, 50:\penalty0
  1093--1130, 2018.

\bibitem[Gudendorf and Segers(2010)]{GudendorfSergers2010}
G.~Gudendorf and J.~Segers.
\newblock Extreme-value copulas.
\newblock In \emph{Copula Theory and Its Applications: Proceedings of the
  Workshop Held in Warsaw, 25-26 September 2009}, pages 127--145. Springer,
  2010.

\bibitem[Joe(1996)]{Joe1996}
H.~Joe.
\newblock Families of m-variate distributions with given margins and m (m-1)/2
  bivariate dependence parameters.
\newblock \emph{Lecture Notes-Monograph Series}, 28:\penalty0 120--141, 1996.

\bibitem[Joe(1997)]{Joe1997}
H.~Joe.
\newblock \emph{Multivariate Models and Multivariate Dependence Concepts}.
\newblock Chapman and Hall, 1997.

\bibitem[Joe(2014)]{Joe2014}
H.~Joe.
\newblock \emph{Dependence Modeling with Copulas}.
\newblock CRC press, 2014.

\bibitem[Jones and Faddy(2003)]{JonesFaddy2003}
M.C. Jones and M.J. Faddy.
\newblock A skew extension of the t-distribution, with applications.
\newblock \emph{Journal of the Royal Statistical Society: Series B (Statistical
  Methodology)}, 65:\penalty0 159--174, 2003.

\bibitem[Kiriliouk et~al.(2023)Kiriliouk, Lee, and Segers]{Kiriliouk_etal2023}
A.~Kiriliouk, J.~Lee, and J.~Segers.
\newblock X-vine models for multivariate extremes.
\newblock \emph{arXiv preprint arXiv:2312.15205}, 2023.

\bibitem[Kopeliovich et~al.(2015)Kopeliovich, Novosyolov, Satchkov, and
  Schachter]{Kopeliovich_etal2015}
Y.~Kopeliovich, A.~Novosyolov, D.~Satchkov, and B.~Schachter.
\newblock Robust risk estimation and hedging: A reverse stress testing
  approach.
\newblock \emph{The Journal of Derivatives}, 22:\penalty0 10--25, 2015.

\bibitem[Kraus and Czado(2017)]{KrausCzado2017}
D.~Kraus and C.~Czado.
\newblock D-vine copula based quantile regression.
\newblock \emph{Computational Statistics \& Data Analysis}, 110:\penalty0
  1--18, 2017.

\bibitem[Lee et~al.(2018)Lee, Joe, and Krupskii]{Lee_etal2018}
D.~Lee, H.~Joe, and P.~Krupskii.
\newblock Tail-weighted dependence measures with limit being the tail
  dependence coefficient.
\newblock \emph{Journal of Nonparametric Statistics}, 30:\penalty0 262--290,
  2018.

\bibitem[MacDonald et~al.(2011)MacDonald, Scarrott, Lee, Darlow, Reale, and
  Russell]{Macdonald_etal2011}
A.~MacDonald, C.J. Scarrott, D.~Lee, B.~Darlow, M.~Reale, and G.~Russell.
\newblock A flexible extreme value mixture model.
\newblock \emph{Computational Statistics \& Data Analysis}, 55:\penalty0
  2137--2157, 2011.

\bibitem[McNeil and Smith(2012)]{McneilSmith2012}
A.J. McNeil and A.D. Smith.
\newblock Multivariate stress scenarios and solvency.
\newblock \emph{Insurance: Mathematics and Economics}, 50:\penalty0 299--308,
  2012.

\bibitem[Morales-Napoles(2010)]{Morales2011}
O.~Morales-Napoles.
\newblock Counting vines.
\newblock In \emph{Dependence modeling: Vine Copula Handbook}, pages 189--218.
  World Scientific, 2010.

\bibitem[Mullen et~al.(2011)Mullen, Ardia, Gil, Windover, and Cline]{DEoptim}
K.~Mullen, D.~Ardia, D.L. Gil, D.~Windover, and J.~Cline.
\newblock {DE}optim: An {R} package for global optimization by differential
  evolution.
\newblock \emph{Journal of Statistical Software}, 40:\penalty0 1--26, 2011.

\bibitem[Nagler and C.Czado(2016)]{NaglerCzado2016}
T.~Nagler and C.Czado.
\newblock Evading the curse of dimensionality in nonparametric density
  estimation with simplified vine copulas.
\newblock \emph{Journal of Multivariate Analysis}, 151:\penalty0 69--89, 2016.

\bibitem[Nagler and Vatter(2022)]{rvinecopulib}
T.~Nagler and T.~Vatter.
\newblock \emph{rvinecopulib: High Performance Algorithms for Vine Copula
  Modeling}, 2022.
\newblock R package version 0.6.2.1.2.

\bibitem[Nagler et~al.(2022)Nagler, Schepsmeier, Stoeber, Brechmann, Graeler,
  and Erhardt]{VineCopula}
T.~Nagler, U.~Schepsmeier, J.~Stoeber, E.C. Brechmann, B.~Graeler, and
  T.~Erhardt.
\newblock \emph{VineCopula: Statistical Inference of Vine Copulas}, 2022.
\newblock R package version 2.4.4.

\bibitem[Noh et~al.(2013)Noh, Ghouch, and Bouezmarni]{Noh_etal2013}
H.~Noh, A.E. Ghouch, and T.~Bouezmarni.
\newblock Copula-based regression estimation and inference.
\newblock \emph{Journal of the American Statistical Association}, 108:\penalty0
  676--688, 2013.

\bibitem[Parsa and Klugman(2011)]{ParsaKlugman2011}
R.A. Parsa and S.A. Klugman.
\newblock Copula regression.
\newblock \emph{Variance Advancing and Science of Risk}, 5:\penalty0 45--54,
  2011.

\bibitem[Pesenti et~al.(2019)Pesenti, Millossovich, and Tsanakas]{Pesenti2019}
S.M. Pesenti, P.~Millossovich, and A.~Tsanakas.
\newblock Reverse sensitivity testing: What does it take to break the model?
\newblock \emph{European Journal of Operational Research}, 274:\penalty0
  654--670, 2019.

\bibitem[Politis and Romano(1994)]{PolitisRomano1994}
D.N. Politis and J.P. Romano.
\newblock The stationary bootstrap.
\newblock \emph{Journal of the American Statistical Association}, 89:\penalty0
  1303--1313, 1994.

\bibitem[Schallhorn et~al.(2017)Schallhorn, Kraus, Nagler, and
  Czado]{Schallhorn_etal2017}
N.~Schallhorn, D.~Kraus, T.~Nagler, and C.~Czado.
\newblock D-vine quantile regression with discrete variables.
\newblock \emph{arXiv preprint arXiv:1705.08310}, 2017.

\bibitem[Sch{\"o}lkopf et~al.(1998)Sch{\"o}lkopf, Smola, and
  M{\"u}ller]{Scholkopf_etal1998}
B.~Sch{\"o}lkopf, A.~Smola, and K.R. M{\"u}ller.
\newblock Nonlinear component analysis as a kernel eigenvalue problem.
\newblock \emph{Neural computation}, 10:\penalty0 1299--1319, 1998.

\bibitem[Sklar(1959)]{Sklar1959}
A.~Sklar.
\newblock Fonctions de r{\'e}partition {\`a} n dimensions et leurs marges.
\newblock \emph{Publications de l'Institut de Statistique de l'Universit{\'e}
  de Paris}, 8:\penalty0 229--231, 1959.

\end{thebibliography}


\begin{appendices}
\renewcommand{\thesection}{Appendix \Alph{section}}

\section{Vine Copula Regression}
\label{appen::vine_reg}

Copula regression has received a considerable attention in recent years;  see, e.g.,~\cite{ParsaKlugman2011}, \cite{Noh_etal2013}, \cite{Cooke_etal2015}, \cite{KrausCzado2017}, \cite{Schallhorn_etal2017} and \cite{ChangJoe2019}. It is shown that copula regression can model not only nonlinear relationships between the response and explanatory variables but also heteroscedasticity. In this section, we review a special case of copula regression where vine copulas are used to model the joint distribution between the response and explanatory variables.

%

Let $\Xb = (X_1,X_2,...,X_d)^\top $ denote the $d$-dimensional vector of explanatory variables and $Y$ be the response variable. From~\eqref{eq:vcop}, we know that the joint density function of $(\Xb, Y)$ can be specified using a vine copula. Using the joint density function, the conditional distribution of $Y$ given $\Xb$ can be calculated through numerical integration. However, to avoid this integration, \cite{ChangJoe2019} propose to constrain the node that contains the response variable $Y$ as the conditioned variable to be a leaf node in all trees.

To be more specific, a locally optimal vine tree structure is first constructed for $\Xb\in \rbb^d$. Then from level 1 to $d-1$, the node containing random variable $Y$ is sequentially linked to the node that satisfies the proximity condition and has the largest absolute correlation with the node containing $Y$.  As an example, consider the first two trees $\TC_1, \TC_2$ of an R-vine for $\Xb\in \rbb^5$ given in Figure~\ref{fig:rvine}. To add $Y$ into $\TC_1$, correlation coefficients $\rho_{i,Y}, i \in \{1,2,3,4,5\}$ between $X_i$'s and $Y$ are calculated. If say $|\rho_{2,Y}|$ is the largest, then an edge to connect $Y$ and $X_2$ is added; see Figure~\ref{fig:rvine_add_Y} for illustration. In the second tree, the node $X_2,Y$ can either be connected with node $X_1,X_2$ or with node $X_2,X_3$. Next calculate $\rho_{1,Y;2}$ and $\rho_{3,Y;2}$, which are the conditional correlations between $X_1$ and $Y$ given $X_2$ and between $X_3$ and $Y$ given $X_2$. If $|\rho_{1,Y;2}| > |\rho_{3,Y;2}|$, then nodes $X_2,Y$ and $X_1,X_2$ are connected in $\TC_2$.

\begin{figure}[ht]
\centering
\includegraphics[scale=0.5]{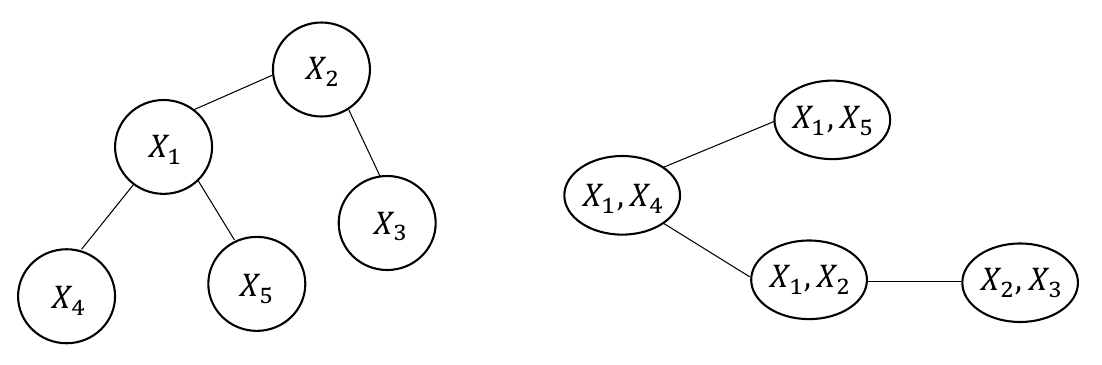}
\caption{The first two trees $\TC_1, \TC_2$ of an R-vine with 5 explanatory variables.}
\label{fig:rvine}
\end{figure}

\begin{figure}[ht]
\centering
\includegraphics[scale=0.5]{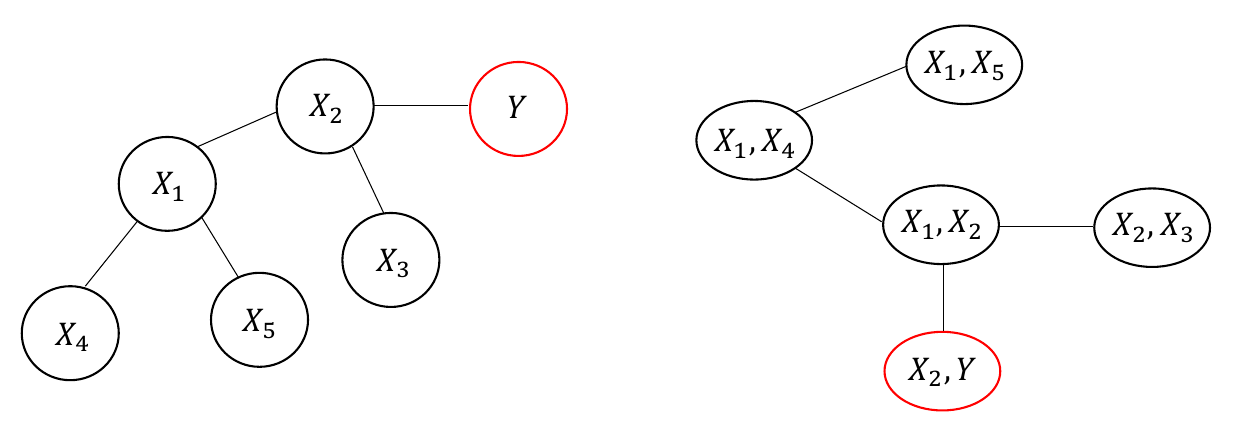}
\caption{The first two trees $\TC_1, \TC_2$ of an R-vine with 5 explanatory variables and response variable $Y$.}
\label{fig:rvine_add_Y}
\end{figure}

Under this construction, a pair copula $C_{X_j,Y; \Xb_{-j}}$ for $X_j$ and $Y$ given $\Xb_{-j}$ can be obtained in the last tree, where $\Xb_{-j}\in \rbb^{d-1}$ contains all components of $\Xb$ except $X_j$. Then the conditional cdf of $Y$ given $\Xb = \xb$ can be expressed as

$$F_{Y|\Xb}\big(y|\xb\big) = h_{Y|X_j,\Xb_{-j}}\Big(F_{X_j|\Xb_{-j}}\big(x_j|\xb_{-j}\big)\big| F_{Y|\Xb_{-j}}\big(y|\xb_{-j}\big)\Big),$$
where 
$$h_{Y|X_j,\Xb_{-j}}(v|u) = \frac{\partial C_{X_j,Y; \Xb_{-j}}(u,v)}{\partial u}.$$
Conditional cdf's $F_{X_j|\Xb_{-j}}(x_j|\xb_{-j})$ and $F_{Y|\Xb_{-j}}(y|\xb_{-j})$ can be obtained recursively; see Algorithm 3.1 in \cite{ChangJoe2019}.

To illustrate, consider $(\Xb, Y)^T = (X_1, X_2, X_3, Y)^T$. Suppose we have a D-vine tree with order $Y-X_1-X_2-X_3$. From $\TC_1$, we obtain $F_{Y|X_1}(y|x_1)$, $F_{X_2|X_1}(x_2|x_1)$ and $F_{X_3|X_2}(x_3|x_2)$. From $\TC_2$, we have $F_{Y|X_1, X_2}(y|x_1, x_2)$ and $F_{X_3|X_1, X_2}(x_3|x_1, x_2)$. In the last tree $\TC_3$, we finally calculate $F_{Y|X_1, X_2, X_3}(y|x_1, x_2, x_3)$.

\section{Marginal estimation}
\label{appen::margin}

In this section, we summarize the procedure we used to fit univariate marginal components in the application. 

As a first step, we fit the skew-t distribution proposed in~\cite{JonesFaddy2003}. It is a flexible four-parameter family of distributions that allows for asymmetry through different decay in the lower and upper tails. The density function of the skew-t distribution with location $\mu\in\rbb$, scale $\sigma>0$ and shape parameters $\alpha,\beta>0$ is given by
$$f(x;\mu, \sigma, \alpha, \beta) = \dfrac1{\sigma C_{\alpha,\beta}}\left\{1+\frac{1}{(\alpha+\beta+z^2)^{1/2}}\right\}^{\alpha+1/2}\left\{1-\frac{1}{(\alpha+\beta+z^2)^{1/2}}\right\}^{\beta+1/2},\quad x\in\rbb,$$
where $z = (x-\mu)/\sigma$, $C_{\alpha,\beta} = 2^{\alpha+\beta-1}B(\alpha, \beta)(\alpha+\beta)^{1/2}$ and $B(\alpha,\beta) = \int_0^1 t^{\alpha - 1}(1-t)^{\beta-1}dt$.
The Student t distribution with $2\alpha$ degrees freedom is a special case of this skew-t distribution when $\alpha = \beta$. 

After fitting the skew-t distribution to each marginal component, we assess its goodness-of-fit by examining the quantile-quantile (Q-Q) plots. We found that for most of the risk factors the skew-t distribution provides an acceptable fit. In a few cases, however, this parametric family could not capture both the central part and tails of the data. In these situations, we considered a semi-parametric approach that combines the kernel density estimator for the central observations and the generalized Pareto distribution for the upper and lower tails of the marginal components. The generalized Pareto distribution, whose cdf has the form:
$$F_{GP}(x; \sigma,\xi) = 1- \left[1 + \xi x/\sigma\right]_+^{-1/\xi} $$ 
for scale parameter $\sigma>0$ and shape parameter $\xi\in\rbb$, is a standard model for excesses above a high threshold, motivated by extreme value theory (see, e.g., \cite{Embrechts_etal1997}).

%


Let $u_L$ and $u_R$ denote the lower and upper threshold values, respectively. Let $H(\cdot)$ and $h(\cdot)$ denote the cdf and probability density function of a random variable $X$. By approximating the tail behaviour of $X$ below $u_L$ and above $u_R$ by generalized Pareto distributions with parameters $(\sigma_L,\xi_L)$ in the lower tail and $(\sigma_R,\xi_R)$ in the upper tail, the cdf of $X$ can be written as
$$F(x) = \left\{\begin{aligned}
&H(u_L)\cdot \big[1 + \xi_L(-x + u_L)/\sigma_L\big]_+^{-1/\xi_L}, & x < u_L;\\
&H(x), & u_L\leq x\leq u_R;\\
&1-[1-H(u_R)]\cdot \big[1+\xi_R(x - u_R)/\sigma_R\big]_+^{-1/\xi_R}, & x > u_R,
\end{aligned}
\right.
$$
with the density function given by
$$f(x) = \left\{\begin{aligned}
&\frac{H(u_L)}{\sigma_L}\cdot \big[1 + \xi_L(-x + u_L)/\sigma_L\big]_+^{-1/\xi_L-1}, & x < u_L;\\
&h(x), & u_L\leq x\leq u_R;\\
&\frac{1-H(u_R)}{\sigma_R}\cdot \big[1+\xi_R(x - u_R)/\sigma_R\big]_+^{-1/\xi_R-1}, & x > u_R.
\end{aligned}
\right.
$$

We should note that, in general, there are two jump points in the density function at the two thresholds. 
We solve this problem by adding two constraints on the parameters $\xi_L, \xi_R, \sigma_L, \sigma_R$:
$$\frac{H(u_L)}{\sigma_L}\cdot \left[1 + \xi_L(-u_L + u_L)/\sigma_L\right]_+^{-1/\xi_L-1} = h(u_L)$$
and
$$\frac{1-H(u_R)}{\sigma_R}\cdot \left[1+\xi_R(u_R - u_R)/\sigma_R\right]_+^{-1/\xi_R-1} = h(u_R).$$
After simplifying, the constraints involve only scale parameters $\sigma_L$ and $\sigma_R$ via:
$$\frac{H(u_L)}{\sigma_L} = h(u_L),\quad \frac{1-H(u_R)}{\sigma_R} = h(u_R).$$
The semi-parametric model is fitted by first estimating $H(x)$ and $h(x)$ using all observations with the kernel density estimator as in~\cite{Macdonald_etal2011}. Next, we select the lower and upper thresholds. In our procedure, we use the $0.15$ and $0.85$-empirical quantiles of the data. Then, we estimate scale parameters as $\hat{\sigma}_L = \hat{H}(u_L)/\hat{h}(u_L)$ and $\hat{\sigma}_R = [1-\hat{H}(u_R)]/\hat{h}(u_R)$. Finally, we estimate $\xi_L$ and $\xi_R$ using observations below/above the lower/upper threshold using maximum likelihood estimation. The Q-Q plots are used to validate the semi-parametric model for the marginal components.

\section{Figures}

\subsection{Plot of Trees in the R-vine Selected for Portfolio A}
\label{appen::tree_detail}
Figure~\ref{fig:trees_A} provides the four trees in the R-vine selected for Portfolio A data. The chosen copula families and their estimated parameters are given beside each edge, where

\begin{itemize}
    \item ``I($0$)" stands for independent copula;
    \item ``N($\rho$)" stands for Gaussian copula with correlation coefficient $\rho\in [-1,1]$;
    \item  ``t($\rho, \nu$)" stands for t copula with correlation coefficient $\rho\in [-1,1]$ and degrees of freedom $\nu > 0$;
    \item ``C($\delta$)" represents Clayton copula with dependence parameter $\delta > 0$;
    \item ``F($\delta$)" represents Frank copula with dependence parameter $\delta\in [-\infty,\infty]\backslash \{0\}$;
    \item ``BB1($\theta, \delta$)" represents BB1 copula with dependence parameters $\theta > 0$ and $\delta \geq 1$.
\end{itemize}

\begin{figure}[H]
\centering
\subfigure{\includegraphics[scale=0.4]{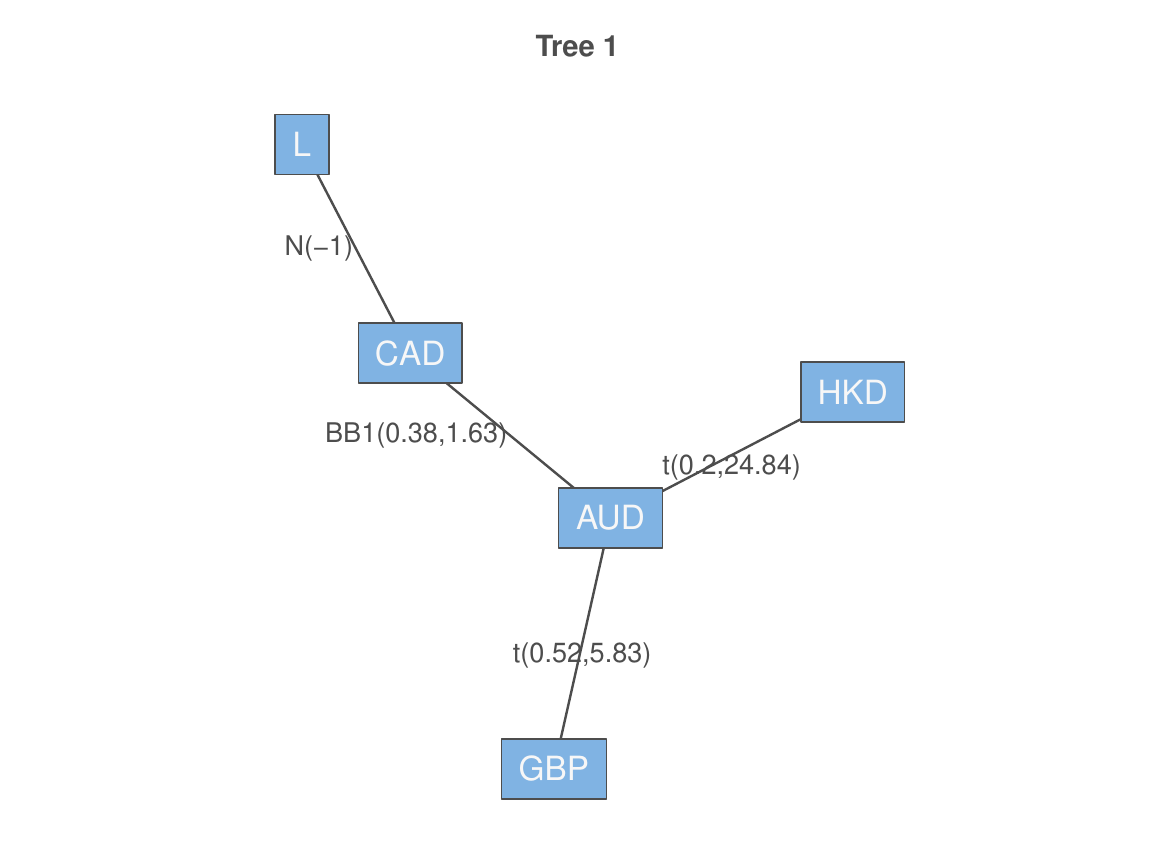}}
\subfigure{\includegraphics[scale=0.4]{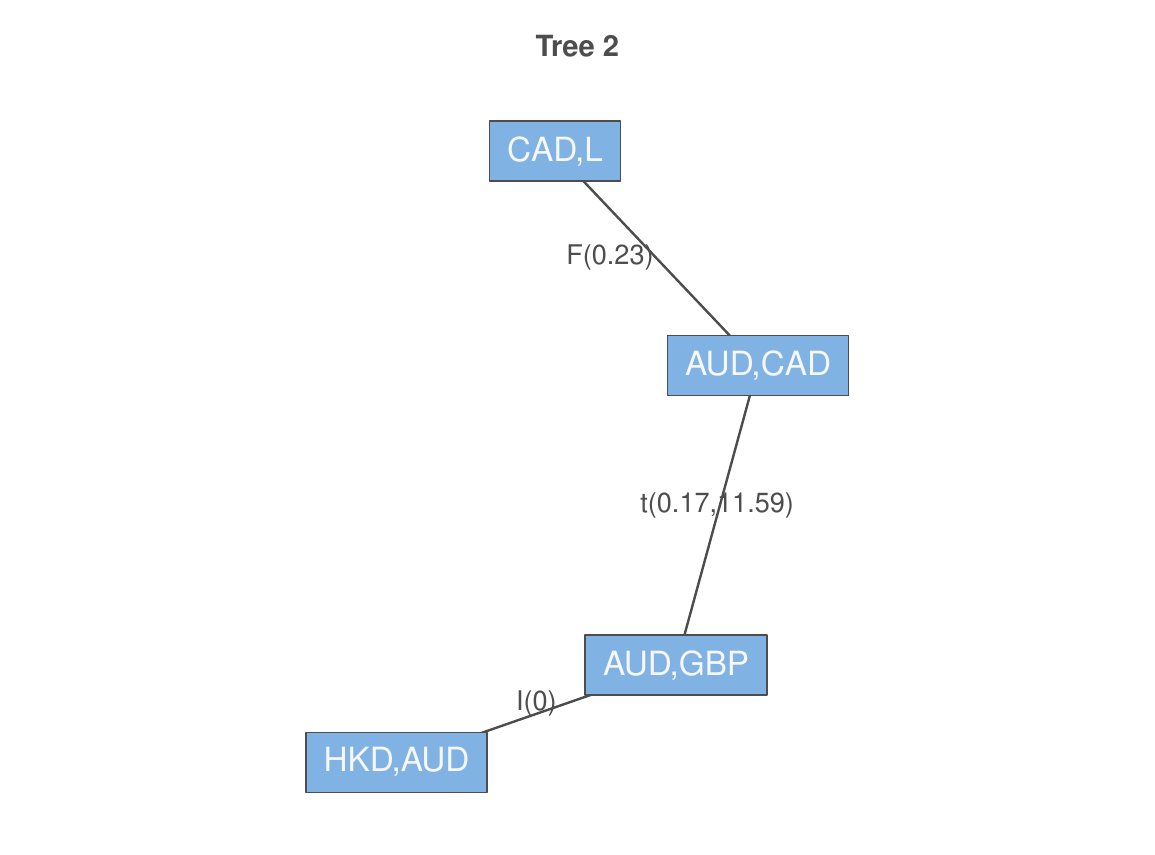}}
\subfigure{\includegraphics[scale=0.4]{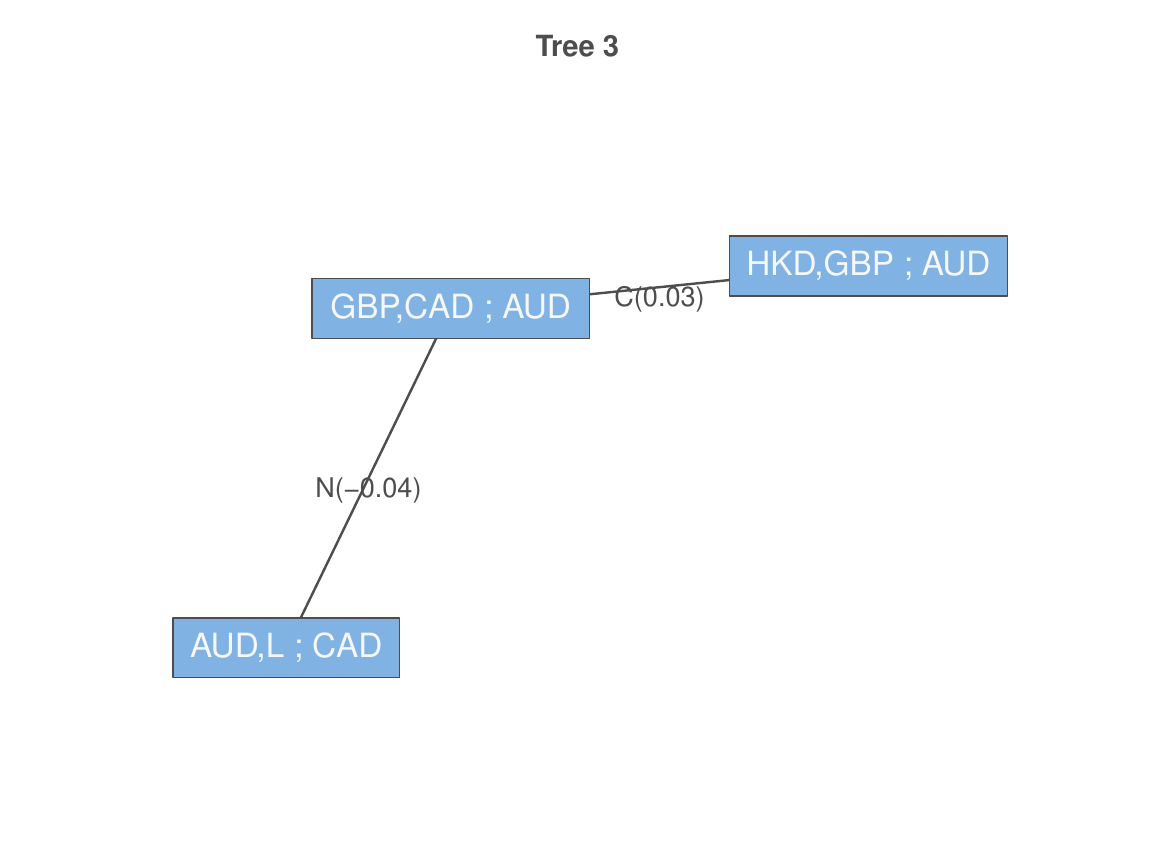}}
\subfigure{\includegraphics[scale=0.4]{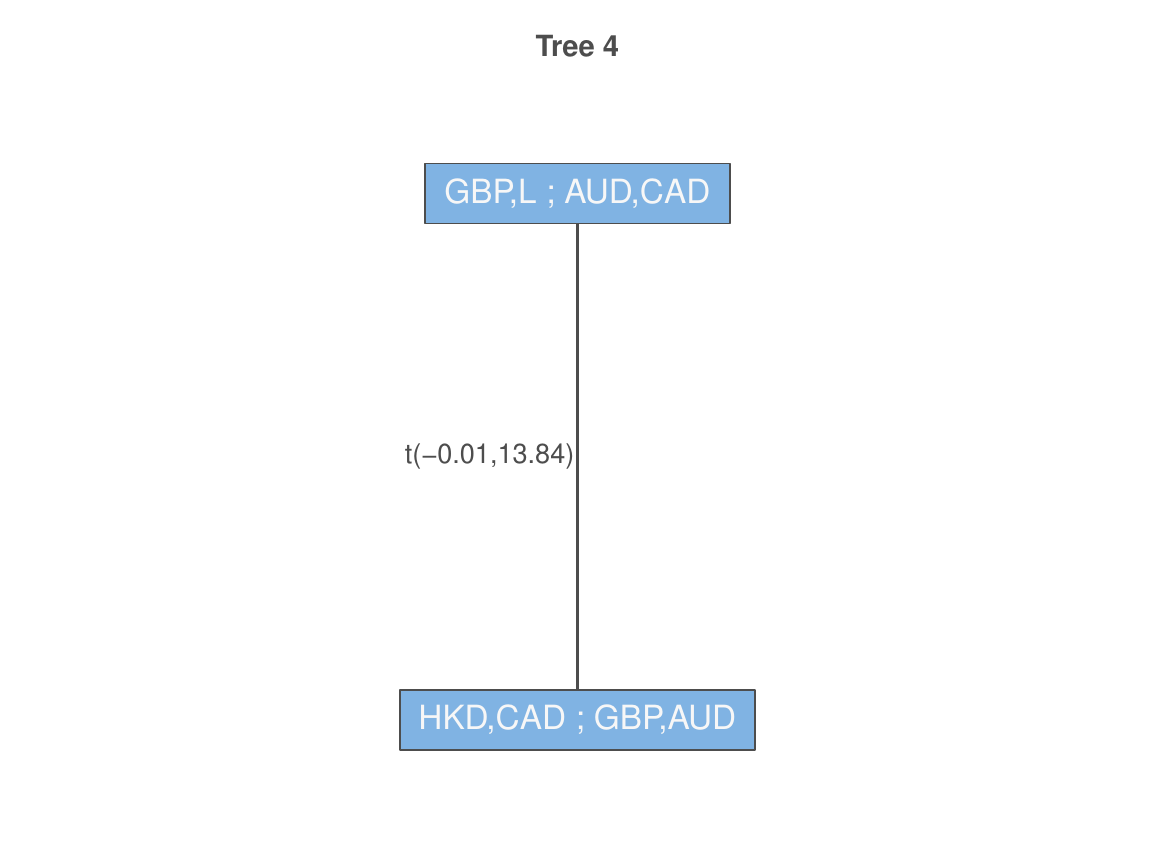}}
\caption{The four trees in the R-vine selected for Portfolio A data.}
\label{fig:trees_A}
\end{figure}

\subsection{Scatter plots of normal scores for Portfolio C}
\label{appen::sc}

\begin{figure}[H]
\centering
\subfigure{\includegraphics[scale=0.22]{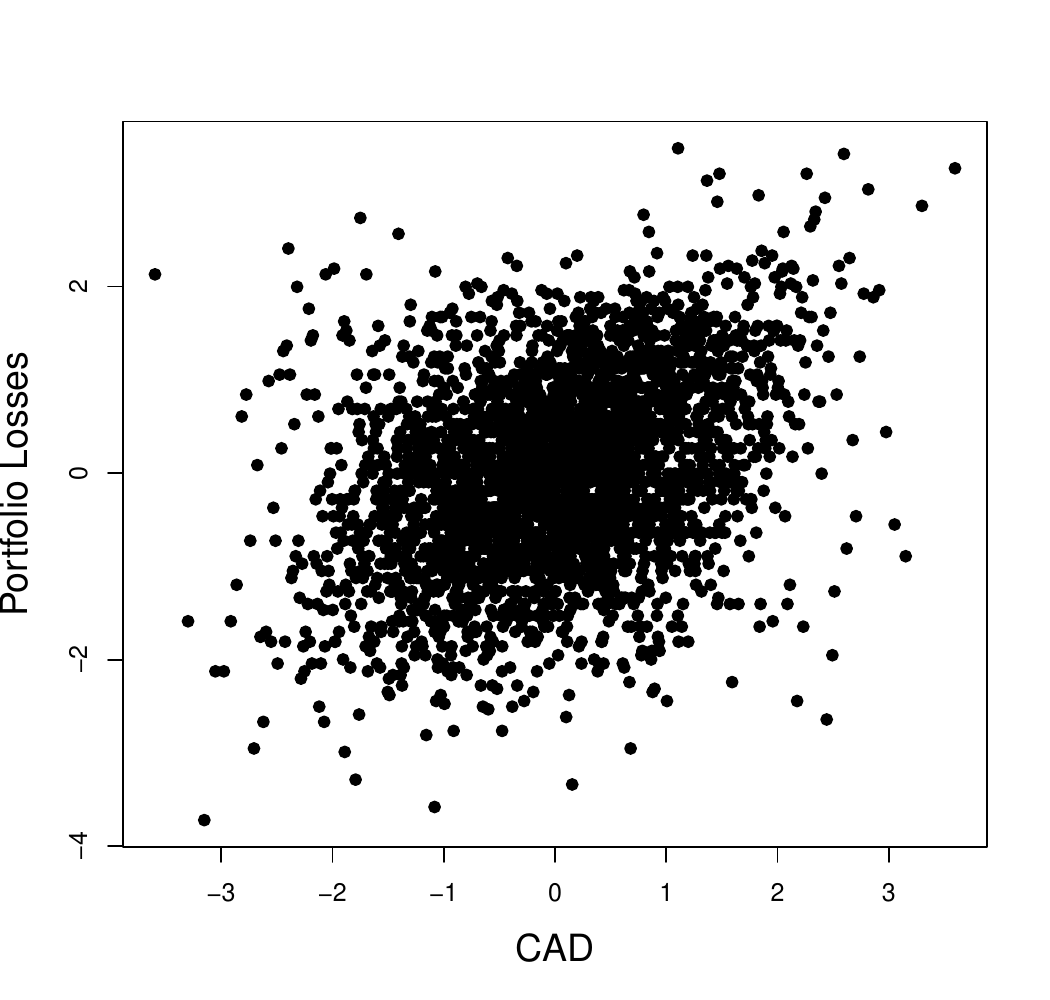}}
\subfigure{\includegraphics[scale=0.22]{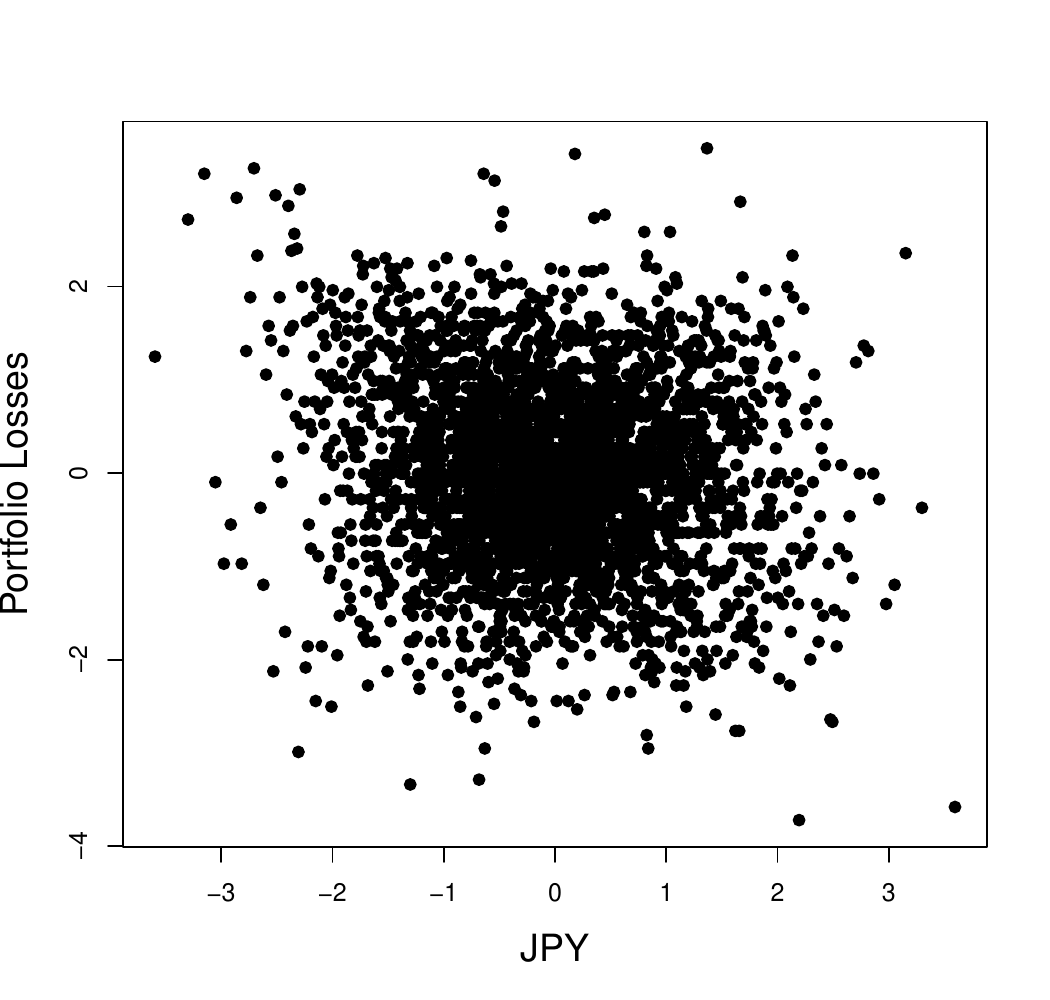}}
\subfigure{\includegraphics[scale=0.22]{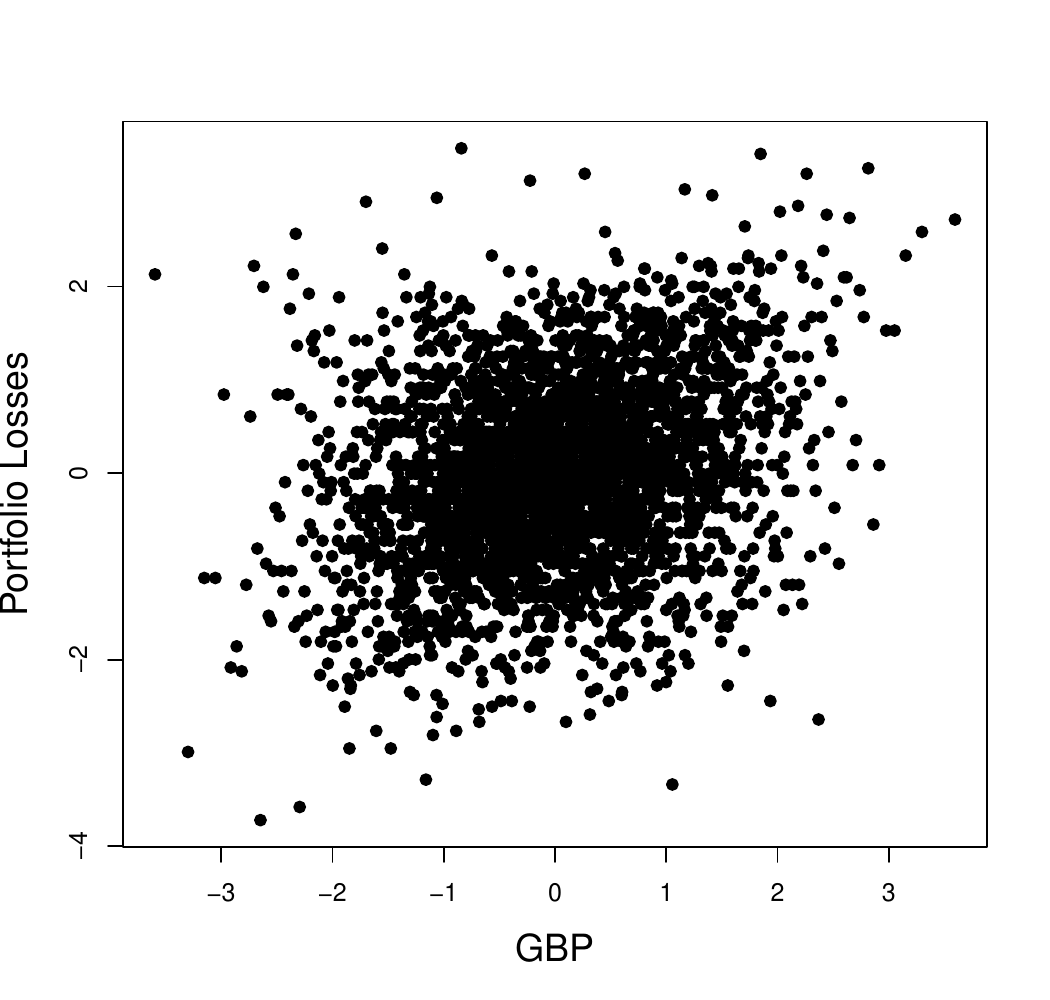}}
\subfigure{\includegraphics[scale=0.22]{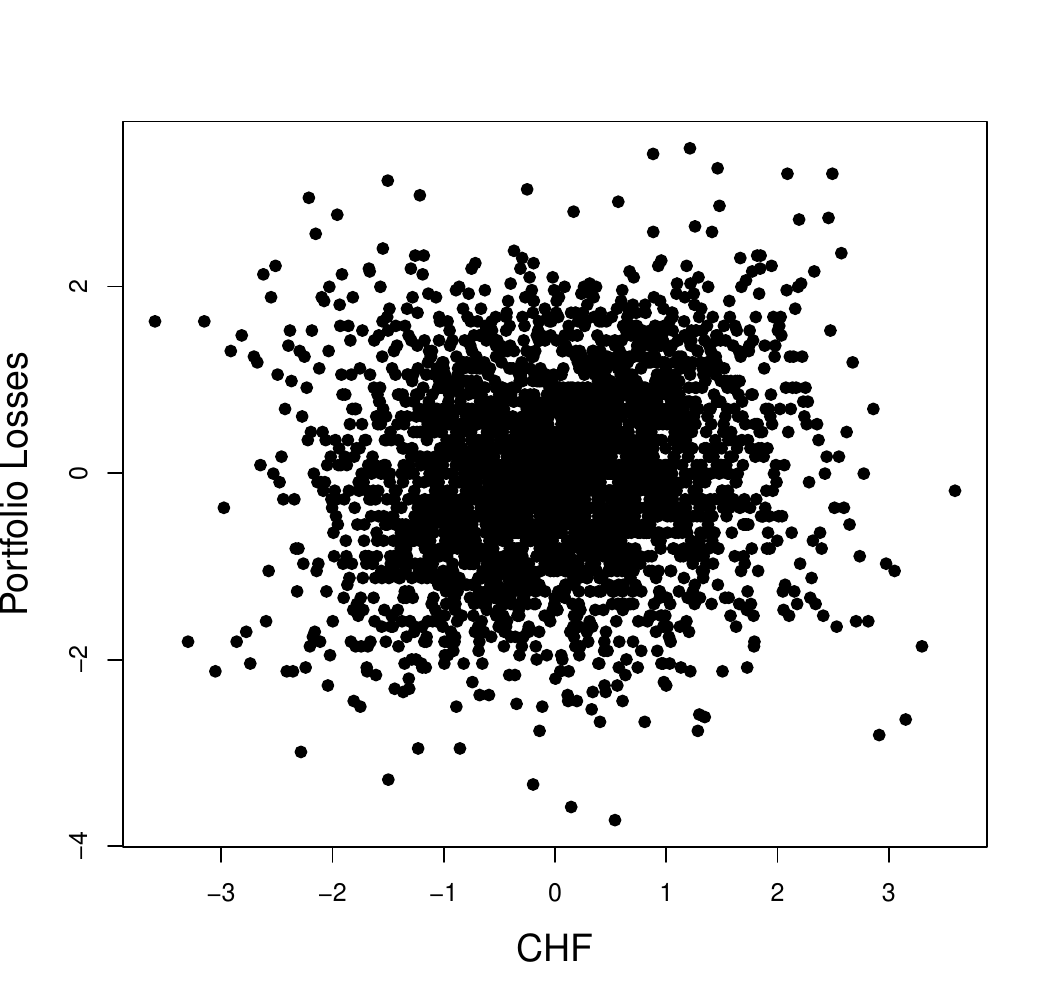}}

\subfigure{\includegraphics[scale=0.22]{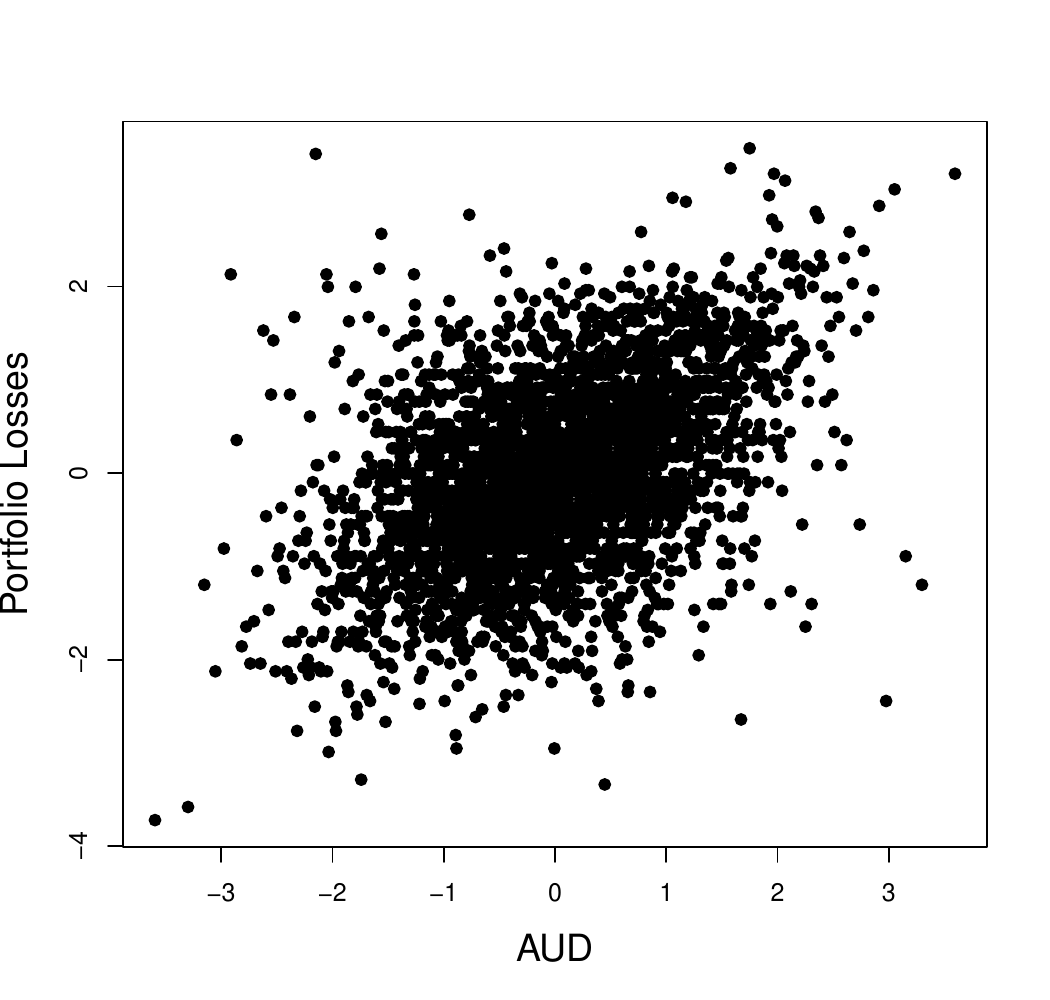}}
\subfigure{\includegraphics[scale=0.22]{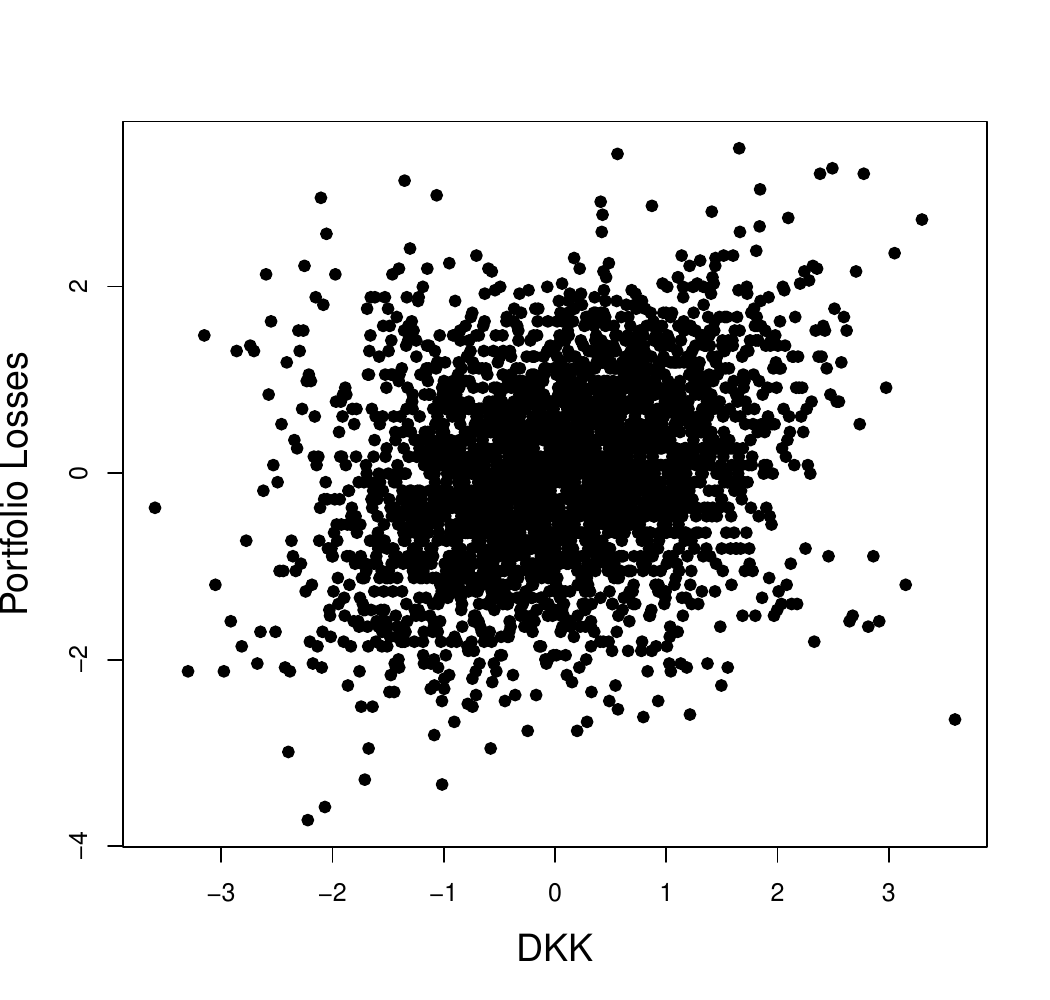}}
\subfigure{\includegraphics[scale=0.22]{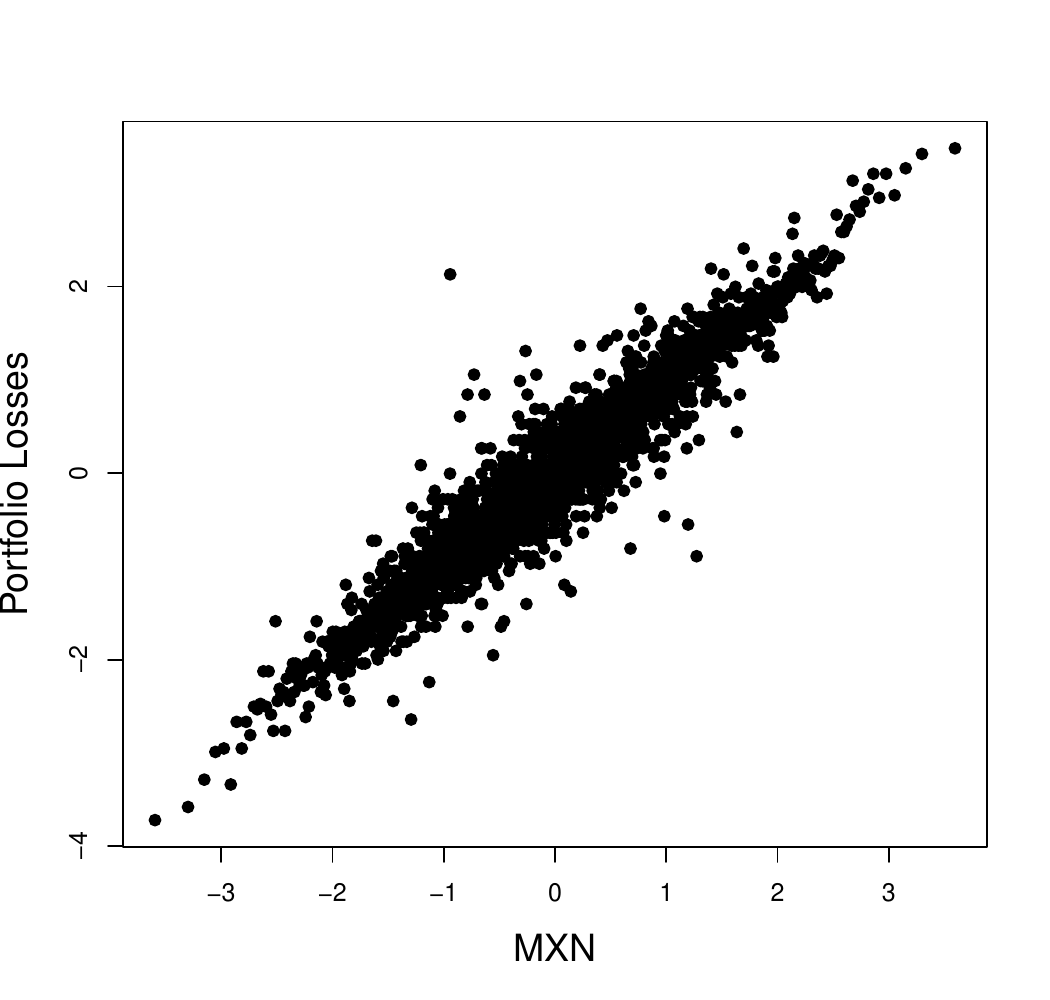}}
\subfigure{\includegraphics[scale=0.22]{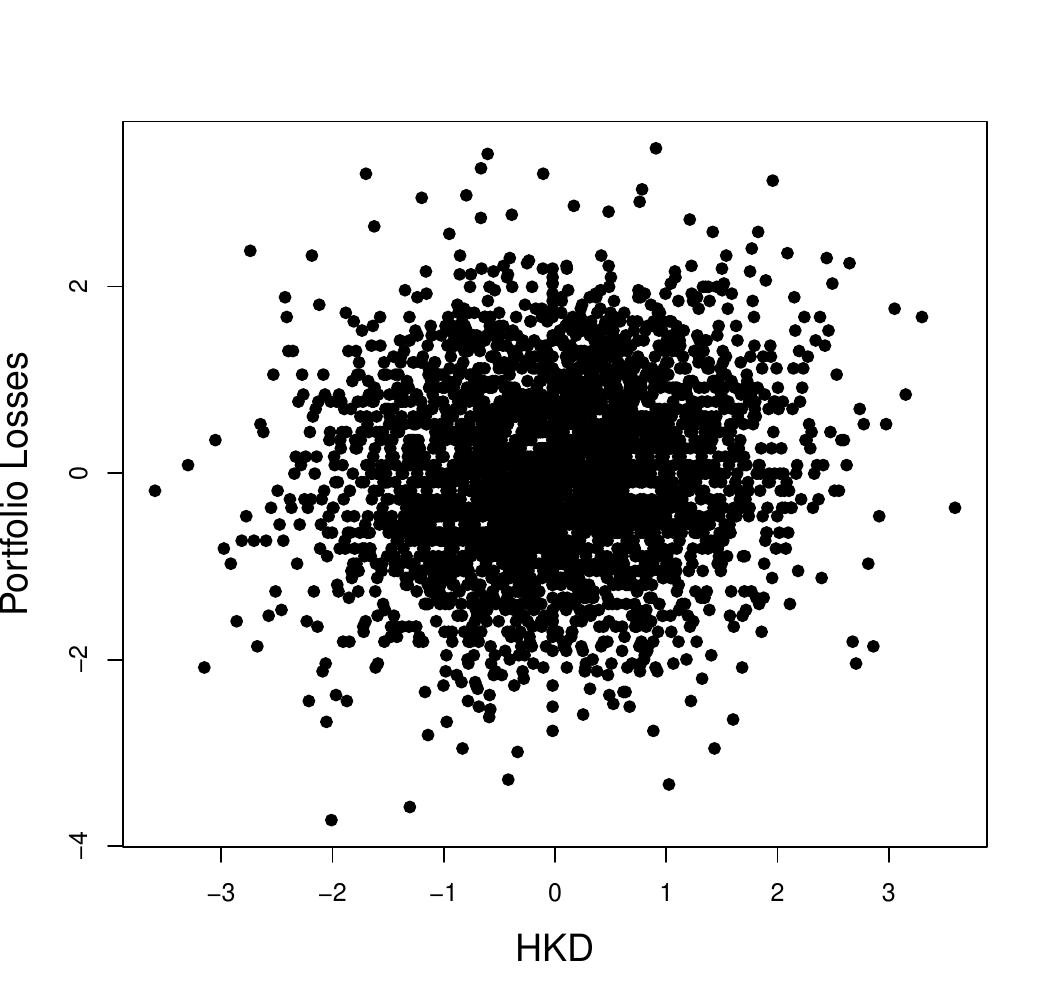}}
\end{figure}
\begin{figure}[H]
\centering
\subfigure{\includegraphics[scale=0.22]{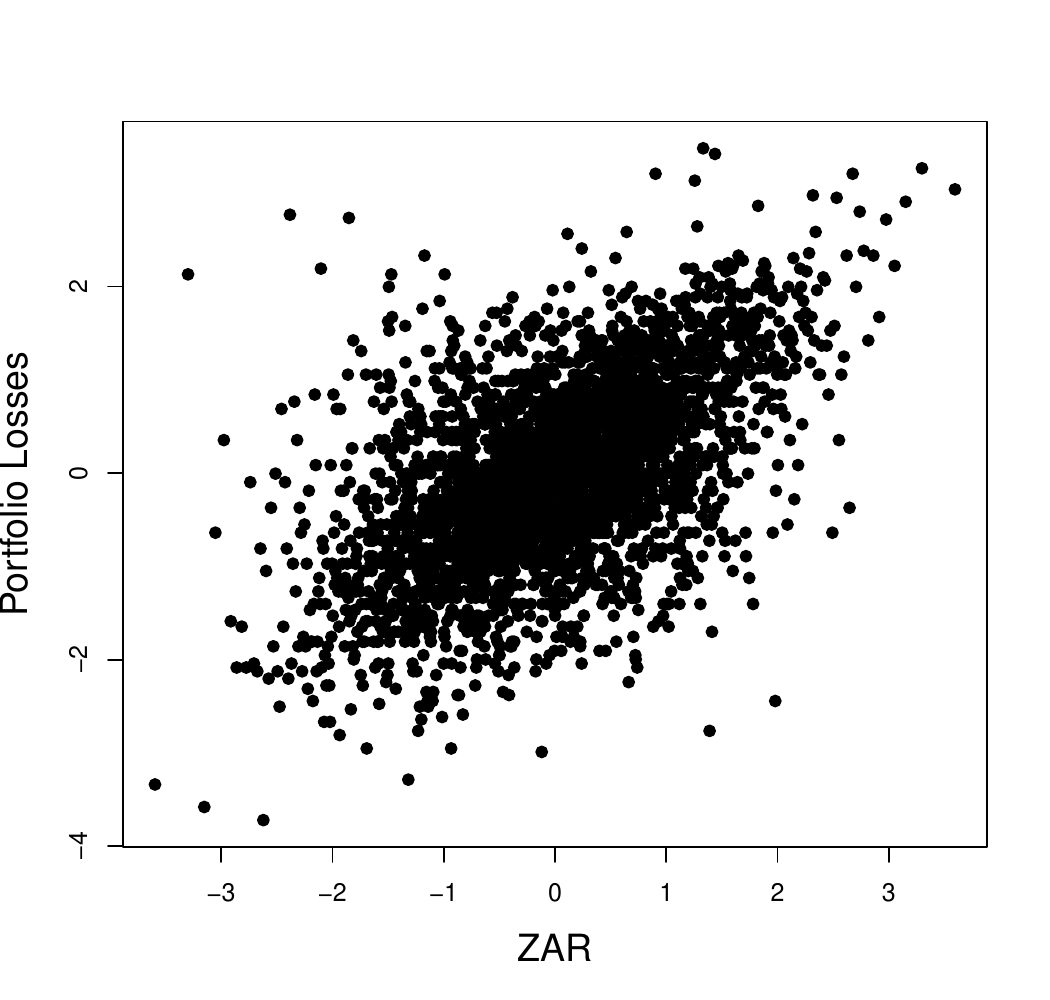}}
\subfigure{\includegraphics[scale=0.22]{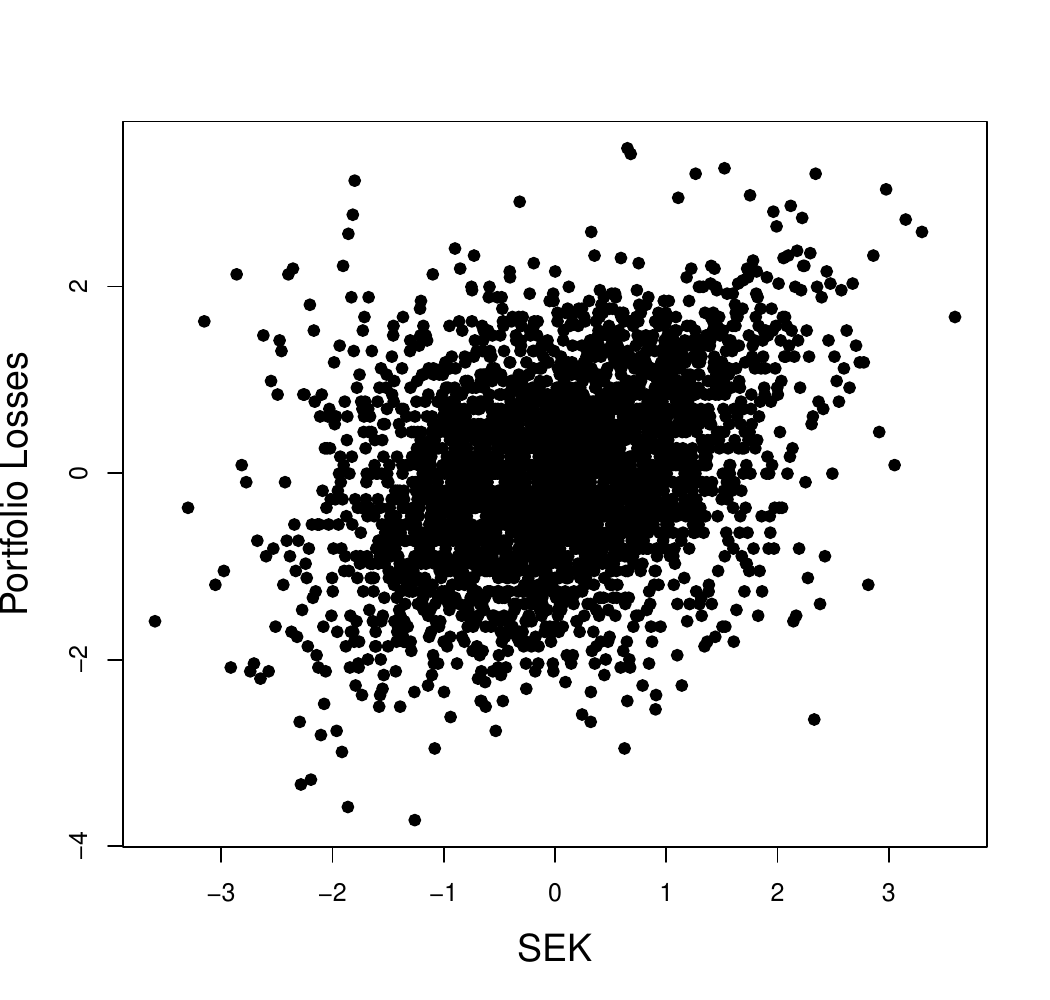}}
\subfigure{\includegraphics[scale=0.22]{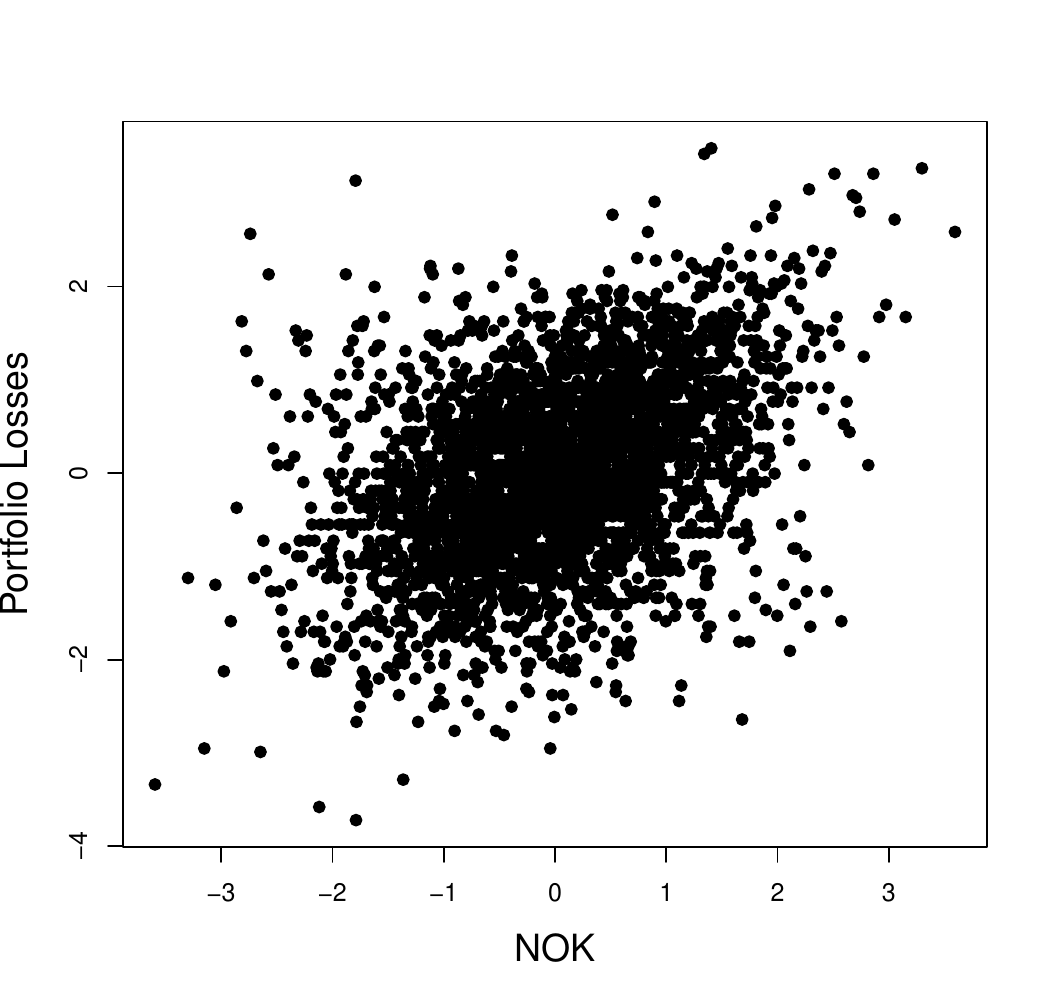}}
\subfigure{\includegraphics[scale=0.22]{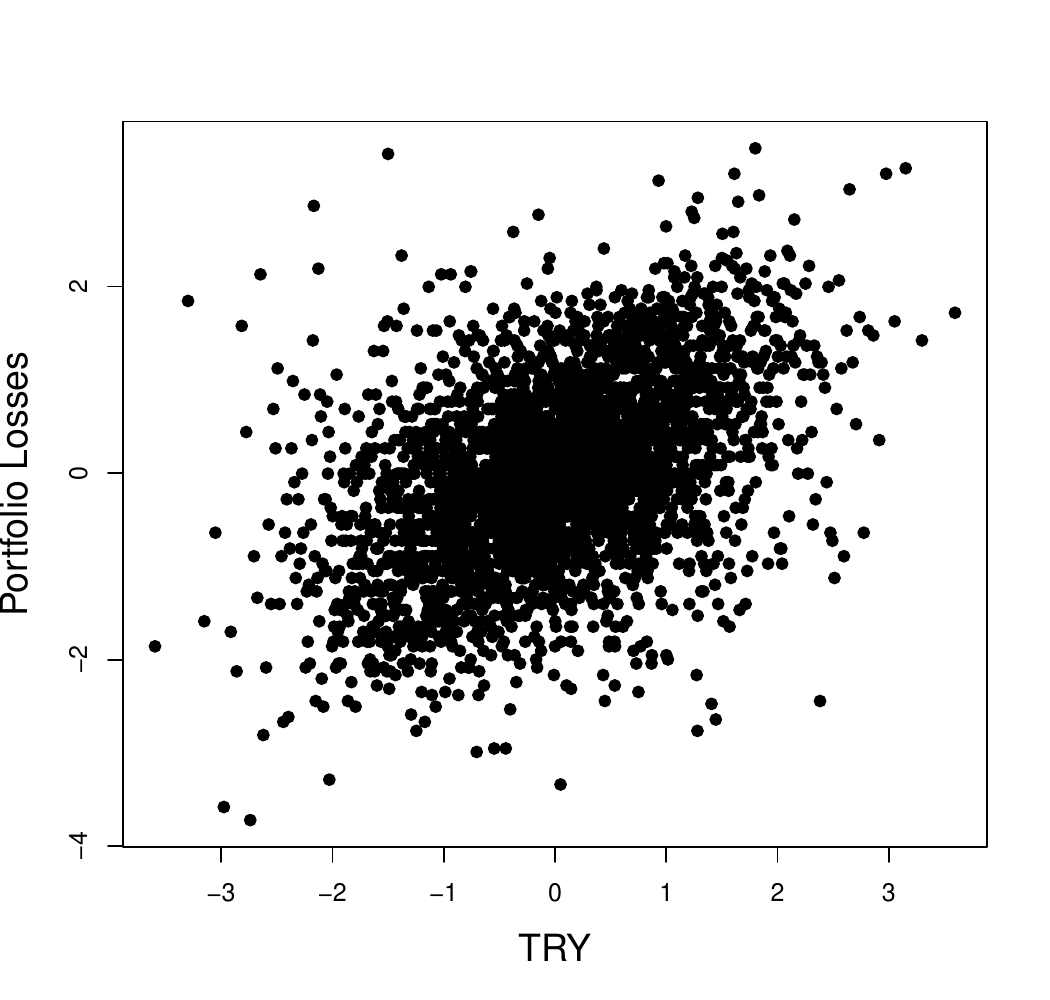}}
\subfigure{\includegraphics[scale=0.22]{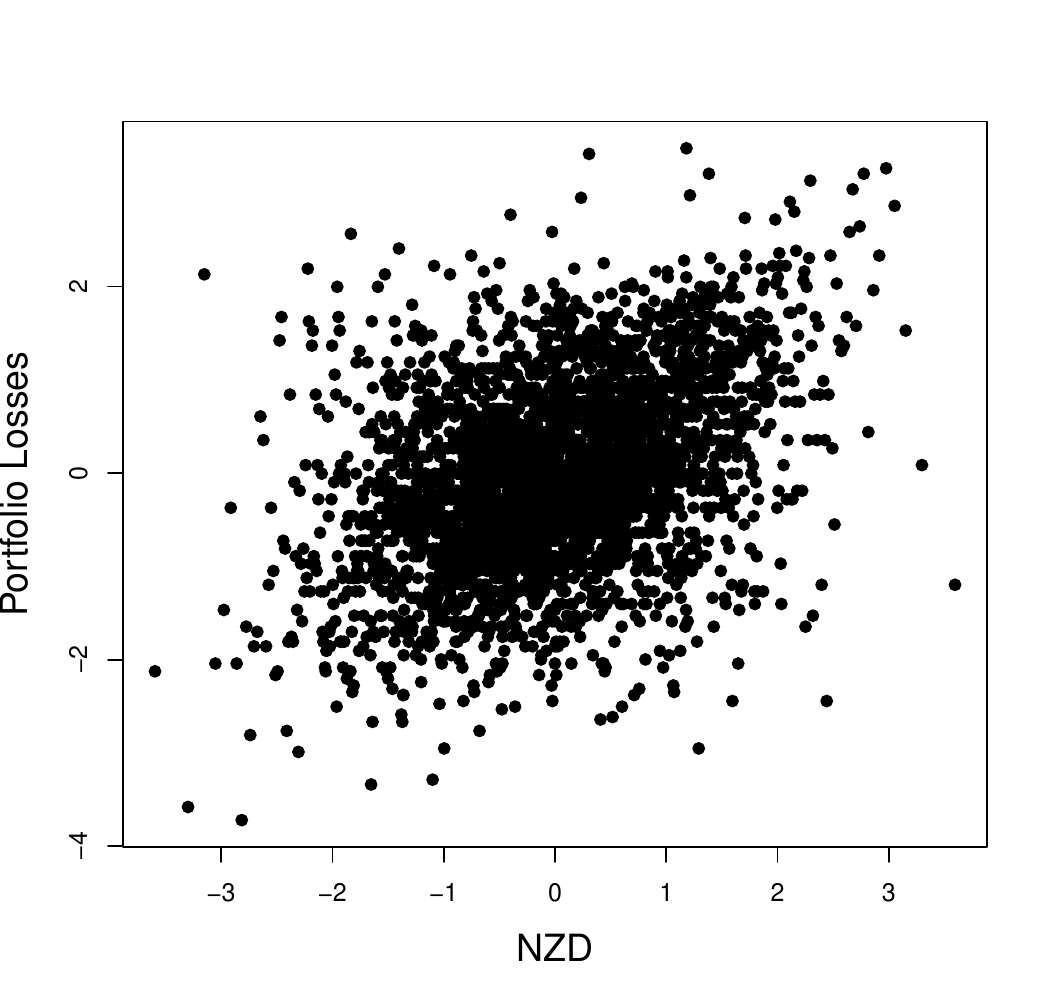}}
\subfigure{\includegraphics[scale=0.22]{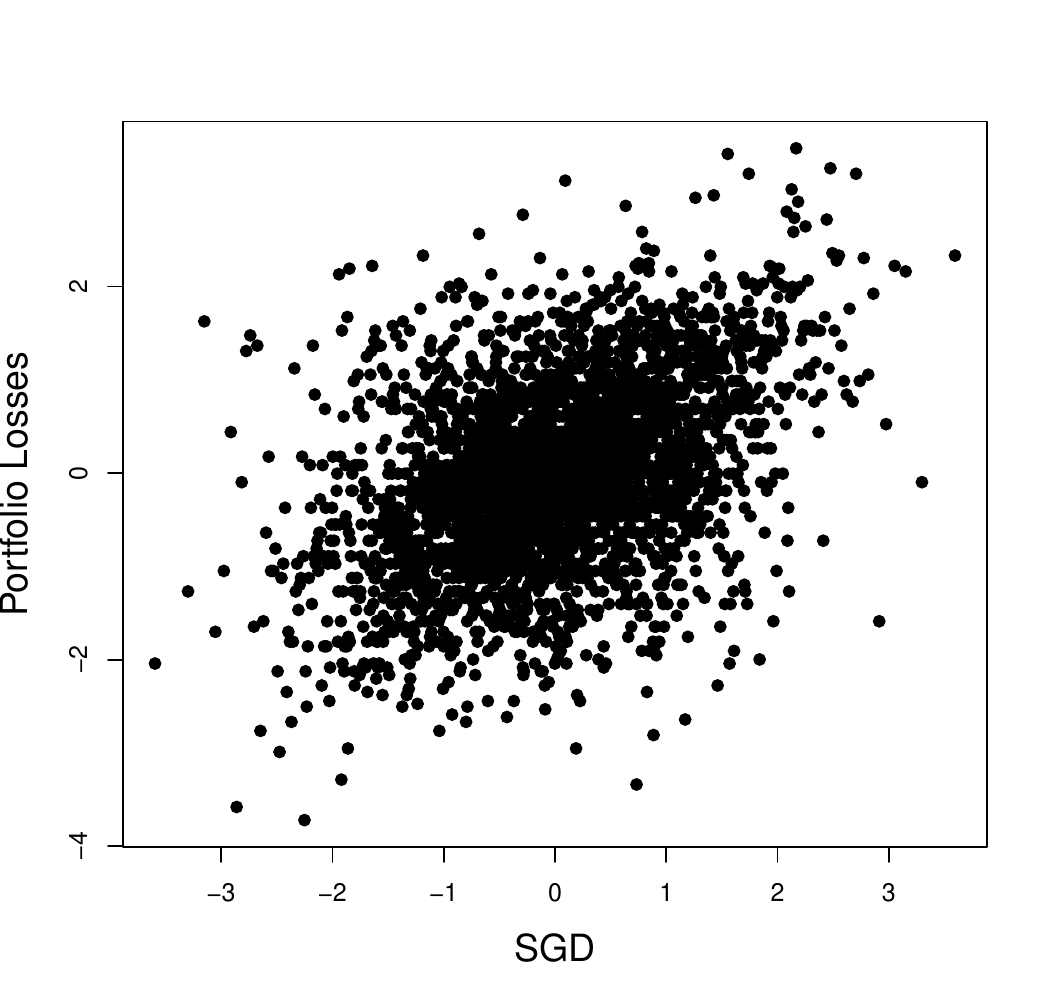}}
\subfigure{\includegraphics[scale=0.22]{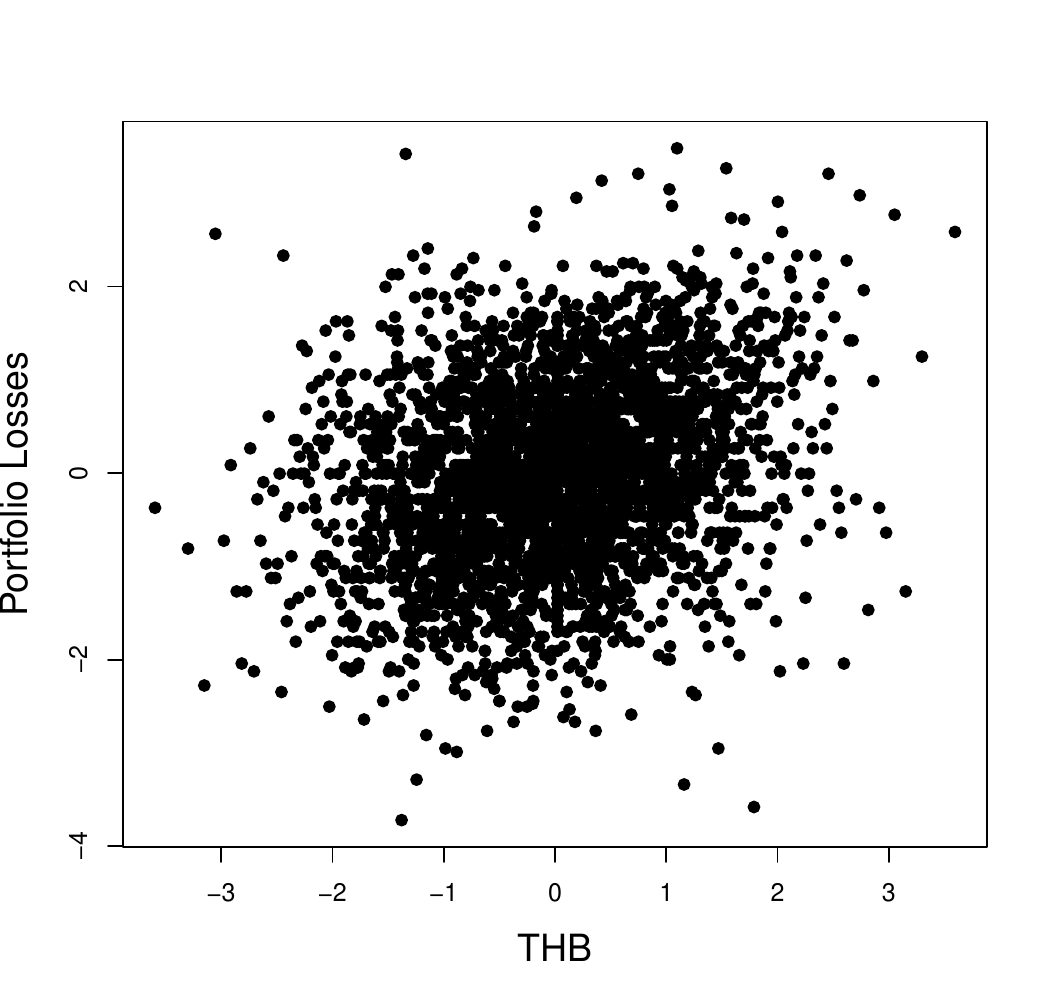}}
\subfigure{\includegraphics[scale=0.22]{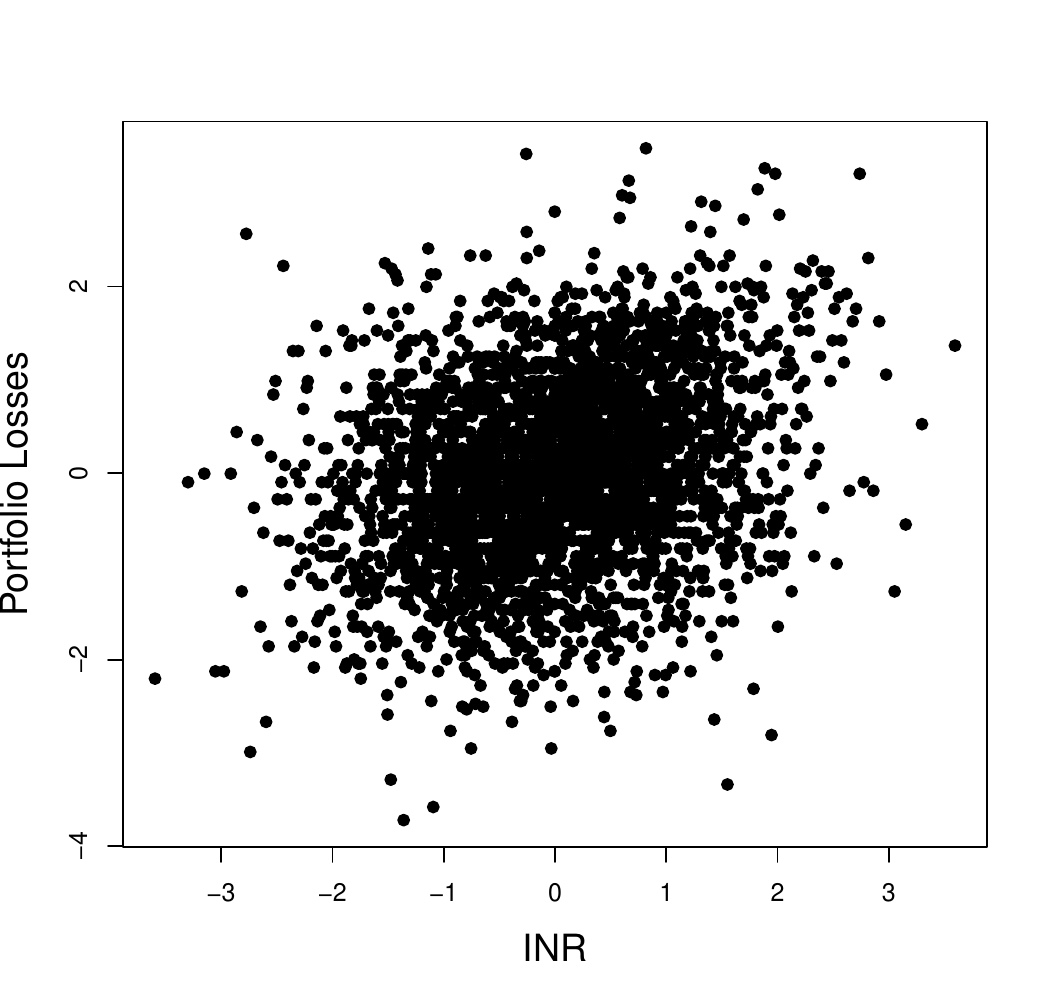}}
\subfigure{\includegraphics[scale=0.22]{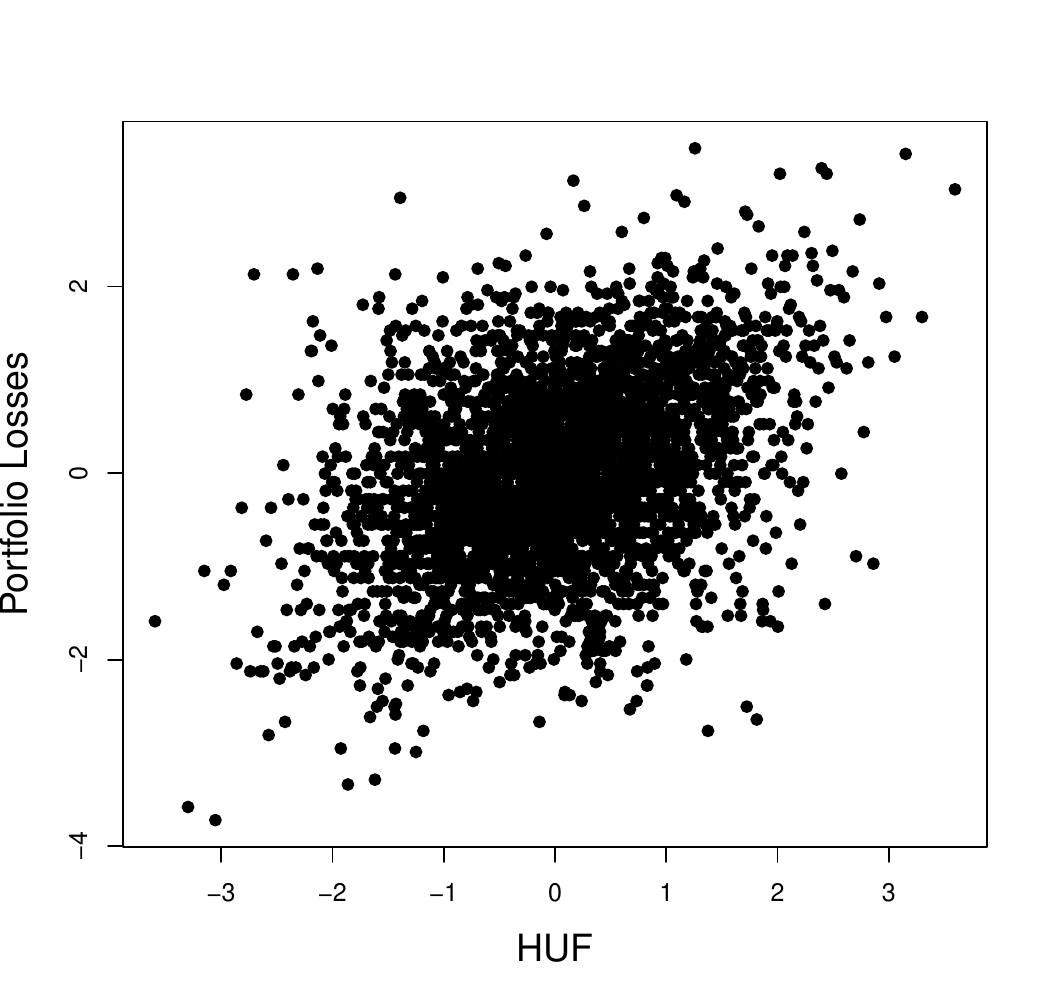}}
\subfigure{\includegraphics[scale=0.22]{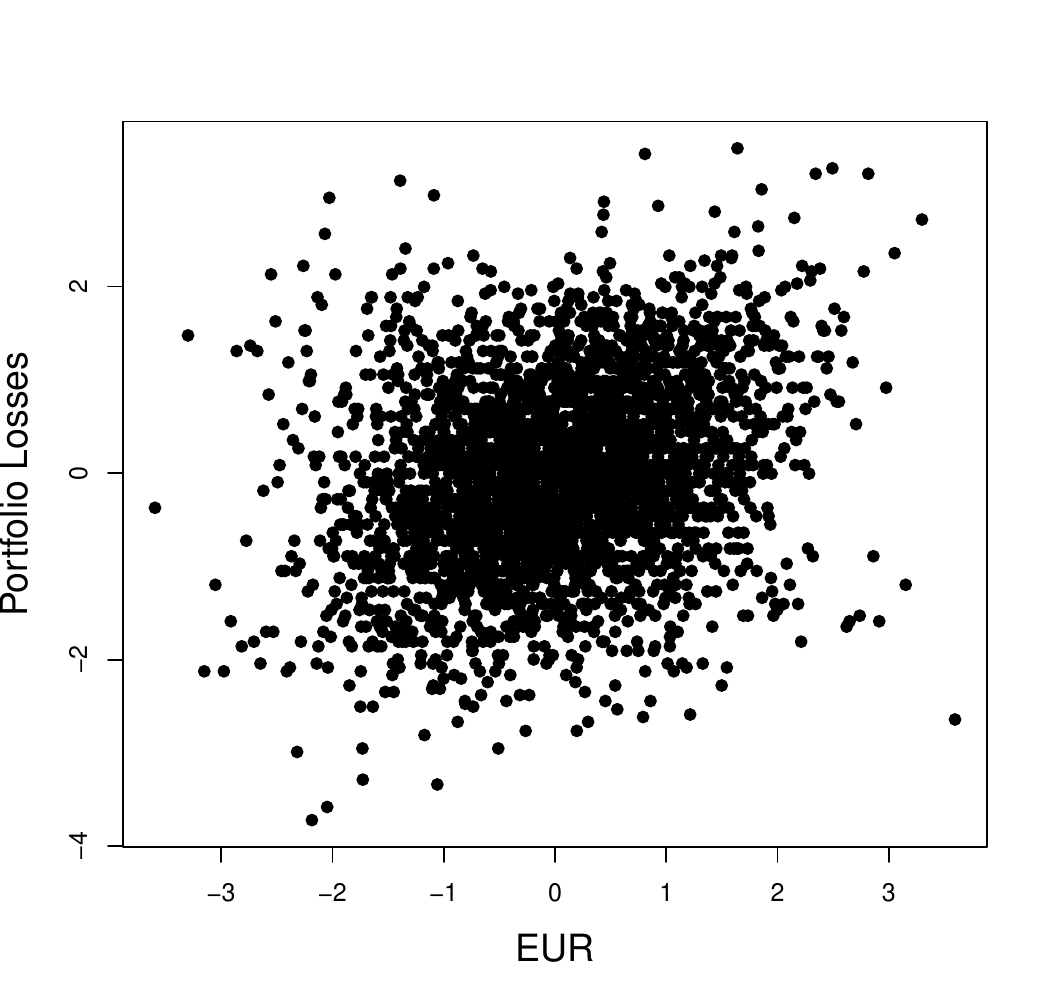}}
\caption{Scatter plots of normal scores between portfolio losses and relative changes in currency exchange rates for Portfolio C.}
\end{figure}

\subsection{Plot of estimates of stress scenarios for Portfolio B}
\label{appen::B_CI_s}
\begin{figure}[H]
\centering
\subfigure{\includegraphics[scale=0.35]{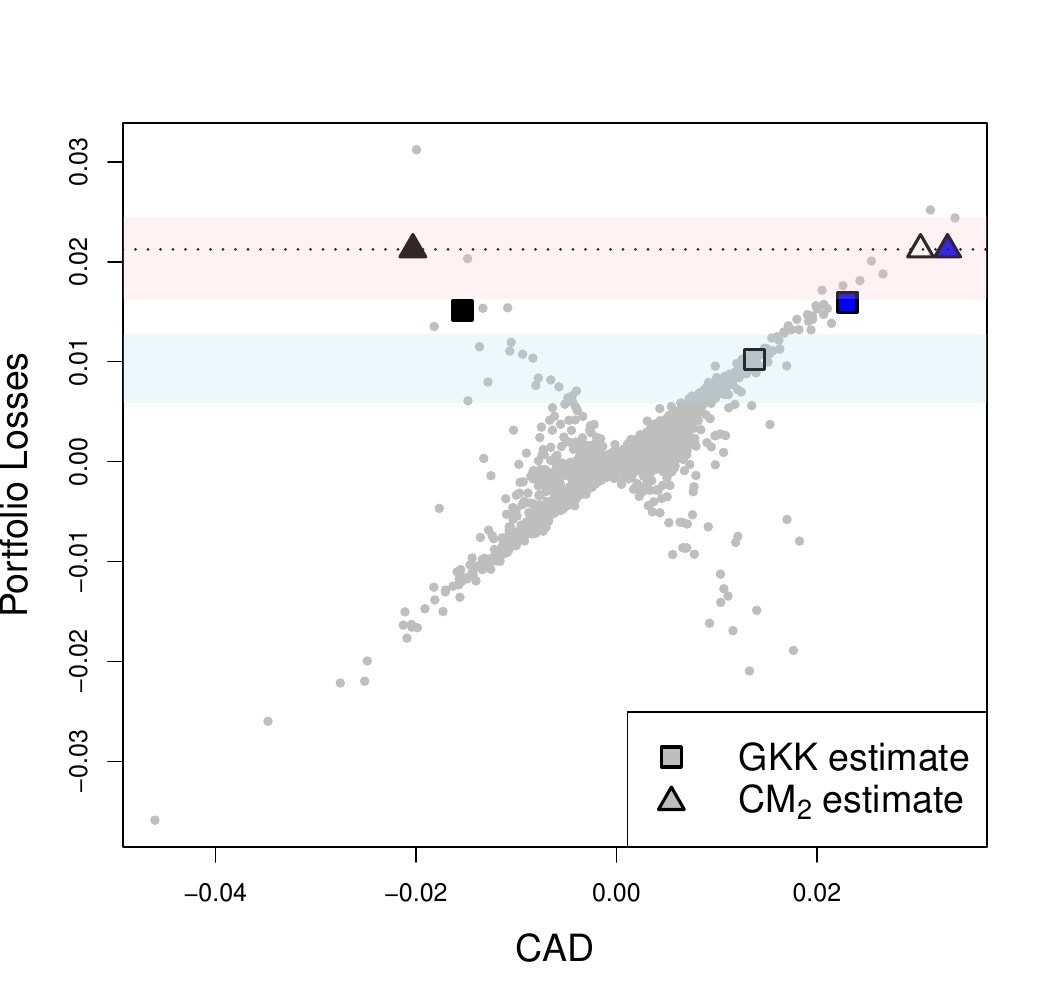}}
\subfigure{\includegraphics[scale=0.35]{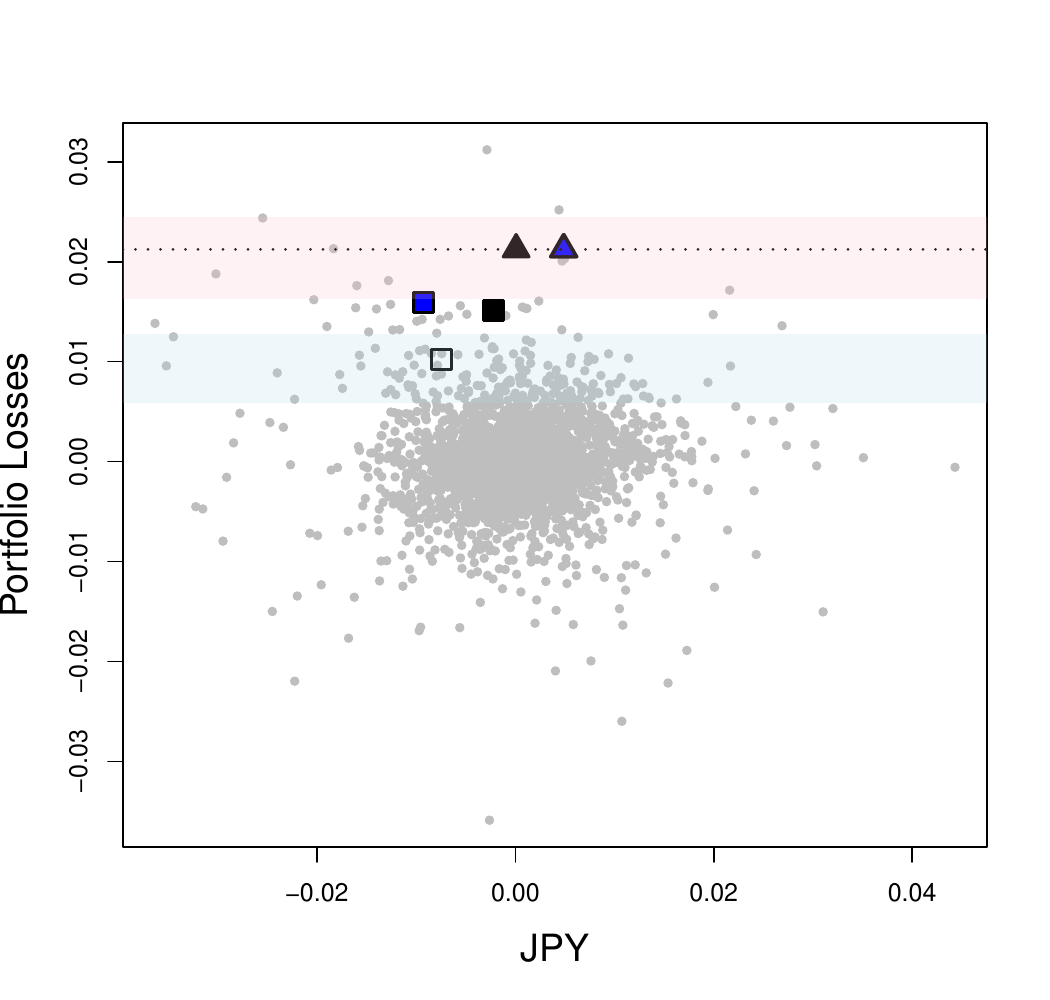}}
\end{figure}
\begin{figure}[H]
\centering
\subfigure{\includegraphics[scale=0.35]{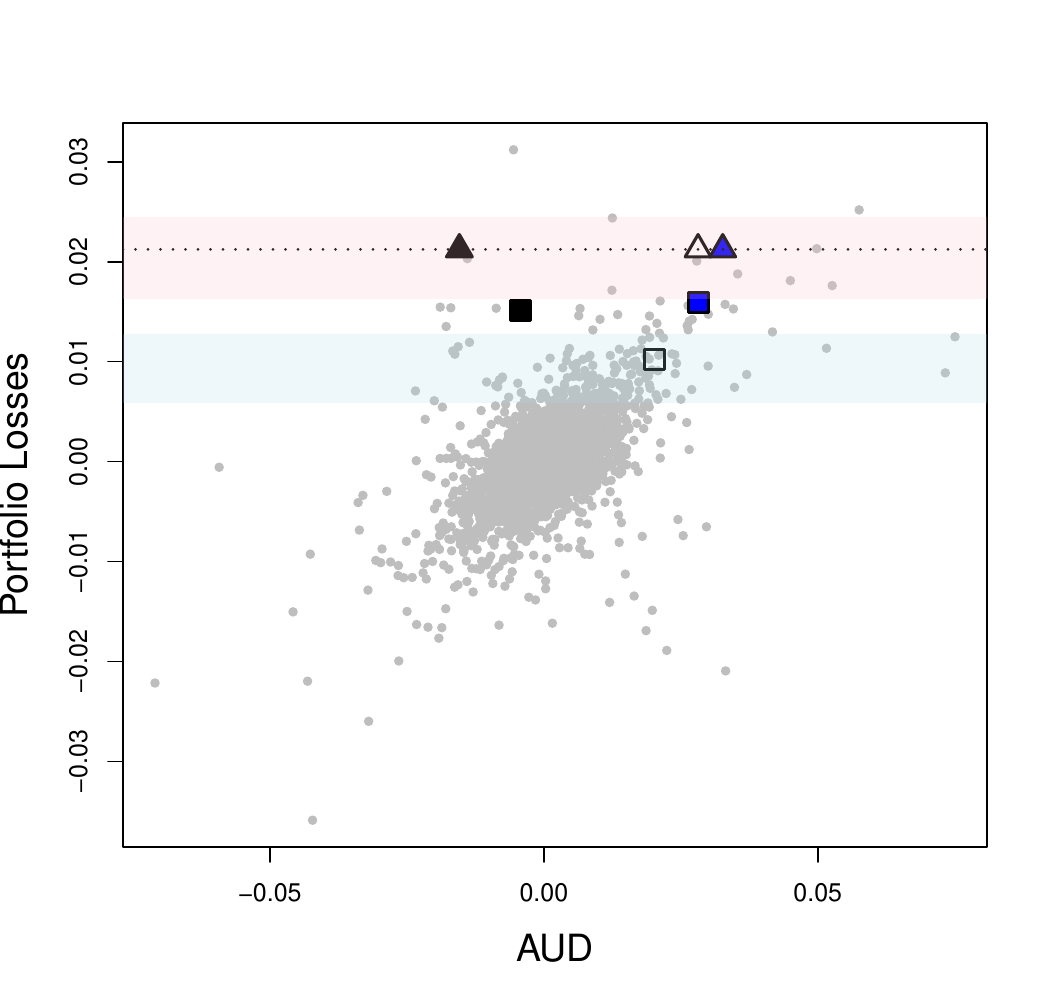}}
\subfigure{\includegraphics[scale=0.35]{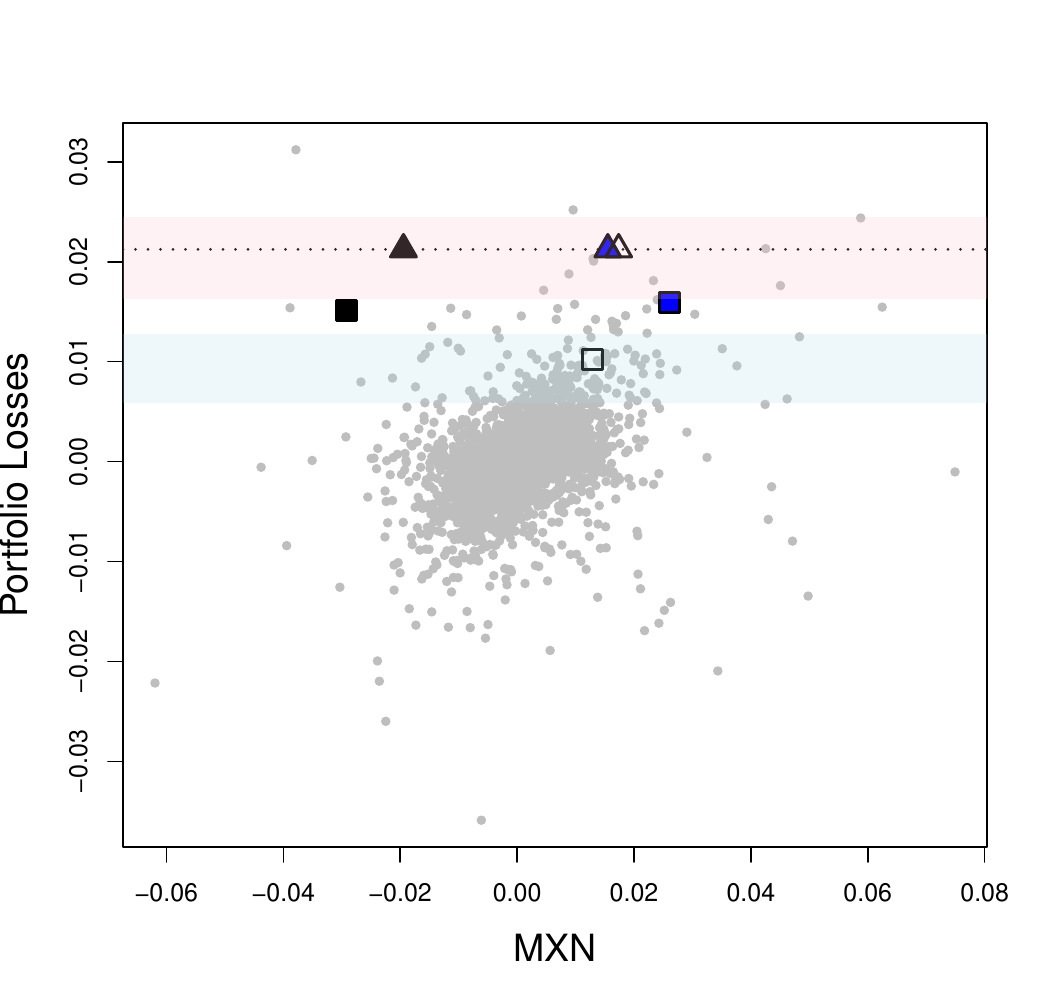}}
\subfigure{\includegraphics[scale=0.35]{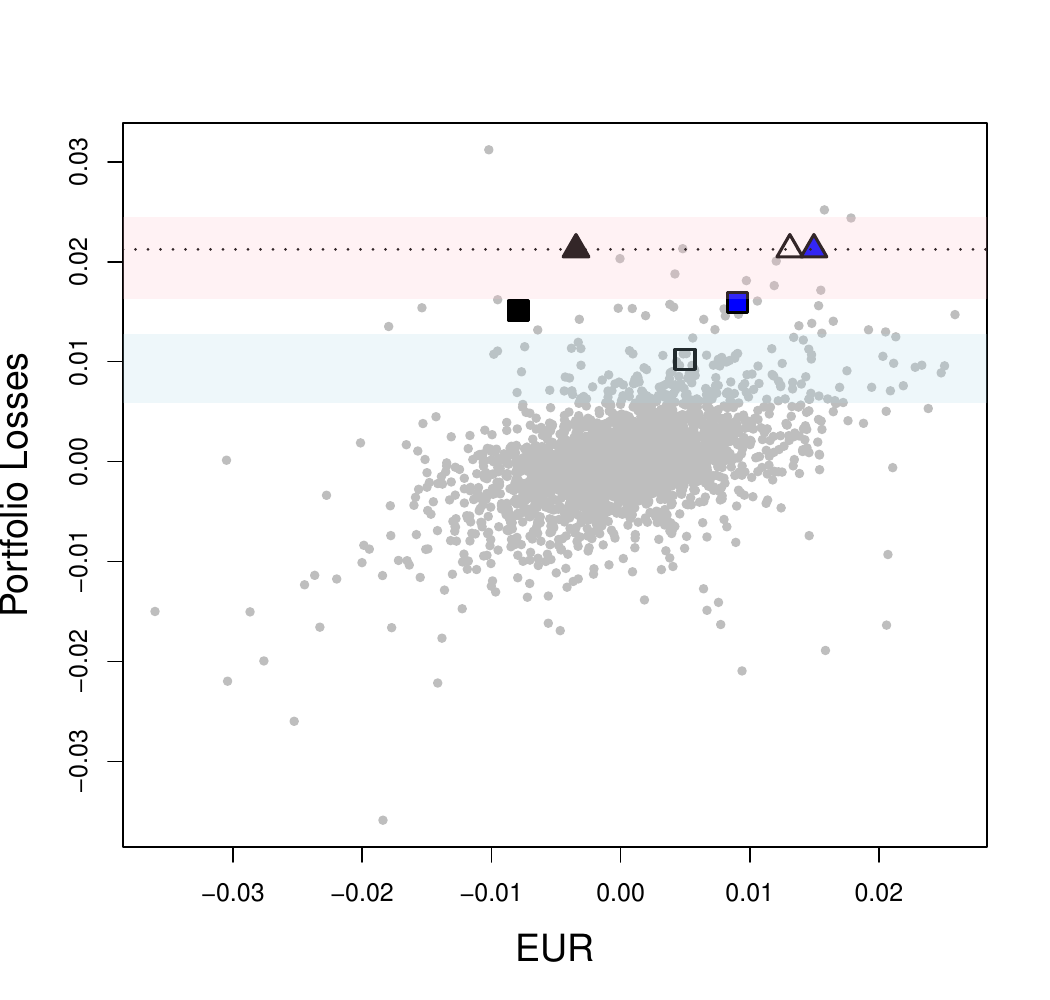}}
\caption{Plots of stress scenario estimates together with the data points for Portfolio B. The horizontal dotted lines indicate the threshold $\ell = 0.0212$. The blue, black and unfilled symbols show, respectively, stress scenario estimates for cluster 1, cluster 2 and all of the data combined. The pink and blue shaded bands show the $95\%$ confidence intervals for predicted portfolio losses for the $\CM_2$ and GKK stress scenario estimates with all observations included.}
\label{fig:B_CI_s}
\end{figure}

\subsection{Three-dimensional clouds for Portfolio B}
\label{appen::B_3d}

\begin{figure}[H]
\centering
\subfigure{\includegraphics[scale=0.5]{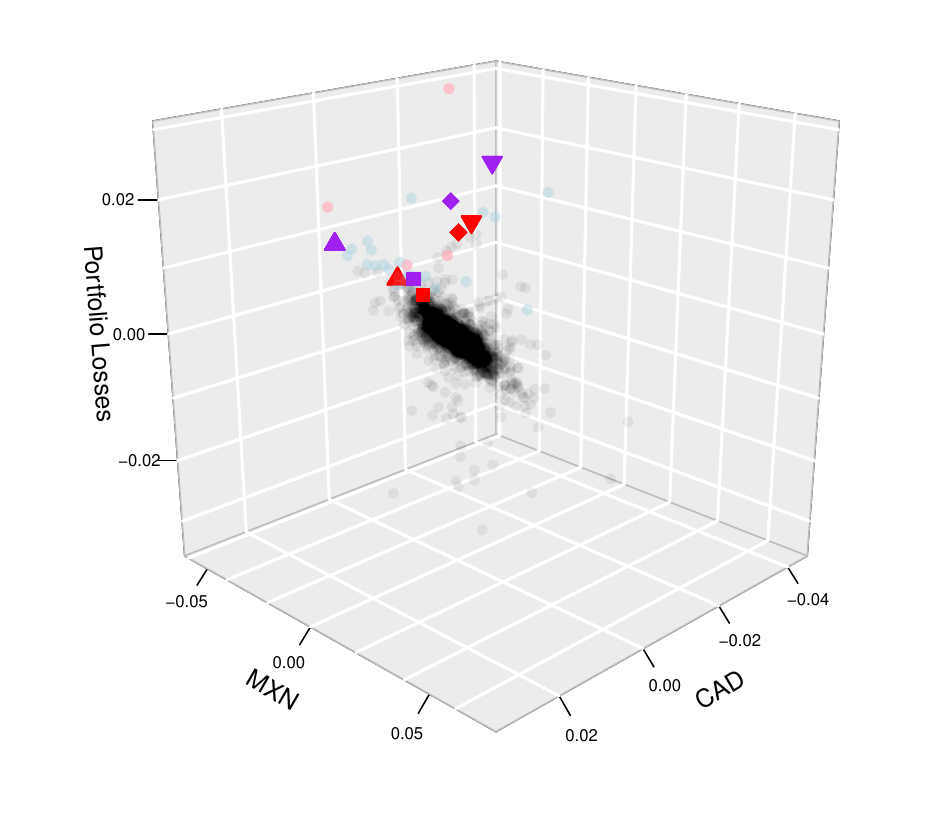}}
\subfigure{\includegraphics[scale=0.5]{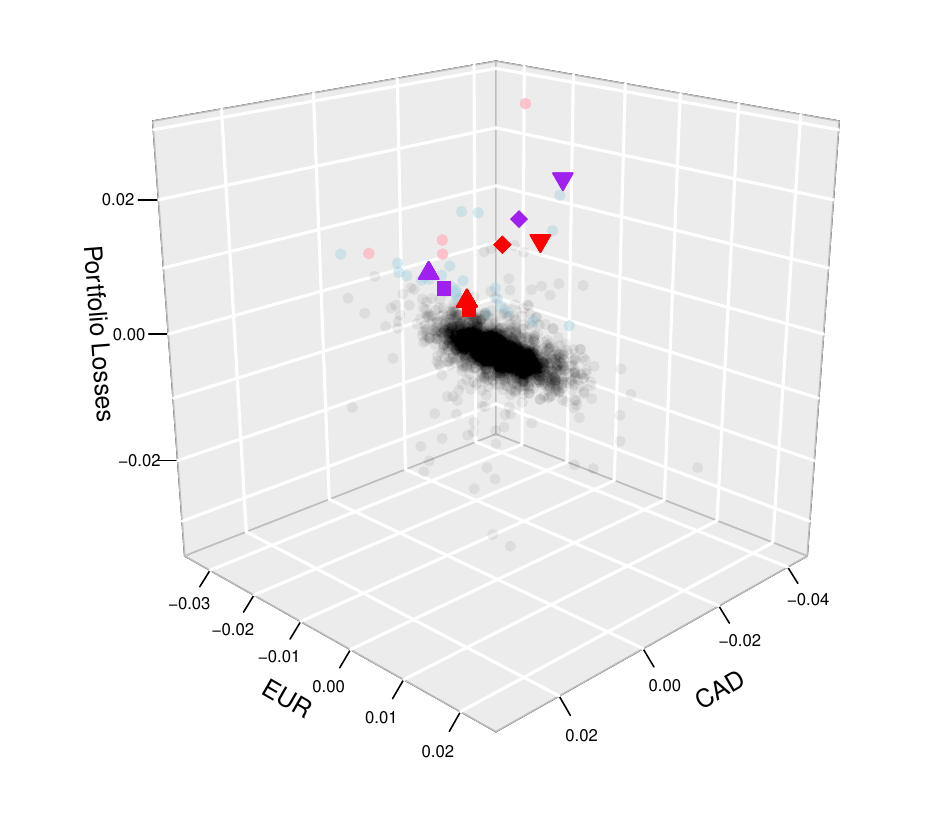}}
\end{figure}
\begin{figure}[H]
\centering
\subfigure{\includegraphics[scale=0.5]{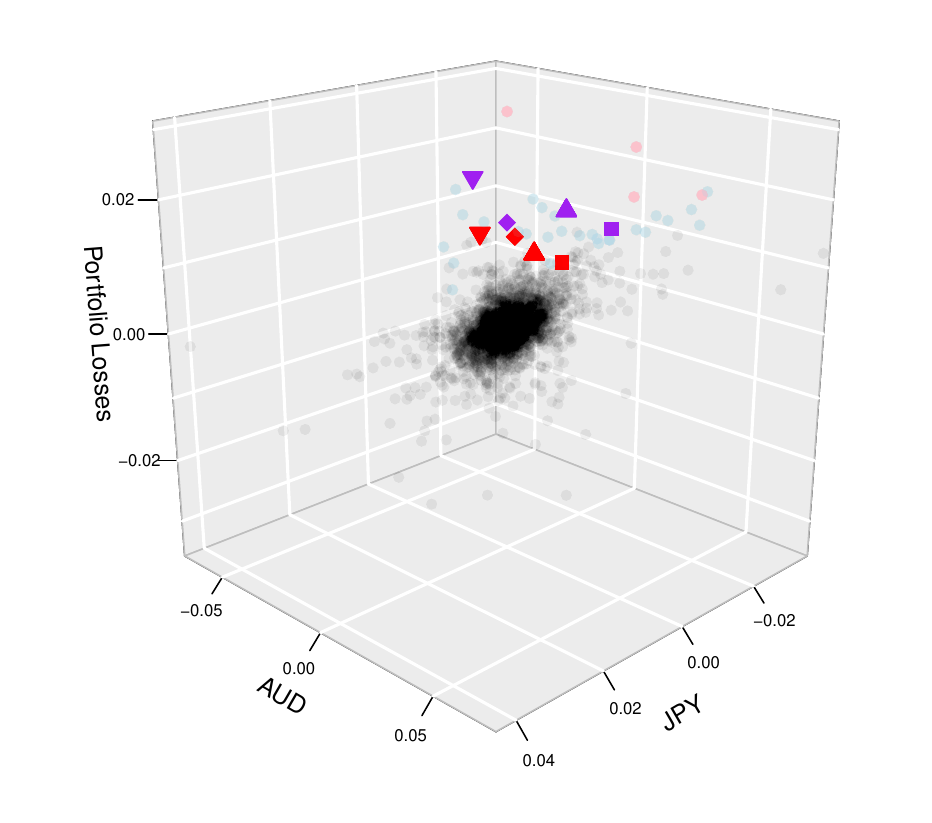}}
\subfigure{\includegraphics[scale=0.5]{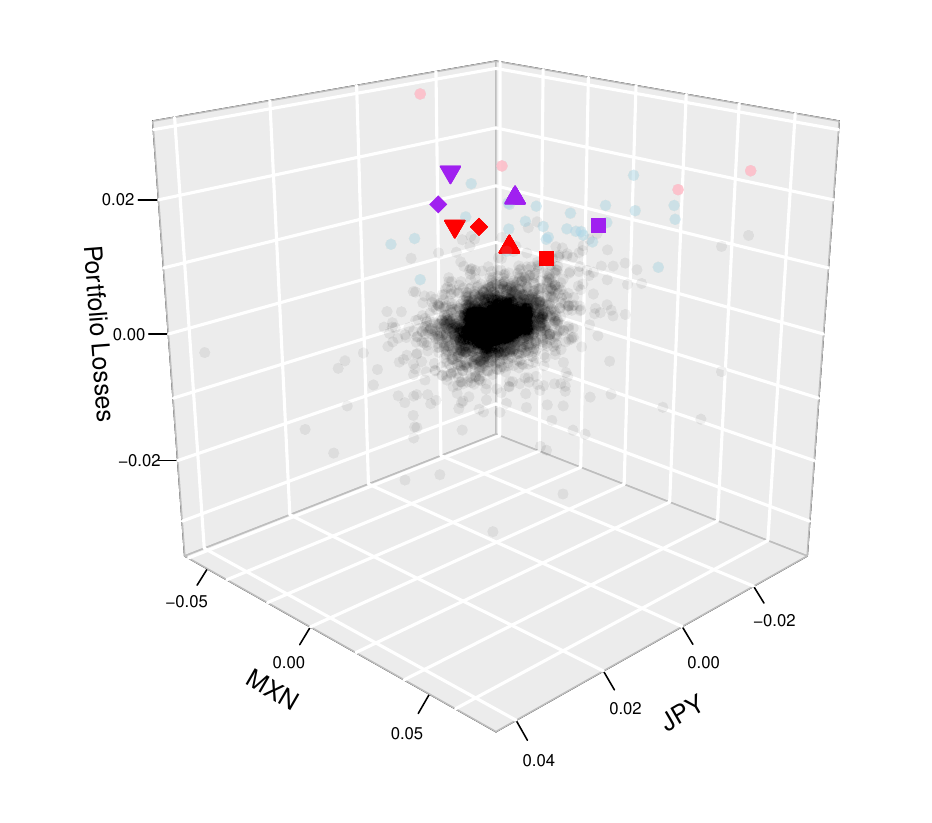}}
\subfigure{\includegraphics[scale=0.5]{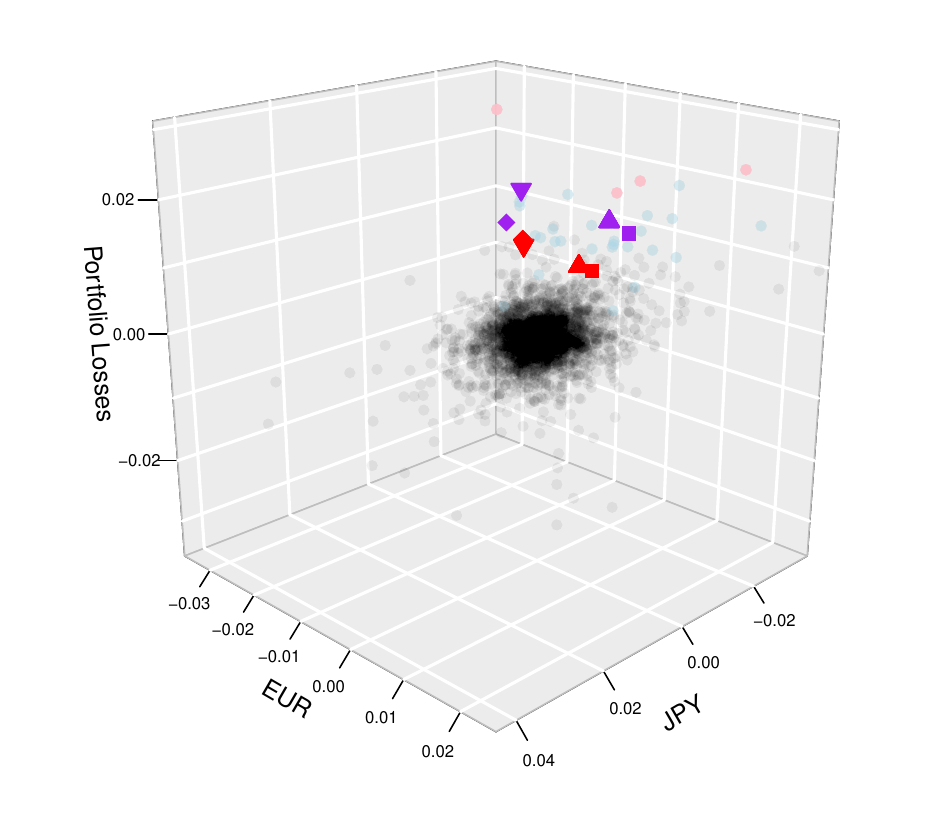}}
\subfigure{\includegraphics[scale=0.5]{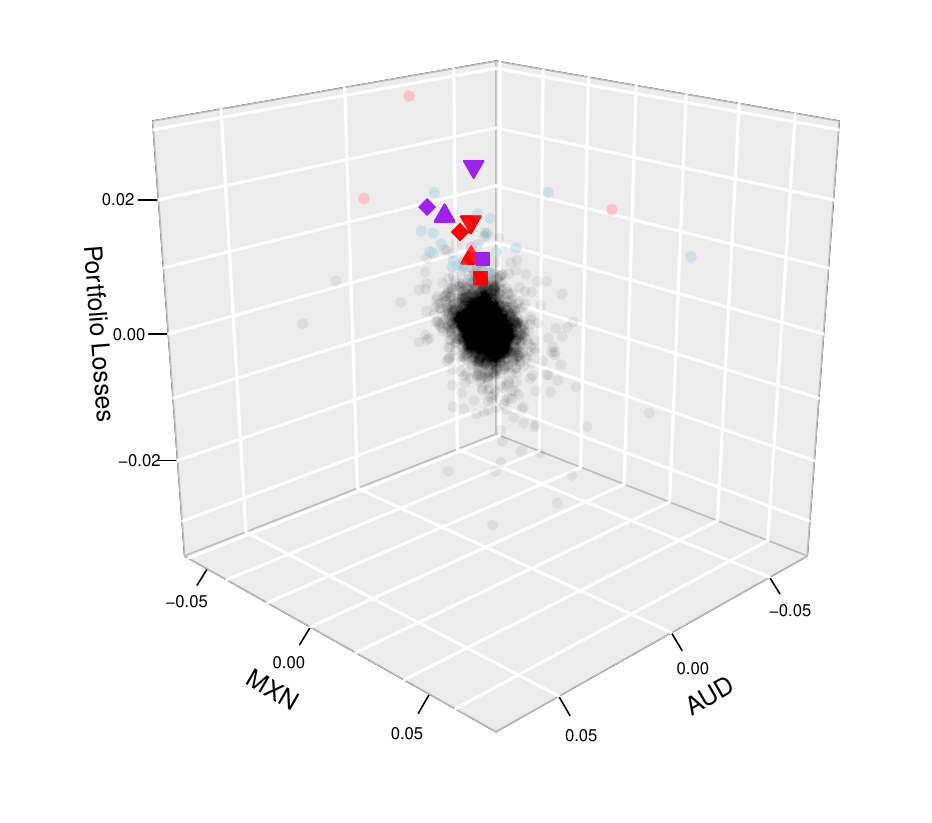}}
\subfigure{\includegraphics[scale=0.5]{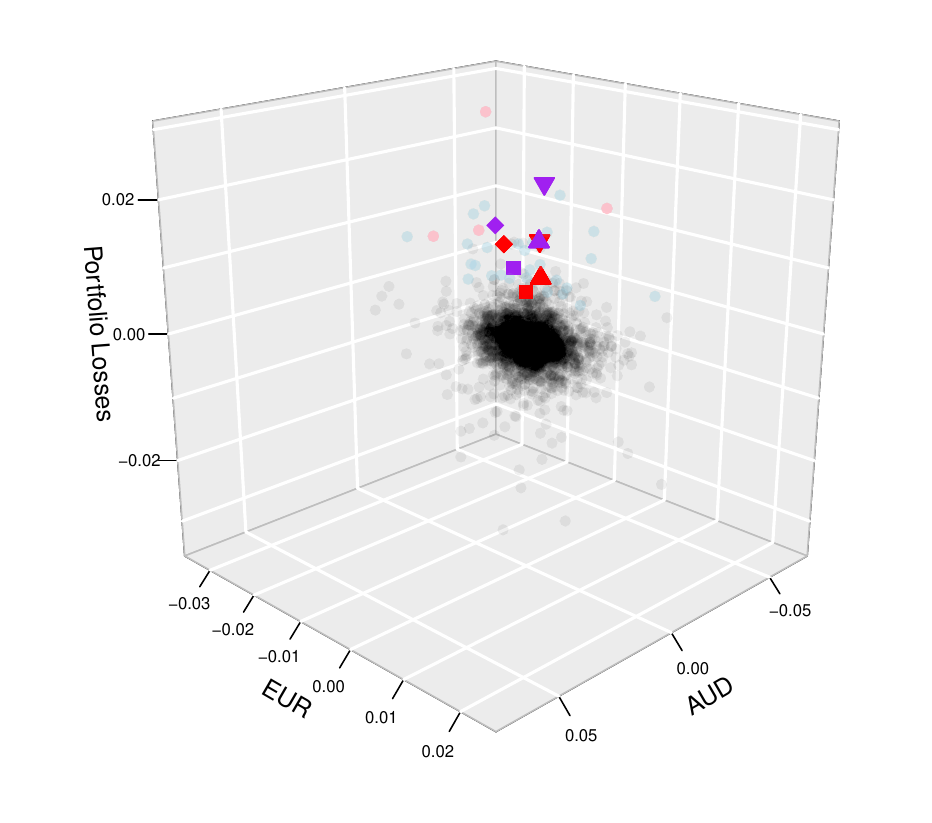}}
\subfigure{\includegraphics[scale=0.5]{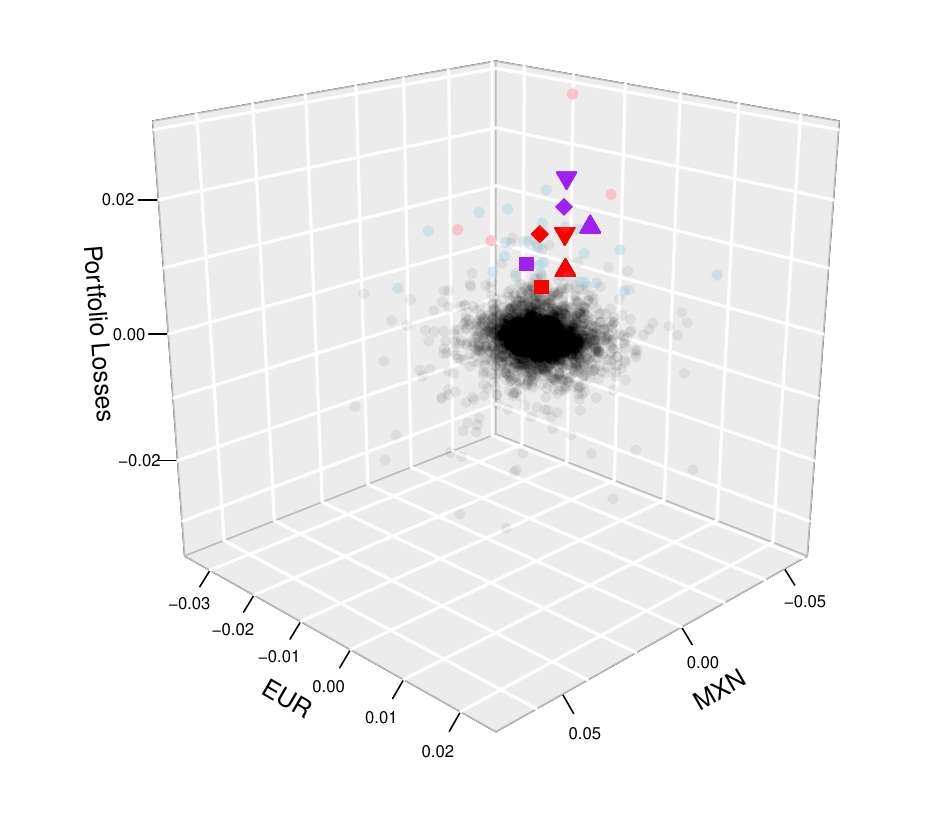}}
\caption{Plots of pairs of risk factors and associated portfolio losses for Portfolio B. The light blue and pink points indicate observations for which portfolio losses are greater than $\ell = 0.0132$ and $\ell = 0.0212$, respectively. The red and purple triangles/squares show stress scenario estimates at thresholds $\ell = 0.0132$ and $\ell = 0.0212$, respectively.}
\label{fig:T_3d_more}
\end{figure}

\end{appendices}

\end{document}